\documentclass[manuscript,nonacm]{acmart}

\usepackage{amsmath}
\usepackage{amstext}
\usepackage{fp-frame}
\usepackage[latin1]{inputenc}
\usepackage{mathpartir}
\usepackage{fancyvrb}
\usepackage{moreverb}
\usepackage{stmaryrd}
\usepackage{enumerate}
\usepackage{thmtools,thm-restate}
\usepackage{comment}
\usepackage{booktabs,array}
\usepackage{rotating}
\usepackage{xcolor}
\usepackage{wrapfig}
\usepackage{circuitikz}
\usetikzlibrary{positioning,fit,arrows.meta}

\newcommand{\vb}[1]{\verb+#1+}

\newcommand{\defeq}{\triangleq}

\def\squareforqed{\hbox{\rlap{$\sqcap$}$\sqcup$}}
\def\qed{\ifmmode\squareforqed\else{\unskip\nobreak\hfil
\penalty50\hskip1em\null\nobreak\hfil\squareforqed
\parfillskip=0pt\finalhyphendemerits=0\endgraf}\fi}

\newcommand{\ttt}[1]{\texttt{#1}}
\newcommand{\gdesc}[1]{\text{\textit{#1}}}
\newcommand{\true}{\mathrm{true}}

\newcommand{\cod}[1]{\llbracket #1 \rrbracket}
\newcommand{\lcod}[2]{\llbracket #1 \rrbracket_{#2}}

\newcommand{\store}{m}

\newcommand{\extend}[3]{#1\{#2 \mapsto #3\}}

\newcommand{\assign}[2]{#1 := #2}

\newcounter{topiccounter}
\newcommand{\redx}{\rightarrow}
\newcommand{\redxs}{\redx^*}

\newcommand{\elet}[3]{\mathrm{let}\ #1 = #2\ \mathrm{in}\ #3}

\newcommand{\dom}{\mathrm{dom}}
\newcommand{\rvs}{\mathrm{rvs}}

\newcommand{\latel}{\ell}
\newcommand{\hilab}{\mathrm{High}}
\newcommand{\lolab}{\mathrm{Low}}

\newcommand{\minifed}{\mathit{Overture}}
\newcommand{\minicat}{\minifed}
\newcommand{\fedprot}{\minifed}
\newcommand{\metaprot}{\mathit{Prelude}}

\newcommand{\prog}{\pi}
\newcommand{\main}{\mathit{main}}

\newcommand{\bop}{\ \mathit{binop}\ }

\newcommand{\mems}{\mathit{mems}}

\newcommand{\margd}[2]{{#1}_{#2}}
\newcommand{\condd}[3]{#1_{({#2}|{#3})}}
\newcommand{\progd}{\mathrm{BD}}
\newcommand{\progtt}{\mathrm{BD}}
\newcommand{\vars}{\mathit{vars}}
\newcommand{\iov}{\mathit{iovars}}
\newcommand{\flips}{\mathit{flips}}
\newcommand{\fedcat}{\minifed}

\newcommand{\idealf}{\mathcal{F}}

\def\TirName#1{\text{\sc #1}}

\newcommand{\cidty}{\ttt{cid}}
\newcommand{\stringty}{\ttt{string}}

\newcommand{\flip}[2]{\ttt{flip[}#1\ttt{,}#2\ttt{]}}
\newcommand{\secret}[2]{\ttt{s[}#1\ttt{,}#2\ttt{]}}

\renewcommand{\elet}[3]{\ttt{let}\ #1\ \ttt{=}\ #2\ \ttt{in}\ #3}

\renewcommand{\redx}{\xrightarrow{}}
\renewcommand{\redxs}{\xrightarrow{}^{*}}

\newcommand{\secrets}{\mathit{secrets}}

\newcommand{\views}{\mathit{views}}

\newcommand{\cid}{\iota}

\newcommand{\OT}[3]{\ttt{OT(} #1 \ttt{,}\ #2 \ttt{,}\ #3 \ttt{)}}

\newcommand{\mux}[3]{\ttt{mux(} #1 \ttt{,}\ #2 \ttt{,}\ #3 \ttt{)}}
\newcommand{\codebase}{\mathcal{C}}

\newcommand{\flab}{\ell}
\newcommand{\be}{\varepsilon}
\newcommand{\instr}{\mathbf{c}}

\newcommand{\runs}{\mathit{runs}}

\newcommand{\concat}{\ttt{++}}

\newcommand{\sep}[3]{#1 \vdash #2 * #3}

\newcommand{\condp}[3]{#1|#2 \vdash #3}

\newcommand{\condsep}[4]{\condp{#1}{#2}{#3 * #4}}

\newcommand{\pmf}{\mathcal{P}}

\renewcommand{\flip}[1]{\ttt{r[}#1\ttt{]}}
\newcommand{\locflip}{\ttt{r[}\mathtt{local}\ttt{]}}
\renewcommand{\secret}[1]{\ttt{s[}#1\ttt{]}}

\newcommand{\mesg}[1]{\ttt{m[}#1\ttt{]}}
\newcommand{\outkw}{\ttt{out}}
\newcommand{\out}[1]{\elab{\outkw}{#1}}
\newcommand{\rvl}[1]{\ttt{p[}#1\ttt{]}}

\newcommand{\elab}[2]{#1\ttt{@}#2}
\newcommand{\eassign}[4]{\elab{#1}{#2} := \elab{#3}{#4}}
\newcommand{\xassign}[3]{#1 := \elab{#2}{#3}}
\newcommand{\pubout}[3]{\out{#1} := \elab{#2}{#3}}
\newcommand{\reveal}[3]{\rvl{#1} := \elab{#2}{#3}}

\newcommand{\adversary}{\mathcal{A}}
\newcommand{\aredx}{\redx_{\adversary}}
\newcommand{\aredxs}{\redxs_{\adversary}}
\newcommand{\arewrite}{\mathit{rewrite}_{\adversary}}
\newcommand{\cinputs}{V_{C \rhd H}}
\newcommand{\houtputs}{V_{H \rhd C}}
\newcommand{\imesgs}[3]{#1_{#2 \rhd #3}}
\newcommand{\aruns}{\mathit{runs}_\adversary}
\newcommand{\botruns}{\mathit{runs}_{\adversary,\bot}}

\renewcommand{\store}{\sigma}

\renewcommand{\runs}{\mathit{runs}}

\newcommand{\fcod}[1]{\lcod{#1}{}}
\renewcommand{\flips}{\mathit{rands}}

\newcommand{\ftimes}{*}
\newcommand{\fplus}{+}
\newcommand{\fminus}{-}
\newcommand{\macgv}[1]{\langle #1 \rangle}

\newcommand{\preproc}{\mathit{preproc}}
\newcommand{\assert}[1]{\ttt{assert(}#1\ttt{)}}
\newcommand{\mv}{\nu}

\newcommand{\sx}[2]{\elab{\secret{#1}}{#2}}
\newcommand{\mx}[2]{\elab{\mesg{#1}}{#2}} 
\newcommand{\px}[1]{\rvl{#1}} 
\newcommand{\rx}[2]{\elab{\flip{#1}}{#2}}

\newcommand{\itj}[3]{\vdash_{#1} #2 : #3}
\newcommand{\ipj}[3]{#1 \vdash #2 : #3}

\newcommand{\setit}[1]{\{ #1 \}}

\newcommand{\ity}[2]{#1 \cdot #2}

\newcommand{\eqs}{\mathit{E}}
\newcommand{\toeq}[1]{\lfloor #1 \rfloor}
\newcommand{\eop}{\equiv}

\newcommand{\seclev}{\mathcal{L}}

\newcommand{\cheatj}[3]{#1 \underset{#2}{\leadsto} #3}

\renewcommand{\redx}{\Rightarrow}
\renewcommand{\redxs}{\redx}
\newcommand{\abort}{\bot}

\newcommand{\cmd}{\instr}
\renewcommand{\OT}[4]{\ttt{OT(}\elab{#1}{#2},#3,#4 \ttt{)}}
\newcommand{\eqspre}{\eqs_{\mathit{pre}}}
\newcommand{\progpre}{\prog_{\mathit{init}}}
\newcommand{\macbdoz}[1]{\psi_{\mathit{BDOZ}}(#1)}

\newcommand{\notg}[1]{\breve{#1}}
\newcommand{\fresh}{\mathit{fresh}}
\newcommand{\fieldp}[1]{\mathbb{F}_{#1}}
\newcommand{\precond}{\mathit{precond}}
\newcommand{\postcond}{\mathit{postcond}}

\newcommand{\htrip}[3]{\{ #1 \}\ #2\ \{ #3 \}}
\newcommand{\peq}{\notg{\eqs}}

\newcommand{\eqtrue}{\mathrm{true}}
\newcommand{\avars}{\mathit{avars}}
\newcommand{\impl}{\rightarrow}
\newcommand{\pskip}{\ttt{skip}}
\newcommand{\sharety}{\mathrm{shares}}
\newcommand{\pubty}{\mathrm{public}}
\newcommand{\outty}{\mathrm{output}}
\newcommand{\genshares}{\mathrm{genshares}}
\newcommand{\recon}{\mathrm{recon}}

\newcommand{\xdefs}{\mathcal{D}}
\newcommand{\kwas}{\ \ttt{as}\ }

\newcommand{\shtrip}[3]{#1 \Vdash #2 \rhd #3}

\newcommand{\realviews}[2]{V_{#1 \rhd #2}}

\newcommand{\oid}{\cid_{\mathit{oracle}}}
\newcommand{\ov}{\mathit{ovscrape}}

\newcommand{\pubx}[1]{\langle #1 \rangle}
\newcommand{\sta}{\mathbf{a}}
\newcommand{\stredx}{\leadsto}
\newcommand{\sty}{\varsigma}
\newcommand{\stt}{\theta}
\newcommand{\clean}{\mathit{clean}}
\newcommand{\bv}[1]{\ttt{|}#1\ttt{|}}

\newcommand{\compwrapfig}
{
  \begin{figure}[t]
    \begin{tabular}{lccccccc}
      & 
      \begin{sideways} probabilistic features \end{sideways} &
      \begin{sideways} probabilistic conditioning \end{sideways} & 
      \begin{sideways} low-level protocols \end{sideways} & 
      \begin{sideways} passive security \end{sideways} & 
      \begin{sideways} malicious security \end{sideways}& 
      \begin{sideways} hyperproperties \end{sideways}& 
      \begin{sideways} automation \end{sideways}\\\hline\hline
      Haskell EDSL \cite{6266151} & \checkmark &  & \checkmark  & \checkmark & & \checkmark & \checkmark \\\hline
      MPC in SecreC \cite{almeida2018enforcing} & \checkmark & \checkmark &   & \checkmark & & \checkmark & \checkmark \\\hline
      $\lambda_{\text{obliv}}$ \cite{darais2019language} & \checkmark & & \checkmark & & & \checkmark & \checkmark \\\hline
      PSL \cite{barthe2019probabilistic} & \checkmark & & \checkmark & & & & \\\hline
      Lilac \cite{li2023lilac} & \checkmark & \checkmark & & & & & \\\hline
      Wys$^*$ \cite{wysstar} & & & & \checkmark & & \checkmark & \checkmark\\\hline
      Viaduct \cite{10.1145/3453483.3454074,viaduct-UC} & & & & \checkmark & \checkmark & \checkmark & \checkmark\\\hline
      MPC in EasyCrypt \cite{8429300} &  \checkmark &  \checkmark &  \checkmark & \checkmark & \checkmark & \checkmark & \\\hline
      Swanky \cite{8429300} &  \checkmark & &  \checkmark & & & & \checkmark \\\hline
      $\metaprot$/$\minifed$ 1 \cite{skalka-near-ppdp24} & \checkmark & \checkmark & * & * & * & \checkmark & * \\\hline
      $\metaprot$/$\minifed$ 2 \cite{skalka-near-esop25} & \checkmark & \checkmark & \checkmark & * & * & \checkmark & * \\\hline
      This work & \checkmark & \checkmark & \checkmark & \checkmark & \checkmark & \checkmark & \checkmark\\
      \hline
    \end{tabular}
    \caption{Comparison of systems for verification of MPC security in PLs. * indicates relatively limited support.}
    \label{fig-comp-wrap}
  \end{figure}
}

\newcommand{\minicatsyntaxfig}{
\begin{fpfig}[t]{Syntax of $\minicat$}{fig-minicat-syntax}
$$
    \begin{array}{rcl@{\hspace{4mm}}r}
      \multicolumn{4}{l}{v \in \fieldp{p},\ w \in \mathrm{String},\ \cid \in \mathrm{Clients} }\\[2mm] 
      \be &::=& v \mid \flip{w} \mid \secret{w} \mid \mesg{w} \mid \rvl{w} \mid \be \fminus \be \mid \be \fplus \be \mid \be \ftimes \be \mid \OT{\be}{\cid}{\be}{\be} & \textit{expressions}\\[1mm]
      x &::=& \elab{\flip{w}}{\cid} \mid \elab{\secret{w}}{\cid} \mid \elab{\mesg{w}}{\cid} \mid  \rvl{w} \mid \out{\cid} & \textit{located variables} \\[1mm]
      \prog &::=& \eassign{\mesg{w}}{\cid}{\be}{\cid} \mid \reveal{w}{\be}{\cid} \mid \pubout{\cid}{\be}{\cid} \mid \prog;\prog \mid \pskip & \textit{protocols}
    \end{array}
$$
\end{fpfig}    
}

\newcommand{\minicatredxfig}{
\begin{fpfig}[t]{Semantics of $\minicat$ expressions (top) and programs (bottom).}{fig-minicat-redx}
 $$
  \begin{array}{c@{\hspace{5mm}}c}
  \begin{array}{rcl}
    \lcod{\store, v}{\cid} &=& v\\
    \lcod{\store, \be_1 \fplus \be_2}{\cid} &=& \fcod{\lcod{\store, \be_1}{\cid} \fplus \lcod{\store, \be_2}{\cid}}\\ 
    \lcod{\store, \be_1 \fminus \be_2}{\cid} &=& \fcod{\lcod{\store, \be_1}{\cid} \fminus \lcod{\store, \be_2}{\cid}}\\ 
    \lcod{\store, \be_1 \ftimes \be_2}{\cid} &=& \fcod{\lcod{\store, \be_1}{\cid} \ftimes \lcod{\store, \be_2}{\cid}}\\
  \end{array} & 
  \begin{array}{rcl}
    \lcod{\store, \flip{w}}{\cid} &=& \store(\elab{\flip{w}}{\cid})\\
    \lcod{\store, \secret{w}}{\cid} &=& \store(\elab{\secret{w}}{\cid})\\
    \lcod{\store, \mesg{w}}{\cid} &=& \store(\elab{\mesg{w}}{\cid})\\
    \lcod{\store, \rvl{w}}{\cid} &=& \store(\rvl{w})\\
  \end{array}
  \end{array}
  $$

\begin{mathpar}
  (\store, \pskip) \redx \store
  
  (\store, \xassign{x}{\be}{\cid}) \redx \extend{\store}{x}{\lcod{\store,\be}{\cid}}
  
  \inferrule
      {(\store_1,\prog_1) \redx \store_2 \\ (\store_2,\prog_2) \redx \store_3 }
      {(\store_1,\prog_1;\prog_2) \redx \store_3}
\end{mathpar}
\end{fpfig}
}

\newcommand{\minicataredxfig}{
\begin{fpfig}[t]{Adversarial semantics of $\minicat$.}{fig-minicat-aredx}
\begin{mathpar}
  \inferrule
      { \cid \in H }
      { (\store, \xassign{x}{\be}{\cid}) \aredx \extend{\store}{x}{\lcod{\store,\be}{\cid}} }
      
  \inferrule
      {\cid \in C }
      { (\store, \xassign{x}{\be}{\cid}) \aredx \extend{\store}{x}{\lcod{\arewrite(\store_C,\be)}{\cid}}}
      
  \inferrule
      {\lcod{\store,\be_1}{\cid} = \lcod{\store,\be_2}{\cid}  \text{\ or\ } \cid \in C}
      { (\store,\elab{\assert{\be_1 = \be_2}}{\cid}) \aredx \store }
      
  \inferrule
      {\lcod{\store,\be_1}{\cid} \ne \lcod{\store,\be_2}{\cid}}
      {(\store,\elab{\assert{\be_1 = \be_2}}{\cid}) \aredx \abort}
  
  \inferrule
      {(\store_1,\prog_1) \aredx \store_2 \\ (\store_2,\prog_2) \aredx \store_3 }
      {(\store_1,\prog_1;\prog_2) \aredx \store_3}

  \inferrule
      {(\store_1,\prog_1) \aredx \abort}
      {(\store_1,\prog_1;\prog_2) \aredx \abort}
      
  \inferrule
      {(\store_1,\prog_1) \aredx \store_2 \\ (\store_2,\prog_2) \aredx \abort }
      {(\store_1,\prog_1;\prog_2) \aredx \store_2}
\end{mathpar}
\end{fpfig}
}

\newcommand{\toeqfig}{
\begin{fpfig}[t]{Interpretation of $\minicat$ expressions (top) and programs (bottom) as
  constraints}{fig-toeq}
\begin{mathpar}
  \toeq{v} = v
  
  \toeq{x} = x

  \toeq{\elab{\be_1 \fplus \be_2}{\cid}} = (\toeq{\elab{\be_1}{\cid}} \fplus \toeq{\elab{\be_2}{\cid}})

  \toeq{\elab{\be_1 \fminus \be_2}{\cid}} = (\toeq{\elab{\be_1}{\cid}} \fminus \toeq{\elab{\be_2}{\cid}})

  \toeq{\elab{\be_1 \ftimes \be_2}{\cid}} = (\toeq{\elab{\be_1}{\cid}} \ftimes \toeq{\elab{\be_2}{\cid}})
\end{mathpar}
\begin{mathpar}
  \toeq{\xassign{x}{\be}{\cid}} = (x \eop \toeq{\elab{\be}{\cid}})
  
  \toeq{\elab{\assert{\be_1 = \be_2}}{\cid}} =  (\toeq{\elab{\be_1}{\cid}} \eop \toeq{\elab{\be_2}{\cid}})

  \toeq{\prog_1;\prog_2} = (\toeq{\prog_1} \wedge \toeq{\prog_2})
\end{mathpar}
\end{fpfig}
}

\newcommand{\minicathtripfig}{
\begin{fpfig}[t]{$\minicat$ Hoare logic derivation rules}{fig-minicathtrip}
\begin{mathpar}
  \inferrule[Skip]
            {}
            {\htrip{\eqs}{\pskip}{\eqs}}
            
  \inferrule[Assign]
            {}
            {\htrip{\eqs[\toeq{\elab{\be}{\cid_2}}/x]}{\xassign{x}{\be}{\cid_2}}{\eqs}}
  
  \inferrule[Assert]
            {E \models \toeq{\elab{\be_1}{\cid}} \eop \toeq{\elab{\be_2}{\cid}}}
            {\htrip{E}{\elab{\assert{\be_1 = \be_2}}{\cid}}{E}}

  \inferrule[Seq]
      {\htrip{\eqs_1}{\prog_1}{\eqs_2} \\ \htrip{\eqs_2}{\prog_2}{\eqs_3} }
      {\htrip{\eqs_1}{\prog_1;\prog_2}{\eqs_3}}

  \inferrule[Consequence]
      {\htrip{\eqs_1'}{\prog}{\eqs_2'} \\ \eqs_1 \models \eqs_1' \\  \eqs_2' \models \eqs_2 }
      {\htrip{\eqs_1}{\prog}{\eqs_2}}

  \inferrule[Frame]
      {\htrip{\eqs_1}{\prog}{\eqs_2} \\ \vars(\eqs) \cap \avars(\prog) = \varnothing}
      {\htrip{\eqs_1 \wedge \eqs}{\prog}{\eqs_2 \wedge \eqs}}
      
\inferrule[Hyp]
          {\htrip{\eqs_1}{\prog}{\eqs_2}}
          {\htrip{\eqs_1 \wedge (\eqs_2 \impl \eqs)}{\prog}{\eqs_2 \wedge \eqs}}
\end{mathpar}
\end{fpfig}
}

\newcommand{\stylocalfig}{
\begin{fpfig}[t]{Local (non-interactive) share type derivation rules.}{fig-stylocal}
\begin{mathpar}
  \inferrule[Shared]
      {(\stt : \sty) \in \Gamma}
      {\Gamma,\eqs \vdash \stt : \sty}
      
  \inferrule[Val]
      {}
      {\Gamma,\eqs \vdash \pubx{v} : \pubty}

  \inferrule[HEAdd]
      {\Gamma,\eqs \vdash [\phi^1_1,\phi^1_2] : \sharety\\
        \Gamma,\eqs \vdash [\phi^2_1,\phi^2_2] : \sharety}
      {\Gamma,\eqs \vdash [\phi^1_1\fplus\phi^2_1,\phi^1_2\fplus\phi^2_2] : \sharety}

  \inferrule[HEAddPub]
      {\Gamma,\eqs \vdash \pubx{\phi} : \pubty\\
        \Gamma,\eqs \vdash [\phi_1,\phi_2] : \sharety}
      {\Gamma,\eqs \vdash [\phi_1\fplus\phi,\phi_2] : \sharety}
            
  \inferrule[HEMultPub]
      {\Gamma,\eqs \vdash [\phi] : \pubty\\
        \Gamma,\eqs \vdash [\phi_1,\phi_2] : \sharety}
      {\Gamma,\eqs \vdash [\phi_1\ftimes\phi,\phi_2\ftimes\phi] : \sharety}

  \inferrule[PubOp]
      {\Gamma,\eqs \vdash \pubx{\phi_1} : \pubty\\
        \Gamma,\eqs \vdash \pubx{\phi_2} : \pubty}
      {\Gamma,\eqs \vdash \pubx{\phi_1\bop\phi_2}: \pubty}
      
  \inferrule[ShareEntails]
      {\Gamma,\eqs \vdash [\phi_1',\phi_2'] : \sharety \\
        \eqs \models \phi_1 \eop \phi_1' \wedge \phi_2 \eop \phi_2'}
      {\Gamma,\eqs \vdash [\phi_1,\phi_2] : \sharety} 
      
  \inferrule[PubEntails]
      {\Gamma,\eqs \vdash \pubx{\phi'} : \sty \\ \eqs \models \phi \eop \phi'}
      {\Gamma,\eqs \vdash \pubx{\phi} : \sty} 
\end{mathpar}
\end{fpfig}
}

\newcommand{\styinteractivefig}{
\begin{fpfig}[t]{Interactive share type derivation rules.}{fig-styinteractive}
\begin{mathpar}
  \inferrule[GenShares]
      {
        \genshares(\phi,\rx{w}{\cid}) = \phi_1,\phi_2 \\ 
        \eqs \models x_1 \eop \phi_1 \wedge x_2 \eop \phi_2  
      }
      {\eqs \vdash \Gamma,R \stredx \Gamma \cup \setit{[x_1,x_2] : \sharety}, R \uplus \setit{\rx{w}{\cid}}}

  \inferrule[SecOpen]
      {
        \Gamma, \eqs \vdash [\px{w_1},\px{w_2}] : \sharety\\
        \eqs \models  \px{w_1} + \px{w_2} \eop \phi' + \rx{w}{\oid}
      }
      {\eqs \vdash \Gamma,R \stredx \Gamma \cup \setit{\pubx{\recon(\px{w_1},\px{w_2})}  : \pubty}, R \uplus \setit{\rx{w}{\cid_{\oid}}}}

  \inferrule[Reveal]
      {
        \Gamma, \eqs \vdash [\px{w_1},\px{w_2}] : \sharety
      }
      {\eqs \vdash \Gamma,R \stredx \Gamma \cup \setit{\pubx{\recon(\px{w_1},\px{w_2})}  : \outty}, R}

  \inferrule[Output]
      {
        \Gamma, \eqs \vdash \pubx{\out{\cid}} : \outty
      }
      {\eqs \vdash \Gamma,R \stredx \Gamma \cup \setit{\pubx{\out{\cid}} : \outty}, R}
\end{mathpar}
\end{fpfig}
}

\newcommand{\ipjfig}{
\begin{fpfig}[t]{Syntax and derivation rules of $\minicat$ integrity types}{fig-ipj}
$$
\begin{array}{rcl@{\hspace{4mm}}l}
  \latel &::=& \hilab \mid \lolab & \gdesc{integrity labels} \\
  \Delta &::=& \varnothing \mid \Delta; x : \ity{\cid}{V} & \gdesc{integrity type environments}
\end{array} 
$$

\begin{mathpar}
  \inferrule[Value]
  {}
  {\itj{\cid}{v}{\varnothing}}
  
  \inferrule[Secret]
  {}
  {\itj{\cid}{\secret{w}}{\varnothing}}
  
  \inferrule[Rando]
  {}
  {\itj{\cid}{\flip{w}}{\varnothing}}
  
  \inferrule[Mesg]
  {}
  {\itj{\cid}{\mesg{w}}{\setit{\mx{w}{\cid}}}}
    
  \inferrule[PubM]
  {}
  {\itj{\cid}{\rvl{w}}{\setit{\rvl{w}}}}

  \inferrule[Binop]
  {\itj{\cid}{\be_1}{V_1} \\
   \itj{\cid}{\be_2}{V_2} \\ \oplus \in \{ \fplus,\fminus,\ftimes \}}
  {\itj{\cid}{\be_1 \oplus \be_2}{V_1 \cup V_2}}
\end{mathpar}

\begin{mathpar}
  \inferrule[Send]
            {\itj{\cid}{\be}{V}}
            {\ipj{\eqs}{\xassign{x}{\be}{\cid}}{(x : \ity{\cid}{V})}}
             
  \inferrule[Seq]
            {\ipj{\eqs}{\prog_1}{\Delta_1}\\
             \ipj{\eqs}{\prog_2}{\Delta_2}}
            {\ipj{\eqs}{\prog_1;\prog_2}{\Delta_1;\Delta_2}}

  \inferrule[MAC]
            {\eqs \models \toeq{\elab{\assert{\macbdoz{w}}}{\cid}}}
            {\ipj{\eqs}{\elab{\assert{\macbdoz{w}}}{\cid}}{(\mx{w\ttt{s}}{\cid}: \ity{\cid}{\varnothing})}}
\end{mathpar}
\end{fpfig}
}

\newcommand{\cheatjfig}{
\begin{fpfig}[t]{Assigning integrity labels to variables}{fig-cheatj}
\begin{mathpar}
  \inferrule
      {}
      {\cheatj{\varnothing}{H,C}{\seclev_{H,C}}}
      
  \inferrule
      {\cheatj{\Delta}{H,C}{\seclev} \\ \cid \in H}
      {\cheatj{\Delta; x : \ity{\cid}{V}}{H,C}{\extend{\seclev}{x}{\hilab \wedge (\bigwedge_{x \in V} \seclev(x))}}}
      
  \inferrule
      {\cheatj{\Delta}{H,C}{\seclev} \\ \cid \in C}
      {\cheatj{\Delta; x : \ity{\cid}{V}}{H,C}{\extend{\seclev}{x}{\lolab}}}
\end{mathpar}
\end{fpfig}
}

\newcommand{\metaprotsyntaxfig}{
  \begin{fpfig}[t]{$\metaprot$ syntax}{fig-metaprotsyntax}
$$
\begin{array}{rcl@{\hspace{7mm}}r}
  \multicolumn{3}{l}{\flab \in \mathrm{Field},\   y \in \mathrm{EVar}, \  f \in \mathrm{FName}}\\[1mm]
  e &::=& \mv \mid \flip{e} \mid \secret{e} \mid \mesg{e} \mid \rvl{e} \mid \outkw \mid e \bop e 
   \mid y \mid \bv{e,\ldots,e} \mid & \gdesc{expressions}\\
  & & e.\flab \mid \elab{e}{e} \mid \elet{y}{e}{e} \mid  f(e,\ldots,e) \mid \{ \flab = e; \ldots; \flab = e \} \\[1mm]
  \cmd &::=& 
  \assign{e}{e} \mid f(e,\ldots,e) \mid \elet{y}{e}{\instr} \mid  \cmd;\cmd 
    \mid \elab{\assert{e = e}}{e} & \gdesc{instructions}\\[1mm] 
  \bop &::=& \fplus \mid \fminus \mid \ftimes \mid \oplus \mid \ominus \mid \concat  \\[1mm]
  \mv &::=& w \mid \cid \mid \be \mid x \mid \{ \flab = \mv;\ldots;\flab = \mv \} 
  & \gdesc{values}\\[1mm] 
  \tau &::=& \stringty \mid \cidty \mid \{ \flab_1 : \tau_1; \ldots; \flab_n : \tau_n \} & \gdesc{basic types}\\[1mm]
\mathit{fn} &::=& f(y,\ldots,y) \{ e \} \mid  f(y : \tau,\ldots,y : \tau) \{ \cmd \} & \textit{functions} 
\end{array}
$$
\end{fpfig}
}

\newcommand{\metaprotexprsemanticsfig}{
  \begin{fpfig}[t]{Semantics of $\metaprot$ expressions.}{fig-metaprotexprsemantics}
  \begin{mathpar}
  \inferrule
      {e_1 \redx \mv \\ e_2[\mv/y] \redx \mv'}
      {\elet{y}{e_1}{e_2} \redx \mv'}
      
  \inferrule
      {e_1 \redx \be \\ e_2 \redx \cid}
      {\elab{e_1}{e_2} \redx \elab{\be}{\cid}}

  \inferrule
      {\codebase(f) = y_1,\ldots,y_n,\ e \\ e_1 \redx \mv_1 \cdots e_n \redx \mv_n \\
        e[\mv_1/y_1]\cdots[\mv_n/y_n] \redx \mv}
      {f(e_1,\ldots,e_n) \redx \mv}

  \inferrule
      {e_1 \redx \mv_1 \\ \cdots \\ e_n \redx \mv_n }
      {\{ \flab_1 = e_1; \ldots; \flab_n = e_n \} \redx \{ \flab_1 = \mv_1; \ldots; \flab_n = \mv_n \} }

  \inferrule
      {e \redx \{\ldots; \flab = \mv; \ldots\}}
      {e.\flab \redx \mv}

  \inferrule
      {e_1 \redx w_1 \\ e_2 \redx w_2}
      {e_1 \concat e_2 \redx w_1w_2}

  \inferrule
      {e \redx w}
      {\mesg{e} \redx \mesg{w}}
      
  \inferrule
      {e_1 \redx \be_1 \\ e_2 \redx \be_2}     
      {e_1 \fplus e_2 \redx \be_1 \fplus \be_2}
      
  \inferrule
      {}     
      {\mv \redx \mv}
\end{mathpar}
\end{fpfig}
}

\newcommand{\metaprotinstrsemanticsfig}{
  \begin{fpfig}[t]{Semantics of $\metaprot$ instructions.}{fig-metaprotinstrsemantics}
\begin{mathpar}
  \inferrule
      {e_1 \redx x \\ e_2 \redx \elab{\be}{\cid}}
      {\assign{e_1}{e_2} \redx \xassign{x}{\be}{\cid}}

  \inferrule
      {e_1 \redx \be_1 \\ e_2 \redx \be_2 \\ e_3 \redx \cid}
      {\elab{\assert{e_1 = e_2}}{e_3} \redx \elab{\assert{\be_1 = \be_2}}{\cid}}

  \inferrule
      {\codebase(f) = y_1 : \tau_1, \ldots, y_n : \tau_n, \instr \\
        e_1 \redx \mv_1 \ \cdots \ e_n \redx \mv_n \\
        \instr[\mv_1/y_1]\cdots[\mv_n/y_n] \redx \prog
      }
      {f(e_1,\ldots,e_n) \redx \prog}
      
  \inferrule
      {e \redx \mv \\ \instr[\mv/y] \redx \prog}
      {\elet{y}{e}{\instr} \redx \prog}

  \inferrule
      {e_1 \redx \prog_1 \\ e_2 \redx \prog_2}
      {e_1;e_2 \redx \prog_1;\prog_2}
\end{mathpar}
\end{fpfig}
}

\newcommand{\peqfig}{
\begin{fpfig}[t]{Abstract constraints and their evaluation semantics.}{fig-peq}
$$
  \begin{array}{rclr}
    \peq &::=& e \eop e \mid \notg{\eqs} \wedge \notg{\eqs} \mid \notg{\eqs} \vee \notg{\eqs}
    \mid \notg{\eqs} \impl \notg{\eqs} \mid \neg\notg{\eqs} & \qquad \textit{abstract constraints}
  \end{array}
$$
\begin{mathpar}
  \inferrule
      {e_1 \redx \phi_1 \cdots e_n \redx \phi_n}
      {\bv{e_1,\ldots,e_n} \redx \bv{\phi_1,\ldots,\phi_n}}
  
  \inferrule
      {e_1 \redx \elab{\be_1}{\cid_1} \\ e_2 \redx \elab{\be_2}{\cid_2} }
      {e_1 \bop e_2 \redx \toeq{\elab{\be_1}{\cid_1}} \bop \toeq{\elab{\be_2}{\cid_2}}}
      
  \inferrule
      {e_1 \redx \elab{\be_1}{\cid_1} \\ e_2 \redx \elab{\be_2}{\cid_2} }
      {e_1 \eop e_2 \redx \toeq{\elab{\be_1}{\cid_1}} \eop \toeq{\elab{\be_2}{\cid_2}}}
      
  \inferrule
      {\notg{\eqs_1} \redx \eqs_1 \\ \notg{\eqs_2} \redx \eqs_2 }
      {\notg{\eqs_1} \wedge \notg{\eqs_2} \redx \eqs_1 \wedge \eqs_2}

  \inferrule
      {\notg{\eqs_1} \redx \eqs_1 \\ \notg{\eqs_2} \redx \eqs_2 }
      {\notg{\eqs_1} \vee \notg{\eqs_2} \redx \eqs_1 \vee \eqs_2}

  \inferrule
      {\notg{\eqs_1} \redx \eqs_1 \\ \notg{\eqs_2} \redx \eqs_2 }
      {\notg{\eqs_1} \impl \notg{\eqs_2} \redx \eqs_1 \impl \eqs_2}

  \inferrule
      {\peq \redx \eqs}
      {\neg\peq\redx \neg\eqs}
\end{mathpar}
\end{fpfig}
}

\newcommand{\metahtripfig}{
\begin{fpfig}[t]{Algorithmic $\metaprot$ Abstract Entailment and Hoare Triple Deduction Rules}{fig-metahtrip}
\begin{mathpar}          
  \inferrule[Mesg]
            {}
            {\htrip{\eqtrue}{\xassign{e_1}{e_2}{e_3}}{e_1 \eop \elab{e_2}{e_3}}}

  \inferrule[Assert]
            {}
            {\htrip{\elab{e_1}{e_3} \eop \elab{e_2}{e_3}}{\elab{\assert{e_1 = e_2}}{e_3}}{\eqtrue}}

  \inferrule[Seq]          
            {\htrip{\peq_1^1}{\cmd_1}{\peq_2^1} \\ \htrip{\peq_1^2}{\cmd_2}{\peq_2^2}}
            {\htrip{\peq_1^1 \wedge (\peq_2^1 \impl \peq_1^2)}{\cmd_1;\cmd_2}{\peq_2^1 \wedge \peq_2^2}}

  \inferrule[Let]
            {\htrip{\peq_1}{\cmd[e/y]}{\peq_2}}
            {\htrip{\peq_1}{\elet{y}{e}{\cmd}}{\peq_2}}

  \inferrule[App]
            {\htrip{\peq_1}{f(y_1:\tau_1,\ldots,y_n:\tau_n)}{\peq_2}}
            {\htrip{\peq_1[e_1/y_1 \cdots e_n/y_n]}{f(e_1,\ldots,e_n)}{\peq_2[e_1/y_1 \cdots e_n/y_n]}}

  \inferrule[Fn]
            {\codebase(f) = y_1 : \tau_1, \ldots, y_n : \tau_n, \cmd \\ \htrip{\peq_1}{\cmd}{\peq_2}}
            {\htrip{\peq_1}{f(y_1:\tau_1,\ldots,y_n:\tau_n)}{\peq_2}}

  \inferrule[Hardpack]
            {\precond(f) = \peq_1 \\ \postcond(f) = \peq_2 \\
              \htrip{\peq_1'}{f(y_1:\tau_1,\ldots,y_n:\tau_n)}{\peq_2'} \\
              \forall y_1:\tau_1,\ldots,y_n:\tau_n . 
            \peq_1 \models \peq_1'\ \text{and}\ \peq_1 \wedge \peq_2' \models \peq_2}
            {\htrip{\peq_1}{f(y_1:\tau_1,\ldots,y_n:\tau_n)}{\peq_2}}
\end{mathpar}\\[4mm]
$$
  \inferrule[GenEntails]
          {\mv_1,\ldots,\mv_n = \fresh(\tau_1,\ldots,\tau_n) \\ \peq_1[\mv_1/y_n \cdots \mv_n/y_n] \redx \eqs_1 \\
           \peq_2[\mv_1/y_n \cdots \mv_n/y_n] \redx \eqs_2 \\ \eqs_1 \models \eqs_2}
          {\forall y_1:\tau_1,\ldots,y_n:\tau_n . \peq_1 \models \peq_2 }
$$
\end{fpfig}
}

\newcommand{\stafig}{
\begin{fpfig}[t]{$\metaprot$ share type annotation syntax (top), semantics (middle), and
    hoare triple derivation rules (bottom).}{fig-sta}
$$
\begin{array}{rclr}
  \notg{\stt} &::=& [e,e] \mid \pubx{e}   & \qquad \textit{abstract share terms} \\   
  \sta &::=& e_1 \kwas e_2 \mid \notg{\stt} : \sty & \qquad \textit{share type annotations} \\     
  \cmd &::=& \cdots \mid \sta
\end{array}
$$

\bigskip

\begin{mathpar}
  \inferrule
      {e_1 \redx \phi_1 \\ e_2 \redx \phi_2}
      {[e_1,e_2] \redx [\phi_1,\phi_2]}
      
  \inferrule
      {e \redx \phi}
      {\pubx{e} \redx \pubx{\phi}}
      
  \inferrule
      {e_1 \redx \phi_1 \\ e_2 \redx \phi_2}
      {e_1 \kwas e_2 \redx \pskip}

  \inferrule
      {\notg{\stt} \redx \stt}
      {\notg{\stt} : \sty \redx \pskip}
\end{mathpar}

\bigskip

\begin{mathpar}
  \htrip{e_1 \eop e_2}{e_1 \kwas e_2}{\true}

  \htrip{\true}{\notg{\stt} : \sty}{\true}
\end{mathpar}
\end{fpfig}
}

\newcommand{\shtripstafig}{
\begin{fpfig}[t]{Typing rules for $\metaprot$ share type annotations.}{fig-shtripsta}
\begin{mathpar}
  \inferrule[Coerce]
      {e_1 \redx \phi_1 \\ e_2 \redx \phi_2}
      {\shtrip{\Gamma, R, \xdefs}{e_1 \kwas e_2}{\Gamma, R, \xdefs[\phi_1 \mapsto \phi_2]}}

  \inferrule[LocalShares]
      {e_1 \redx x_1 \\ e_2 \redx x_2  \\
        \Gamma \Vdash [\xdefs(x_1), \xdefs(x_2)] : \sharety}
      {\shtrip{\Gamma, R, \xdefs}{[e_1,e_2] : \sharety}
        {\Gamma \cup \setit{[x_1,x_2] : \sharety}, R, \xdefs}}

  \inferrule[GenShares]
      {e_1 \redx x_1 \\ e_2 \redx x_2  \\ \genshares(\phi, \rx{w}{\cid}) = \xdefs(x_1),\xdefs(x_2)}
      {\shtrip{\Gamma, R, \xdefs}{[e_1,e_2] : \sharety}
        {\Gamma \cup \setit{[x_1,x_2] : \sharety}, R \uplus \setit{\rx{w}{\cid}}, \xdefs}}

  \inferrule[LocalPublic]
      {e \redx \phi \\ \Gamma \Vdash \pubx{\xdefs(\phi)} : \pubty}
      {\shtrip{\Gamma, R, \xdefs}{\pubx{e} : \pubty}
        {\Gamma \cup \setit{\pubx{\phi} : \pubty}, R, \xdefs}}

  \inferrule[MaskedPublic]
      {e_1 \redx x_1 \\ e_2 \redx x_2  \\
        \Gamma \Vdash [x_1, x_2] : \sharety \\
        \xdefs(\recon(x_1,x_2))  = \phi + \rx{w}{\oid}}
      {\shtrip{\Gamma, R, \xdefs}{\pubx{\recon(e_1,e_2)} : \pubty}
        {\Gamma \cup \setit{\pubx{\recon(x_1,x_2)} : \pubty}, R \uplus \setit{\rx{w}{\oid}}, \xdefs}}

  \inferrule[Reveal]
      {e_1 \redx x_1 \\ e_2 \redx x_2  \\
        \Gamma \Vdash [x_1, x_2] : \sharety}
      {\shtrip{\Gamma, R, \xdefs}{\pubx{\recon(e_1,e_2)} : \outty}
        {\Gamma \cup \setit{\pubx{\recon(x_1,x_2)} : \outty}, R, \xdefs}}

  \inferrule[Output]
      {e \redx \out{\cid} \\
       \Gamma \Vdash \pubx{\xdefs(\out{\cid})} : \outty }
      {\shtrip{\Gamma, R, \xdefs}{\out{\cid} : \outty}
        {\Gamma \cup \setit{\pubx{\out{\cid}} : \outty}, R, \xdefs}}
\end{mathpar}
\end{fpfig}
}

\newcommand{\shtripcmdfig}{
\begin{fpfig}[t]{Typing rules for $\metaprot$ instructions and functions.}{fig-shtripcmd}
\begin{mathpar}
  \inferrule[Defn]
      {e_1 \redx x \\ e_2 \redx \elab{\be}{\cid}}
      {\shtrip{\Gamma, R, \xdefs}{e_1 := e_2}{\Gamma, R, \xdefs[x \mapsto \toeq{\elab{\be}{\cid}}]}}

  \inferrule[Seq]    
      {\shtrip{\Gamma_1,R_1,\xdefs_1}{\cmd_1}{\Gamma_2,R_2,\xdefs_2} \\
       \shtrip{\Gamma_2,R_2,\xdefs_2}{\cmd_2}{\Gamma_3,R_3,\xdefs_3}}
      {\shtrip{\Gamma_1,R_1,\xdefs_1}{\cmd_1;\cmd_2}{\Gamma_3,R_3,\xdefs_3}}

  \inferrule[Let]
      {\shtrip{\Gamma_1,R_1,\xdefs_1}{\cmd[e/y]}{\Gamma_2,R_2,\xdefs_2}}
      {\shtrip{\Gamma_1,R_1,\xdefs_1}{\elet{y}{e}{\cmd}}{\Gamma_2,R_2,\xdefs_2}}

  \inferrule[App]
      {f : \Pi y_1,\ldots,y_n. \notg{\Gamma_1} \rightarrow (\notg{\Gamma_2}, \notg{R'}) \\
        \notg{\Gamma_1}[e_1/y_1 \cdots e_n/y_n] \redx \Gamma_1 \\
        \notg{\Gamma_2}[e_1/y_1 \cdots e_n/y_n] \redx \Gamma_2 \\
        \notg{R'}[e_1/y_1 \cdots e_n/y_n] \redx R'        
      }
      {\shtrip{\Gamma \cup \Gamma_1, R, \xdefs}{f(e_1,\ldots,e_n)}{\Gamma \cup \Gamma_1 \cup \Gamma_2, R \uplus R', \xdefs}}
      
  \inferrule[Fn]
            {\codebase(f) = y_1 : \tau_1, \ldots, y_n : \tau_n, \cmd \\
              \mv_1,\ldots,\mv_n = \fresh(\tau_1,\ldots,\tau_n) \\
              \notg{\Gamma_1}[\mv_1/y_1 \cdots \mv_n/y_n] \redx \Gamma_1 \\
              \notg{\Gamma_2}[\mv_1/y_1 \cdots \mv_n/y_n] \redx \Gamma_2 \\
              \notg{R}[\mv_1/y_1 \cdots \mv_n/y_n] \redx R \\  
     \shtrip{\Gamma_1, \varnothing, \varnothing}{\cmd[\mv_1/y_1 \cdots \mv_n/y_n]}{\Gamma_2, R,\xdefs}}
    {f : \Pi y_1,\ldots,y_n. \notg{\Gamma_1} \rightarrow (\notg{\Gamma_2}, \notg{R})}
\end{mathpar}
\end{fpfig}
}

\newcommand{\preludetestfig}{
\begin{fpfig}[t]{Average verification times for $\metaprot$ binary and arithmetic circuit libraries and small circuits over 20 trials, in seconds. Cells marked with * were limited to a single trial due to excessive verification times. Experiments that timed out did not finish within 8 hours.}{fig-preludetest}
\begin{center}
\begin{tabular}{c}
\begin{tabular}{l|c|c|c|c|c}
         & Bit-Vector & SMFF $\fieldp{2}$ & SMFF $\fieldp{2^{17} - 1}$ &  SMFF $\fieldp{2^{31} - 1}$ &
      SMFF $\fieldp{2^{61} - 1}$ \\
  \hline
  \text{Passive Beaver Library} & 1.012 & .911 & 1.135 & 1.114 & 1.059 \\
  \text{Malicious Beaver (BDOZ) Library} & 1.175 & 1.158 & 1.182 & 1.153 & 1.178 \\
  \text{GMW Library} & .932 & .936 \\ 
\end{tabular}\\[10mm]
\hline\\
\begin{tabular}{l|c|c|c}
  &  \begin{tabular}{l}Bit-Vector,\\Compositional\\(this work)\end{tabular} &
       \begin{tabular}{l} Bit-Vector,\\Whole-Program\\(prior work~\cite{skalka-near-ppdp24})\end{tabular} & \begin{tabular}{l} SMFF $\fieldp{2}$\\(this work)\end{tabular} \\
  \hline
  \text{GMW 2-Bit Zero Test (3 Gates)} & .037 & .029 & 261.098 \\
  \text{GMW 2-Bit Addition (4 Gates)} & .038 & .032 & 261.098 \\
  \text{GMW 8-Bit Zero Test (15 Gates)} & .051 & .057 & Timeout \\
  \text{GMW 64-Bit Zero Test (127 Gates)} & 0.308 & 1.071 & Timeout \\
  \text{GMW 64-Bit Addition (376 Gates)} & 5.735 & 4623.666* & Timeout \\
  \text{GMW 64-Bit Subtraction (439 Gates)} & 9.985 & 10661.325* & Timeout \\
\end{tabular}
\end{tabular}
\end{center}
\end{fpfig}
}

\newcommand{\solvetimescatterfig}{
\begin{fpfig}[t]{Scatter plot of verification solve times.}{fig-solve-time-scatter}
\begin{center}
\includegraphics[width=\linewidth]{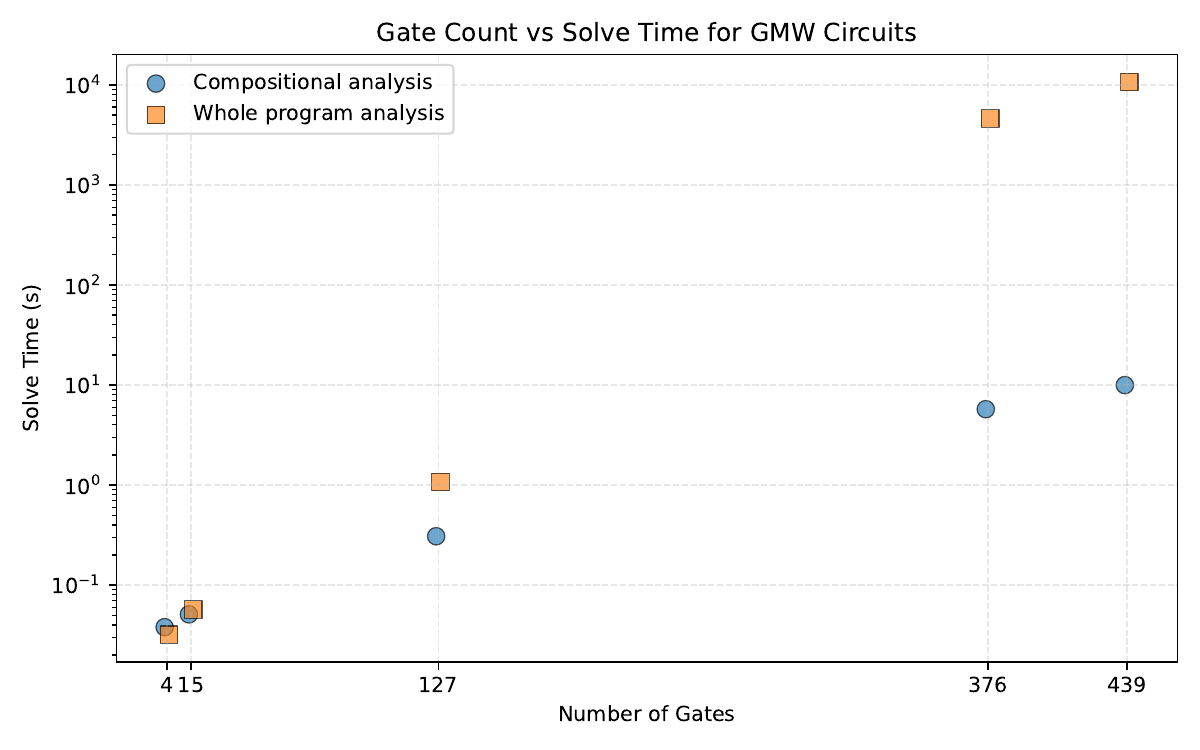}
\end{center}
\end{fpfig}
}

\begin{document}

\title{SMT-Boosted Security Types and Correctness for Low-Level MPC}

\author{Christian Skalka}
\affiliation{%
  \institution{University of Vermont}
  \city{Burlington, Vermont}
  \country{USA}}
\email{ceskalka@uvm.edu}

\author{Steven Baldasty}
\affiliation{%
  \institution{University of Vermont}
  \city{Burlington, Vermont}
  \country{USA}}
\email{sbaldast@uvm.edu}

\author{Yeonhee Yeh}
\affiliation{%
  \institution{University of Vermont}
  \city{Burlington, Vermont}
  \country{USA}}
\email{yyeh@uvm.edu}

\author{Felix Walberg}
\affiliation{%
  \institution{University of Vermont}
  \city{Burlington, Vermont}
  \country{USA}}
\email{fwalberg@uvm.edu}

\author{Joseph P.~Near}
\affiliation{%
  \institution{University of Vermont}
  \city{Burlington, Vermont}
  \country{USA}}
\email{jnear@uvm.edu}

\begin{abstract}
  Secure Multi-Party Computation (MPC) is an important enabling
  technology for data privacy in modern distributed applications. We
  develop a new Hoare logic and type theory to automatically enforce
  correctness, confidentiality, and integrity properties of protocols
  written in the \emph{Prelude/Overture} language framework. Both
  definition and automation of our analyses leverages SMT
  verification, supporting precision and efficiency. Our approach is
  compositional, scalable, and generalizes to arbitrary prime fields
  and secret sharing schemes.
\end{abstract}

\maketitle

\section{Introduction}

Data privacy is a critical concern in distributed applications
including privacy-preserving machine learning
\cite{li2021privacy,knott2021crypten,koch2020privacy,liu2020privacy}
and blockchains
\cite{ishai2009zero,lu2019honeybadgermpc,gao2022symmeproof,tomaz2020preserving}.
Secure multiparty computation (MPC) is a key software-based enabling
technology for these applications. MPC protocols provide both confidentiality
and integrity properties, formulated as real/ideal (aka simulator)
security and universal composability (UC) \cite{evans2018pragmatic},
with well-studied manual proof methods \cite{Lindell2017}.  MPC
security semantics has also been reformulated as
\emph{hyperproperties}
\cite{8429300,10.1145/3453483.3454074,skalka-near-ppdp24} to support
automated language-based proof methods. Recent research in the SMT
community has developed theories of finite fields \cite{SMFF} with
clear relevance to verification of cryptographic schemes that rely on
field arithmetic. The goal of this paper is to combine hyperproperty
formulations of security with SMT verification as a foundation for
enforcing correctness and security properties of MPC protocols.

Our focus is on low-level MPC protocols. Previous high-level
MPC-enabled languages such as Wysteria \cite{rastogi2014wysteria} and
Viaduct \cite{10.1145/3453483.3454074} are designed to provide
effective programming of full applications, and incorporate
sophisticated compilation techniques such as orchestration
\cite{viaduct-UC} to guarantee high-level security properties. But
these approaches rely on \emph{libraries} of low-level MPC protocols,
such as binary and arithmetic circuits. These low-level protocols are
probabilistic, encapsulate abstractions such as secret sharing and
semi-homomorphic encryption, and are complementary to high-level
language design.

\subsection{The $\metaprot/\minicat$ Language Framework}

Our language framework incorporates two languages. The $\minicat$
language supports direct definition of message passing protocols.  In
this work we develop a \emph{share type} system enforcing the
confidentiality property of MPC which is called passive simulator
security. We also develop and integrity type systems for $\minicat$
that enforces malicious MPC security. Both share types and integrity
types are supported with the use of an SMT-style constraint system.
In addition to security properties, the constraint system also
supports protocol correctness properties.

While $\minicat$ provides a foundation for low-level MPC protocol
specification, it does not provide higher-level features such as
function definitions, and the share and integrity type systems are not
automated. So we define the $\metaprot$ language as a higher-level
metaprogramming language that evaluates to $\minicat$ protocols.
$\metaprot$ supports compositional definition of protocol components
and also supports automation of constraint verification and type
checking especially by leveraging SMT automation.  Thus $\metaprot$
program analysis automates security and correctness properties of the
low-level $\minicat$ protocols that $\metaprot$ programs evaluate to.

The specification of an MPC protocol is called its \emph{ideal
functionality} and a standard approach \cite{269581} is to implement
ideal functionalities as either arithmetic or boolean circuits. An
advantage of this is that subprotocols for primitive operations, such
as addition and multiplication, or logic gates, can be implemented and
optimized as subprotocols. Consider the example of a two-bit adder
binary circuit in Figure \ref{fig-adder}. The ideal functionality of
the circuit is two-bit vector addition (ignoring carry). On the left
is a $\metaprot$ program that implements the protocol by ``wiring''
secure logic gates implemented using the Goldreich-Micali-Wigderson
(GMW) protocol~\cite{goldreich1987play}. These are encapsulated in the $\ttt{xorgmw}$ and
$\ttt{andgmw}$ definitions. In the first lines of the program, clients
1 and 2 each generate secret shares of their secret bits-- $x_1$ and
$x_2$ for client 1, and $y_1$ and $y_2$ for client 2. Shares are
identified in the program by strings $\ttt{"x1"}$, $\ttt{"y1"}$, etc.,
as are the gates in the protocol. So the circuit is wired in the
same way as the logic circuit diagram on the right. Each gate
specifies a subprotocol where the operations are performed on
encrypted wire data-- in the case of GMW, these are secret shares.

\begin{fpfig}[t]{A two-bit adder circuit (without carry) in $\metaprot$ with
    verification postcondition, using the GMW protocol.}{fig-adder}
\begin{center}
\begin{tabular}{cc}
\begin{tabular}[b]{l}
\verb+main() {+\\
\verb+   genshares("x1", 1, 2); genshares("x2", 1, 2);+\\ 
\verb+   genshares("y1", 2, 1); genshares("y2", 2, 1);+\\ 
\verb+   xorgmw("z1", "x1", "y1");+ \\
\verb+   andgmw("g1", "x1", "y1");+\\
\verb+   xorgmw("g2", "x2", "y2");+\\
\verb+   xorgmw("z2", "g2", "g1")+\\
\verb+}+ \\
\verb+postcondition:+\\
\verb+  (|RECON("z2"), RECON("z1")| ==+\\
\verb+      BVAdd(|RECON("x2"), RECON("x1")|,+\\
\verb+            |RECON("y2"), RECON("y1")|))+
\end{tabular}
& 
{
\scriptsize
\begin{circuitikz}
  \draw
  (0,0)
  node[label=left:$x_1$] {}
  to[short, o-]
  (1,0)
  node[xor port, anchor=in 1] (z1) {}
  (0,52 |- z1.in 2)
  node[label=left:$y_1$] {}
  to[short, o-]
  (z1.in 2)
  (z1.out)
  to[short, -o]
  ++(3,0)
  node[label=right:$z_1$] {}
  (2,-1.5)
  node[and port, anchor=in 1] (g1) {}
  (g1.in 1)
  -- ++(-1,0)
  to[short,-*] (1,0)
  (g1.in 2)
  -- ++(-1.5,0)
  to[short, -*] (.5, 52 |- z1.in 2)
  (0,-3)   
  node[label=left:$x_2$] {}
  to[short, o-]
  (1,-3)
  node[xor port, anchor=in 1] (g2) {}
  (0,52 |- g2.in 2)
  node[label=left:$y_2$] {}
  to[short, o-]
  (g2.in 2)
  (g1.out)
  to[short,-]
  ++(.2,0)
  node[xor port, anchor=in 1] (z2) {}
  (g2.out)
  -- ++(1.2,0)
  to[short, -] (3.75, 52 |- z2.in 2)
  (z2.out)
  to[short, -o]
  ++(.25,0)
  node[label=right:$z_2$] {}  
  ;
\end{circuitikz}
}
\end{tabular}
\end{center}
\end{fpfig}

An important feature of the example is the $\ttt{postcondition}$ that
specifies the correctness of the circuit. Here the $\ttt{|...|}$
syntax denotes a concatenation of bits into a bit vector, and
$\ttt{BVAdd}$ is bit vector addition. These operations are components
of the SMT cvc5 Bit-Vector theory and the postcondition is
automatically verified using it. $\ttt{RECON}$ is a constraint-level
function indicating the reconstructed value represented by secret
shares. So the postcondition says that the reconstructed shares of
the outputs are equal to the sum of the reconstructed secret shares
of the inputs.

In the remainder of the paper we will go into more complete detail
about the underlying definitions of $\ttt{genshares}$, $\ttt{xorgmw}$,
and $\ttt{andgmw}$ (see in particular Figure \ref{fig-gmwprepost} in
Section \ref{section-metalang} and Figure \ref{fig-gmwsharetypes} in
Section \ref{section-metaty}). This includes their share type checking that
guarantees passive simulator security in this and other examples.
Another important point is that other protocols can be used to define
gates, for example malicious secure protocols, since $\metaprot$
and $\minicat$ support low-level protocol specification.

\subsection{A Brief Introduction to MPC Protocols}

An MPC protocol allows a collection of parties to compute a function
of their private inputs without revealing those inputs to each other.
The function to be computed is typically called the \emph{ideal
  functionality}, and a correct protocol must produce the same output
as the ideal functionality. An MPC protocol is secure if it reveals no
additional information to the parties beyond what is intentionally
revealed by the ideal functionality. The information available to a
party (its input, randomness, received messages, and outputs) is
called its \emph{view}; an adversary can observe the views of all
corrupted parties~\cite{evans2018pragmatic,Lindell2017}.

\emph{Secret sharing} splits a secret value into \emph{shares} that
allow limited computations but prevent individual parties from
learning the secret. \emph{Additive secret sharing} is one example: to
additively share $x \in \mathbb{F}_p$ between two parties, we sample a
uniform $r \in \mathbb{F}_p$ and gives the shares $x_1=r$ and
$x_2=x-r$ to the respective parties. To reconstruct the secret, we
compute $x_1+x_2=x$. Each share by itself is uniform and therefore
does not reveal $x$, but two shares together uniquely determine $x$.
This construction is a finite-field \emph{one-time pad} (OTP): for
every possible $x$, the observed value $x-r$ has the same uniform
distribution, so even an adversary with unlimited computational power
learns nothing about $x$. The mask must be fresh for every secret. For
example, reusing $r$ to mask $x$ and $y$ reveals their difference,
since $(x-r)-(y-r)=x-y$. In this paper $\ttt{RECON}$ denotes
reconstruction, and over $\mathbb{F}_2$ additive reconstruction is
XOR, as used for GMW wire values \cite{goldreich1987play, evans2018pragmatic}.

Additive secret sharing is a form of \emph{semi-homomorphic}
encryption: it allows limited forms of computation without decryption.
Adding two additive sharings componentwise gives a sharing of the sum.
However, additive secret shares cannot be multiplied to obtain shares
of the product; to support multiplication, MPC protocols use
additional machinery (e.g. oblivious transfer in GMW, or preprocessed
Beaver triples in arithmetic protocols~\cite{goldreich1987play,
  evans2018pragmatic, BDOZ,10.1007/978-3-030-68869-1_3}).

\emph{Oblivious transfer} (OT) is a two-party protocol that allows one
party to select a value from the other without revealing which one was
selected. The \emph{sender} supplies two values $m_0,m_1$, and the
\emph{receiver} supplies a choice bit $b$. The receiver obtains $m_b$
without learning $m_{1-b}$, while the sender learns nothing about $b$.
For example, GMW uses OT to select an entry from a table that
implements a Boolean gate without revealing either party's
secret-shared input bits. Different OT implementations rely on
different assumptions, so our language treats OT abstractly and
requires only these functional and privacy guarantees
\cite{goldreich1987play, evans2018pragmatic}.

Another approach moves expensive work into an input-independent
\emph{preprocessing} phase that distributes correlated randomness
before the parties know their private inputs. A Beaver
triple~\cite{beaver1991efficient} consists of secret shares of random
values $a,b,c$ satisfying $c=ab$. To multiply shared $x$ and $y$, the
parties reconstruct the masked differences $d=x-a$ and $e=y-b$, then
locally obtain shares of $xy=c+db+ea+de$. Reusing a Beaver triple
leads to correlations that break security, so triples and other
preprocessing resources must be fresh where required and securely
generated; this paper specifies their needed properties without
implementing that generation protocol \cite{evans2018pragmatic,
  BDOZ,10.1007/978-3-030-68869-1_3}.

MPC confidentiality is conventionally defined by comparing a
\emph{real} protocol execution with an \emph{ideal} execution in which
a trusted party receives the inputs, evaluates the ideal
functionality, and returns only its specified outputs. A simulator for
a corrupt coalition receives only the information available in the
ideal execution---in particular, the coalition's inputs and prescribed
outputs---and must generate a simulated view indistinguishable from
its real protocol view. The existence of such a simulator means that
\emph{anything learned from the real messages could also have been
  learned from the allowed ideal information}---in other words, the
real protocol reveals no more than the ideal
functionality~\cite{Lindell2017}.

Depending on the setting, indistinguishability may mean equality of
distributions, statistical closeness, or computational
indistinguishability. Equality gives \emph{perfect}, or
\emph{information-theoretic}, security even against unlimited
computation. Statistical security permits a small difference between
the distributions, bounding the distinguishing advantage even for
unlimited computation. Computational security instead considers
efficient adversaries and rests on a hardness assumption. Its
\emph{security parameter}, such as a key length (or a field size for
field-based authentication), controls both cost and protection; a
distinguishing advantage is \emph{negligible} when it decreases faster
than any inverse polynomial as that parameter grows. This paper uses
equality of finite distributions in its core model.

The adversary model determines how the corrupt parties behave. A
\emph{passive} (or honest-but-curious) adversary records and analyzes
the combined view of corrupt parties, but those parties follow the
protocol. A \emph{malicious} (or active) adversary may instead choose
arbitrary inputs, messages, and local computations. Malicious security
therefore combines confidentiality and
integrity~\cite{evans2018pragmatic}.

BDOZ and SPDZ protect integrity using \emph{authenticated
  shares}~\cite{SPDZ1, SPDZ2, BDOZ, 10.1007/978-3-030-68869-1_3}. A
\emph{message authentication code} (MAC) is a tag that binds a value
to secret authentication information; in these protocols, another
party holds the information needed to check a share's tag. A corrupt
party cannot tamper with a share, since it is impossible to produce a
corresponding valid tag. The MACs are linearly homomorphic, so tags
can be updated consistently when shares are added. Before revealing an
authenticated value, the parties check the resulting tag and abort if
it is invalid.

Ideal outputs may \emph{intentionally} reveal information, so MPC
security differs from ordinary noninterference. MPC security requires
only that the protocol's execution reveal nothing \emph{additional}
compared to the ideal functionality. Section \ref{section-model}
formalizes this idea as \emph{noninterference modulo output} and uses
it as a hyperproperty that implies passive simulator
security.

\subsection{Summary of the Theory and Implementation}

Here we summarize the main theoretical developments and results in
the paper and their connections.  In Section \ref{section-lang}
we define the $\minicat$ language model. In Section \ref{section-model}
we articulate our security model as hyperproperties of
noninterference modulo output (Definition \ref{definition-NIMO})
and integrity (Definition \ref{def-integrity}). As has been
shown in previous work \cite{skalka-near-ppdp24, skalka-near-esop25},
noninterference modulo output is sound for passive simulator
security, while the addition of integrity obtains malicious security.
Compared to this prior whole-program analysis, our use of weaker
customized postconditions supports compositional reuse of verified
components. Section \ref{section-examples} quantifies the resulting
performance improvement.
Figure \ref{fig-framework-workflow} summarizes how the languages,
analyses, and guarantees fit together.

\begin{figure*}[t]
\centering
\resizebox{\textwidth}{!}{%
\begin{tikzpicture}[
  font=\sffamily\footnotesize,
  language/.style={
    draw=blue!60!black, fill=blue!8, rounded corners=2pt,
    align=center, text width=27mm, minimum height=9mm, inner sep=3pt
  },
  analysis/.style={
    draw=orange!70!black, fill=orange!10, rounded corners=2pt,
    align=center, text width=28mm, minimum height=9mm, inner sep=3pt
  },
  guarantee/.style={
    draw=green!50!black, fill=green!9, rounded corners=2pt,
    align=center, text width=28mm, minimum height=9mm, inner sep=3pt
  },
  implementation/.style={
    draw=black!55, fill=black!4, rounded corners=2pt,
    align=center, text width=27mm, minimum height=9mm, inner sep=3pt
  },
  group/.style={
    draw=black!45, dashed, rounded corners=3pt, inner sep=6pt
  },
  flow/.style={
    -{Latex[length=2mm]}, thick, draw=black!70
  },
  link/.style={
    thick, draw=black!70
  },
  edge label/.style={
    font=\scriptsize, inner sep=0pt, yshift=4pt, text=black!80
  }
]

\node[language] (prelude) at (0,4.1)
  {\textbf{Prelude}\\language and annotations\\Sec.~\ref{section-metalang}};
\node[analysis] (prelude-hoare) at (0,2.4)
  {Automated abstract\\Hoare logic\\Sec.~\ref{section-metahtrip}};
\node[analysis] (prelude-types) at (0,0.7)
  {Automated share\\type checking\\Sec.~\ref{sec:metapr-type-check}};
\node[analysis, dashed] (prelude-integrity) at (0,-1.0)
  {Automated integrity\\type checking\\Prior work~\cite{skalka-near-esop25}};

\node[language] (overture) at (4.2,4.1)
  {\textbf{Overture}\\protocol semantics\\Sec.~\ref{section-lang}};
\node[analysis] (overture-hoare) at (4.2,2.4)
  {SMT interpretation\\and Hoare logic\\Sec.~\ref{section-smt}};
\node[analysis] (share-types) at (4.2,0.7)
  {Share types\\Sec.~\ref{section-sty}};
\node[analysis] (integrity-types) at (4.2,-1.0)
  {Integrity types\\Sec.~\ref{section-ipj}};

\node[guarantee] (correctness) at (8.4,2.4)
  {Functional correctness\\of pre/postconditions};
\node[guarantee] (passive) at (8.4,0.7)
  {Noninterference modulo output\\and passive security\\Sec.~\ref{section-model}};
\node[guarantee] (integrity) at (8.4,-1.0)
  {Integrity\\Sec.~\ref{section-model}};
\node[guarantee] (malicious) at (12.3,-0.15)
  {Malicious security\\Sec.~\ref{section-model}};

\node[group, fit=(prelude)(prelude-hoare)(prelude-types)(prelude-integrity),
  label={[font=\bfseries\small,text=black!70]above:Prelude language}] {};
\node[group, fit=(overture)(overture-hoare)(share-types)(integrity-types),
  label={[font=\bfseries\small,text=black!70]above:Overture language}] {};
\node[group, fit=(correctness)(passive)(integrity)(malicious),
  label={[font=\bfseries\small,text=black!70]above:Guarantees}] {};

\draw[flow] (prelude) -- node[edge label] {evaluates to} (overture);
\draw[flow] (prelude-hoare) -- node[edge label] {instantiates} (overture-hoare);
\draw[flow] (prelude-types) -- node[edge label] {preserves} (share-types);
\draw[flow] (prelude-integrity) -- node[edge label] {preserves} (integrity-types);
\draw[flow] (overture-hoare) -- node[edge label] {proves} (correctness);
\draw[flow] (share-types) -- node[edge label] {enforces} (passive);
\draw[flow] (integrity-types) -- node[edge label] {enforces} (integrity);

\coordinate (security-join) at (10.35,-0.15);
\draw[link] (passive.east) -| (security-join);
\draw[link] (integrity.east) -| (security-join);
\draw[flow] (security-join) -- node[edge label] {} (malicious.west);




\end{tikzpicture}
}
\caption{Relationships among the paper's formal components and security
guarantees. Blue boxes denote languages and their semantics, orange
boxes denote verification analyses, and green boxes denote the
correctness and security guarantees they establish. Dashed outlines
group components by abstraction level.}
\label{fig-framework-workflow}
\end{figure*}
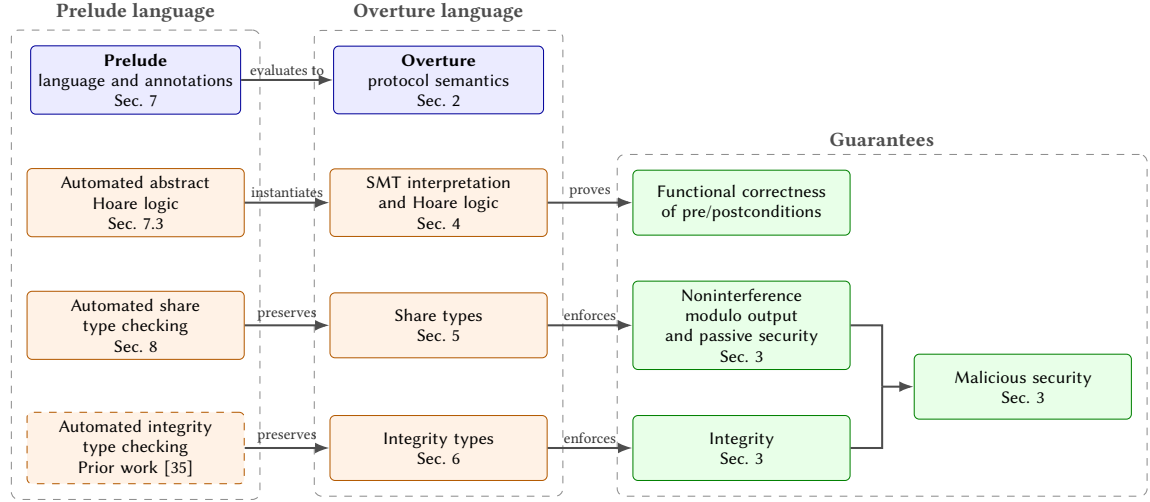

In Section \ref{section-smt} we develop an SMT-based Hoare logic
endowed with notions of validity and derivability, where the latter is
shown to imply the former (Lemma \ref{lemma-hoare-sound}). This is an
important foundation both for our verification method and for notions
of constraint correctness in type systems.  In Section
\ref{section-sty} we develop a novel share type system, enabling
security verification for protocols that use secret sharing as a
building block; our approach works for any secret sharing scheme and
enforces noninterference modulo output and therefore passive security
(Theorem \ref{theorem-sty}). The share type system also supports
verification of protocols with respect to ideal functionality.  In
Section \ref{section-ipj}, we develop a different type system that
enforces integrity (Theorem \ref{theorem-ipj}). Thus both share and
integrity typability of a $\minicat$ protocol imply that it is
malicious secure.

In Section \ref{section-metalang} we define the $\metaprot$ syntax and
semantics, including pre- and postcondition annotations on functions
(as in the example in Figure \ref{fig-adder}). Intuitively, any
$\metaprot$ program evaluates to a $\minicat$ protocol, so we refer to
it as a metaprogramming language. We define an automated Hoare logic
for $\minicat$ programs, to verify pre- and postcondition
annotations. This automated Hoare logic is shown to be sound in the
sense that pre- and postconditions will be valid for any function
applications (Lemma \ref{lemma-metahtrip}). An implication of this
(Corollary \ref{cor-metahtrip}) means that function pre- and
postconditions are compositional and need only be verified once,
achieving greater performance efficiency. In Section
\ref{sec:metapr-type-check}, we develop an automated type system for
$\metaprot$ that guarantees that any well-typed $\metaprot$ program
evaluates to a share-typable $\minicat$ protocol (Theorem
\ref{theorem-metaty}) and hence its passive security.

In Section \ref{section-examples} we develop and discuss extended
example applications of our type analyses to real protocols including
the Goldreich-Micali-Wigderson (GMW), Ben-Or-Goldwasser-Widgerson
(BGW), and Bendlin-Damgard-Orland-Zakarias (BDOZ). We also summarize
our implementation of the $\metaprot$ Hoare logic and evaluate
verification of both binary and arithmetic gate and circuit
properties. Our implementation integrates with cvc5 SMT to support
efficient verification, including Satisfiability Modulo Finite Fields
(SMFF) \cite{SMFF} for binary and arithmetic circuits, and the
Bit-Vector theory for binary circuits. The $\metaprot$ Hoare logic
analysis has been fully implemented in the current work in an
Antlr/Scala/cvc5-based workflow.

The implementation that produced our experimental results is available in our
companion
repository\footnote{\url{https://github.com/uvm-plaid/prelude-public}}. It uses
ANTLR4 to parse Prelude the source code, Scala to convert the Prelude abstract
syntax tree into cvc5 constraints, and the cvc5 solver to prove that for all
functions the precondition and function body meet the postcondition. When
built, the tool runs as a single Java archive (JAR). The repository includes
installation instructions, source code, further documentation, and a variety of
Prelude examples.

\subsection{Related Work and Our Contributions}
\label{section-related-work}

\compwrapfig

The goal of the $\metaprot/\minicat$ framework is to automate
correctness and security verification of low-level protocols.  Other
prior work has considered automated verification of high-level
protocols and manual verification of low-level protocols.  We
summarize this comparison in Figure \ref{fig-comp-wrap}.  In this
Figure we consider related systems in several dimensions.  This
includes whether the language system provides \emph{probabilistic
features} (e.g., sampling) that are fundamental to low-level
protocols. It also includes whether existing verification methods
include \emph{probabilistic conditioning} that is essential to
condition analysis of information leakage inherent in ideal
functionalities. We consider whether the system is designed for
\emph{low-level} protocols, as opposed to being a framework for
integrating them in applications (e.g., during compilation). We
consider whether the system has integrated methods and/or mechanisms
for proving \emph{passive} and \emph{malicious security}, and whether
sound \emph{hyperproperties} have been formulated to support
enforcement of security via program analysis such as type systems.
Finally, we consider whether proof/enforcement mechanisms are
\emph{automated} in the system, vs., e.g., realized in proof
assistants.

The most closely related works are previous versions of the
$\metaprot/\minicat$ framework
\cite{skalka-near-ppdp24,skalka-near-esop25}. In the earliest work
\cite{skalka-near-ppdp24} foundations for $\metaprot/\minicat$ are
introduced including a verification algorithm. However, this algorithm
only works for $\mathbb{F}_2$-- e.g., binary circuits-- and the
datalog-based automation is a brute-force whole program analysis.
Subsequent work \cite{skalka-near-esop25} extends the language model
to arbitrary prime fields and introduces the use of SMT to optimize
verification in prime fields, though the automation is still a
whole-program analysis. The type systems presented in these works
support a gradual release property, weaker than noninterference
modulo output, and integrity properties. The current work provides
several advancements of $\metaprot/\minicat$:
\begin{itemize}
\item We develop a compositional Hoare logic that provides a foundation for
  compositional verification and type safety results.
\item We define a share type system that enforces noninterference modulo
  output as a result of type safety. This implies that share
  types enforce passive security and, together with integrity
  typability, imply malicious security.
\item Share types also support verification of ideal functionality
  specification of protocol outputs.
\item We extend the share type system to $\metaprot$, along with
  a compositional language of pre- and postconditions for functions.
  Type safety in $\metaprot$ ensures that share typing in programs
  is preserved by evaluation to $\minicat$ protocols.
\item The extensions in this work additionally enable the type-based
enforcement of noninterference modulo output for more complex secret
sharing protocols, e.g., that leverage Beaver triples.
\item We implement and evaluate the compositional verification of
  pre- and postconditions by considering examples of large circuits,
  specifically Bristol circuits. This evaluation demonstrates
  empirical advantages of compositional vs. whole-program verification.
\end{itemize}

The SAW/Cryptol framework \cite{10.1007/978-3-319-48869-1_5} supports
definitions of cryptographic protocols and verification of functional
correctness using SMT. The Swanky extension \cite{swanky} is
especially related to our work, since it is intended to apply to MPC
protocols in binary and prime fields. For example, Swanky has been
used in the implementation of ZK systems
\cite{cryptoeprint:2020/1410}. However, while the goals of our system
are overlapping with regard to verification of protocol correctness,
we are also focused on type systems for enforcement of security
properties, incorporating reasoning about probabilistic behavior of
protocols. SAW/Cryptol and Swanky, by contrast, are not designed to
verify security properties of protocols, since these are inherently
probabilistic. Our language model also more readily accommodates
primitives for message passing between multiple parties. But as we
discuss in Section \ref{section-conclusion}, leveraging Swanky to
accelerate verification is an interesting direction for future work.

Other low-level languages with probabilistic features for privacy
preserving protocols have been proposed. The
$\lambda_{\mathrm{obliv}}$ language \cite{darais2019language} uses a
type system to automatically enforce so-called probabilistic trace
obliviousness.  But similar to previous work on oblivious data
structures \cite{10.1145/3498713}, obliviousness is related to pure
noninterference, not the conditional form related to passive MPC
security (Definition \ref{definition-NIMO}). A Haskell-based security
type system has been defined as an embedded DSL \cite{6266151} that
enforces a version of noninterference that is sound for passive
security, but does not consider malicious security. Properties of
real/ideal passive and malicious security for a probabilistic language
have been (manually) formulated in EasyCrypt \cite{8429300}.

Several high-level languages have been developed for writing MPC
applications. Previous work on analysis for the
SecreC language \cite{almeida2018enforcing,10.1145/2637113.2637119} is
concerned with properties of complex MPC circuits, in particular a
user-friendly specification and automated enforcement of
declassification bounds in programs that use MPC in subprograms. The
Wys$^\star$ language \cite{wysstar}, based on Wysteria
\cite{rastogi2014wysteria}, has similar goals and includes a
trace-based semantics for reasoning about the interactions of MPC
protocols. Their compiler also guarantees that underlying
multi-threaded protocols enforce the single-threaded source language
semantics. These two lines of work were focused on passive
security. The Viaduct language \cite{10.1145/3453483.3454074} has a
well-developed information flow type system that automatically
enforces both confidentiality and integrity through hyperproperties
such as robust declassification, in addition to rigorous compilation
guarantees through orchestration \cite{viaduct-UC}. However, these
high level languages lack probabilistic features and other
abstractions of low-level protocols, the implementation and security
of which are typically assumed as a selection of library components.

Program logics for probabilistic languages and specifically reasoning
about properties such as joint probabilistic independence is also
important related work. Probabilistic Separation Logic (PSL)
\cite{barthe2019probabilistic} develops a logical framework for
reasoning about probabilistic independence (aka separation) in
programs, and they consider several (hyper)properties, such as perfect
secrecy of one-time-pads and indistinguishability in secret sharing,
that are critical to MPC. Lilac \cite{li2023lilac} extends this line
of work with formalisms for conditional independence which is 
also particularly important for MPC \cite{skalka-near-ppdp24}.

\section{$\minicat$ Syntax and Operational Semantics}
\label{section-lang}

The $\minifed$ language establishes a basic model of synchronous
protocols between a federation of \emph{clients} exchanging values in
a finite field. Synchronous means that the sender and receiver in each
communication ``block'' until data is exchanged. This is consistent
with standard models of MPC and real protocols \cite{evans2018pragmatic}.
In order to align with existing SMT theories
\cite{SMFF} we restrict our consideration to prime fields. We identify
clients by natural numbers and federations- finite sets of clients-
are always given statically.  Our threat model assumes a partition of
the federation into \emph{honest} $H$ and \emph{corrupt} $C$
subsets. We model probabilistic programming via a random tape
semantics. That is, we will assume that programs can make reference to
values chosen from a uniform random distribution defined in the
initial program memory.  Programs aka protocols execute
deterministically given the random tape.

\subsection{Syntax}

\minicatsyntaxfig

The syntax of $\minifed$, defined in Figure \ref{fig-minicat-syntax},
includes values $v$ and standard operations of addition, subtraction,
and multiplication in a finite field $\mathbb{F}_p$ where $p$ is some
prime. Protocols are given input secret values $\secret{w}$ as well
as random samples $\flip{w}$ on the input tape, implemented using a
\emph{memory} as described below (Section
\ref{section-lang-semantics}) where $w$ is a distinguishing 
identifier string. Protocols are sequences of assignment commands of three
different forms:
\begin{itemize}
\item $\eassign{\mesg{w}}{\cid_2}{\be}{\cid_1}$: This
  is a unicast \emph{message send} where expression $\be$ is computed
  by client $\cid_1$ and sent to client $\cid_2$ as message
  $\mesg{w}$.
\item $\reveal{w}{\be}{\cid}$: This
  is a \emph{public reveal} where expression $\be$ is computed
  by client $\cid$ and broadcast to the federation, typically
  to communicate intermediate results for use in final output
  computations.
\item $\pubout{\cid}{\be}{\cid}$: This
  is an \emph{output} where expression $\be$ is computed
  by client $\cid$ and reported as its output. As a
  sanity condition we disallow commands
  $\pubout{\cid_1}{\be}{\cid_2}$ where $\cid_1\ne\cid_2$.
\end{itemize}
The distinction between
messages and broadcast public reveal is consistent with previous
formulations, e.g., \cite{6266151}. While this distinction is important,
it will frequently be helpful to refer to located variables in general
so we collect them in a syntactic sort $x$. To help identify and distinguish
between collections of variables in protocols we introduce the
following notation.
\begin{definition}
We let $x$ range over \emph{located variables} which are identifiers where
client ownership is specified- e.g.,
$\elab{\mesg{\mathit{foo}}}{\cid}$ is a message $\mathit{foo}$ that
was sent to $\cid$. We let $X$ range over sets of variables, and more
specifically, $S$ ranges over sets of secret variables
$\elab{\secret{w}}{\cid}$, $R$ ranges over sets of random variables
$\elab{\flip{w}}{\cid}$, $M$ ranges over sets of message variables
$\elab{\mesg{w}}{\cid}$, $P$ ranges over sets of public variables
$\rvl{w}$, and $O$ ranges over sets of output variables $\out{\cid}$.
Given a program $\prog$, we write $\iov(\prog)$ to denote the
particular set $S \cup M \cup P \cup O$ of variables in $\prog$ and
$\secrets(\prog)$ to denote $S$, and we write $\flips(\prog)$ to
denote the particular set $R$ of random samplings in $\prog$. We write
$\vars(\prog)$ to denote $\iov(\prog) \cup \flips(\prog)$. For any set
of variables $X$ and clients $I$, we write $X_I$ to denote the subset
of $X$ owned by any client $\cid \in I$, in particular we write $X_H$
and $X_C$ to denote the subsets belonging to honest and corrupt
parties, respectively.
\end{definition}
In all cases we disallow overwriting of variables, since assignment is
intended to represent distinct message sends. An example of a
$\minicat$ protocol, for additive secret sharing, is discussed in
Section \ref{section-smt-toeq}.

\subsection{Semantics}
\label{section-lang-semantics}

\minicatredxfig

\emph{Memories} are fundamental to the semantics of $\fedcat$ and
provide random tape and secret inputs to protocols, and also record
message sends, public broadcast, and client outputs.
\begin{definition}
Memories $\store$ are finite (partial) mappings from variables $x$ to
values $v \in \mathbb{F}_p$.  The \emph{domain} of a memory is written
$\dom(\store)$ and is the finite set of variables on which the memory
is defined.  We write $\store\{ x \mapsto v\}$ for
$x\not\in\dom(\store)$ to denote the memory $\store'$ such that
$\store'(x) = v$ and otherwise $\store'(y) = \store(y)$ for all $y \in
\dom(\store)$. We write $\store \subseteq \store'$ iff $\dom(\store)
\subseteq \dom(\store')$ and $\store(x) = \store'(x)$ for all $x \in
\dom(\store)$. Given any $\store$ and $\store'$ with $\store(x) =
\store'(x)$ for all $x \in \dom(\store) \cap \dom(\store')$, we write
$\store \uplus \store'$ to denote the memory with domain $X =
\dom(\store) \cup \dom(\store')$ such that:
$$
\forall x \in X .
(\store \uplus \store')(x) =
\begin{cases} \store(x) \text{\ if\ } x\in\dom(\store) \\ \store'(x) \text{\ otherwise\ }\end{cases} 
$$
\end{definition}
In our subsequent presentation we will often want to consider arbitrary
memories that range over particular variables and to restrict
memories to particular subsets of their domain:
\begin{definition}
  Given a set of variables $X$ and memory $\store$, we write
  $\store_X$ to denote the memory with $\dom(\store_X) = X$ and
  $\store_X(x) = \store(x)$ for all $x \in X$. We define $\mems(X)$ as
  the set of all memories with domain $X$:
  $$
  \mems(X) \defeq \{ \store \mid \dom(\store) = X \}
  $$
\end{definition}
So for example, the set of all random tapes for
a protocol $\prog$ is $\mems(\flips(\prog))$, and the memory $\store_{\secrets(\prog)}$
is $\store$ restricted to the secrets in $\prog$.

Given a variable-free expression $\be$, we write $\cod{\be}$ to denote
the standard interpretation of $\be$ in the arithmetic field
$\mathbb{F}_{p}$. With the introduction of variables to expressions,
we need to interpret variables with respect to a specific memory, and
all variables used in an expression must belong to a specified client.
Thus, we denote interpretation of expressions $\be$ computed on a
client $\cid$ as $\lcod{\store,\be}{\cid}$. This interpretation is
defined in Figure \ref{fig-minicat-redx}, along with the big-step
reduction relation $\redx$ to evaluate commands. Reduction is a mapping
from \emph{configurations} $(\store,\prog)$ to final memories
where all three command forms- message send, broadcast, and
output- are implemented as updates to the memory $\store$.

\paragraph{Oblivious Transfer} A passive secure oblivious transfer (OT) protocol
based on previous work \cite{barthe2019probabilistic} can be defined
in $\minicat$, albeit assuming some shared randomness. Alternatively,
a simple passive secure OT can be defined with the addition of public
key cryptography as a primitive. But given the diversity of approaches
to OT, we instead assume that OT is abstract with respect to its
implementation. In particular, we do not define a semantics of OT in
Figure \ref{fig-minicat-redx}. However we will be able to define a
formal specification of OT-- we return to this in Section
\ref{section-smt}. Informally, any use is of the form $\mx{w}{\cid_1}
:= \elab{\OT{\be_1}{\cid_1}{\be_2}{\be_3}}{\cid_2}$ -- given a
\emph{choice bit} $\be_1$ provided by a receiver $\cid_1$, the sender
$\cid_2$ sends either $\be_2$ or $\be_3$. Critically, $\cid_2$ learns
nothing about $\be_1$ and $\cid_1$ learns nothing about the unselected
value.

\paragraph{Sanity conditions on protocols} To simplify
the presentation we will disregard protocols that are semantically
ill-defined. That is, we generally assume that for any
program $\prog$ there exists $\store$
with $(\store, \prog) \redxs \store$. This rules out
use of undefined variables and, in the adversarial semantics,
programs consisting of just a failing $\ttt{assert}$.

\subsection{$\minicat$ Adversarial Semantics}
  
\minicataredxfig

In the malicious model we assume that corrupt clients are in
thrall to an adversary $\adversary$ that does not necessarily follow
the rules of the protocol.  We model this by positing a $\arewrite$
function which is given a corrupt memory $\store_C$ and expression
$\be$, and returns a rewritten expression that can be interpreted to
yield a corrupt input. We define the evaluation relation that
incorporates the adversary in Figure \ref{fig-minicat-aredx}.

A key technical distinction of the malicious setting is that it
typically incorporates ``abort''. Honest parties implement strategies
to detect rule-breaking-- aka \emph{cheating}-- by using, e.g.,
message authentication codes with semi-homomorphic properties as in
BDOZ/SPDZ \cite{10.1007/978-3-030-68869-1_3}. If cheating is detected,
the protocol is aborted. Note that abort is an explicit action, not
semantic ill-formedness (disallowed by sanity conditions). To model
this, we extend $\minifed$ with an \ttt{assert} command:
$$
    \begin{array}{rcl@{\hspace{2mm}}r}
      \prog &::=& \cdots \mid \elab{\assert{\be = \be}}{\cid}
    \end{array}
$$
As we will discuss, this form of asserting equality between expressions
is sufficient to capture MAC checking and support integrity in protocols.
The big-step adversarial semantics relation $\aredx$ is then defined
in Figure \ref{fig-minicat-aredx} as a mapping from configurations to
a memory $\store$ or $\bot$. Note that any corrupted party is free to
ignore any failing $\ttt{assert}$, in addition to rewriting their
own protocol steps.

\subsubsection{Protocol Runs}

Our semantics require that random tapes contain values for all program
values $\elab{\flip{w}}{\cid}$ sampled from a uniform distribution
over $\mathbb{F}_p$. Input memories also contain input secret values
and possibly other initial view elements as a result of
pre-processing, e.g., Beaver triples for efficient multiplication,
and/or MACed share distributions as in BDOZ/SPDZ
\cite{evans2018pragmatic,10.1007/978-3-030-68869-1_3}. We therefore
posit a predicate $\preproc$ on memories that specifies relevant
preconditions, e.g., that it contains Beaver triples with
associated properties, or pre-distributed keys for MAC checking.

We define $\runs(\prog)$ as the set of final memories resulting from
execution of $\prog$ given any initial memory, possibly augmented with
a preprocessed memory (a memory that satisfies a predicate
$\preproc$), and $\aruns(\prog)$ is the analog in the malicious
setting. It is necessary to define $\aruns$ independently since
adversaries can rewrite their own protocol steps and thus affect
protocol semantics. Since we constrain programs to not overwrite any
variable, we are assured that final memories contain both a complete
record of all initial secrets as well as views resulting from
communicated information. This formulation will be fundamental to our
consideration of probability distributions and hyperproperties of
protocols.
\begin{definition}
  \label{def-runs}
  Given $\prog$ with $\secrets(\prog) = S$ and $\flips(\prog) = R$ and
  predicate $\preproc$ on memories, define:
  $$
  \begin{array}{c}
    \runs(\prog) \defeq 
    \{ \store \mid \exists \store_1 \in \mems(S \cup R) . 
    \exists \store_2 . \preproc(\store_2) \wedge
    (\store_1 \uplus \store_2,\prog) \redxs \store \}
  \end{array}
  $$
  And similarly for any adversary $\adversary$ we define adversarial runs
  as follows:
  $$
  \begin{array}{c}
    \aruns(\prog) \defeq 
    \{ \store \mid \exists \store_1 \in \mems(S \cup R) . 
    \exists \store_2 . \preproc(\store_2) \wedge
    (\store_1 \uplus \store_2,\prog) \aredxs \store \}
  \end{array}
  $$
\end{definition}

\section{Security Model}
\label{section-model}

MPC protocols are intended to implement some \emph{ideal
functionality} $\idealf$ with per-client outputs. In the $\minifed$
setting, given a protocol $\prog$ that implements $\idealf$, with
$\iov(\prog) = S \cup M \cup P \cup O$, the domain of $\idealf$ is
$\mems(S)$ and its range is $\mems(O)$. Security in the MPC setting
means that, given $\store \in \mems(S)$, a secure protocol $\prog$
does not reveal any more information about honest secrets $\store_H$
to parties in $C$ beyond what is implicitly declassified by
$\idealf(\sigma)$. Security comes in \emph{passive} and
\emph{malicious} flavors, wherein the adversary either follows the
rules or not, respectively. Characterization of both real world
protocol execution and simulation is defined probabilistically.  In
this work we will focus on the enforcement of hyperproperties of
passive and malicious security developed in related work
\cite{skalka-near-ppdp24}, that were shown to be sound for real/ideal
security properties that are traditionally used for MPC.  As in that
work we use probability mass functions to express joint dependencies
between input and output variables, as a metric of information
leakage.

\subsection{Probability Mass Functions (pmfs)} 

Pmfs map \emph{realizations} of variables in a joint
distribution to values in $[0..1]$. For convenience and following
\cite{barthe2019probabilistic,skalka-near-ppdp24} we use memories to
represent realizations, so for example given a pmf $\pmf$ over
variables $\{ \sx{x}{1}, \mx{y}{2} \}$ we write $\pmf(\{
\elab{\secret{x}}{1} \mapsto 0, \elab{\mesg{y}}{2} \mapsto 1 \})$ to
denote the (joint) probability that $\elab{\secret{x}}{1} = 0 \wedge
\elab{\mesg{y}}{2} = 1$. Recall from Section
\ref{section-lang-semantics} that $\uplus$ denotes the combination of
memories, so for example $\{ \elab{\secret{x}}{1} \mapsto 0\} \uplus
\{\elab{\mesg{y}}{2} \mapsto 1 \} = \{ \elab{\secret{x}}{1} \mapsto 0,
\elab{\mesg{y}}{2} \mapsto 1 \}$.
\begin{definition}
  A \emph{probability mass function} $\pmf$ over random variables $X$
  (written $\rvs(\pmf) \defeq X$) is a function
  with domain $\mems(X)$ and range $\mathbb{R}$ such that:
  $$
  \sum_{\store \in \mems(X)} \pmf(\store) \  = \ 1
  $$
\end{definition}
Now, we can define a notion of marginal and conditional
distributions as follows, which are standard for discrete
probability mass functions. 
\begin{definition}
  Given $\pmf$ with $\rvs(\pmf) = X_2$, the \emph{marginal distribution}
  of variables $X_1 \subseteq X_2$ in $\pmf$ is denoted $\margd{\pmf}{X_1}$ and defined as follows:
  $$
  \forall \store \in \mems(X_1) \quad . \quad \margd{\pmf}{X_1}(\store) \defeq
  \sum_{\store' \in \mems(X_2-X_1)} \pmf(\store \uplus \store')
  $$
\end{definition}

\begin{definition}
  Given $\pmf$, the \emph{conditional distribution}
  of $X_1$ given $X_2$ where $X_1 \cup X_2 = \rvs(\pmf)$ 
  is denoted $\condd{\pmf}{X_1}{X_2}$ and defined as follows:
  $$
  \forall \store \in \mems(X_1 \cup X_2)\ .\ 
  \condd{\pmf}{X_1}{X_2}(\store) \defeq
  \begin{cases}
    \begin{array}{ll}
      0 & \text{\ if\ } \margd{\pmf}{X_2}(\store_{X_2}) = 0\\
      \pmf(\store) / \margd{\pmf}{X_2}(\store_{X_2})\ & \text{\ o.w.}
     \end{array}
  \end{cases}
  $$
\end{definition}
The first case applies when the marginal probability that $X_2$ takes
the value assigned by the restricted memory $\store_{X_2}$ is zero;
without this case, the conditional probability would be undefined.
We also define some convenient syntactic sugaring for application of
pmfs to realizations when their domains don't match.
\begin{definition}
  Given $\pmf$ and $\store$ with $\rvs(\pmf) \subseteq \dom(\store)$,
  $\pmf(\store) \defeq \pmf(\store_{\rvs(\pmf)})$.
\end{definition}
As in previous work \cite{darais2019language,barthe2019probabilistic,skalka-near-ppdp24,li2023lilac}, 
probabilistic independence, aka separation, is an important concept and we adopt the
following standard notation to express conditional and unconditional separation:
\begin{definition}
  \label{def-condsep}
  We write $\condsep{\pmf}{X_1}{X_2}{X_3}$ iff for all $\store \in \mems(X_1 \cup X_2 \cup X_2)$
  we have $\condd{\pmf}{X_2 \cup X_3}{X_1}(\store)
  = \condd{\pmf}{X_2}{X_1}(\store) * \condd{\pmf}{X_3}{X_1}(\store)$. If $X_1$ is empty
  we write $\sep{\pmf}{X_2}{X_3}$.
\end{definition}

\subsection{Basic Distribution of a Protocol}

We treat all elements of $\runs(\prog)$ as equally likely.  This
establishes the basic program distribution that can be marginalized
and conditioned to quantify input/output information dependencies.
\begin{definition}
  \label{def-progtt}
  \label{def-progd}
  \label{definition-progd}
  The \emph{basic distribution of $\prog$}, written $\progtt(\prog)$, is
  defined such that for all $\store \in \mems(\vars(\prog))$:
  $$
  \progtt(\prog)(\store) =  1 / |\runs(\prog)| \ \text{if}\ \store \in \runs(\prog), \text{otherwise}\ 0
  $$
  For any $\adversary$ the basic \emph{$\adversary$ distribution of $\prog$}, written
  $\progtt(\prog,\adversary)$, is
  defined such that for all $\store \in \mems(\vars(\prog))$:
  $$
  \progtt(\prog,\adversary)(\store) =  1 / |\botruns(\prog)| \ \text{if}\ \store \in \botruns(\prog), \text{otherwise}\ 0
  $$
  where $\botruns$ pads out undefined views and outputs with $\bot$:
  $$
  \begin{array}{c}
    \botruns(\prog)\ \defeq\ 
    \{ \store\{ x_1 \mapsto \bot, \ldots, x_n \mapsto \bot \} \mid 
    \store \in \aruns(\prog) \wedge \{ x_1,\ldots,x_n\} = \iov(\prog) - \dom(\store) \}
  \end{array}
  $$
\end{definition}

\subsection{Honest and Corrupt Views}

Information about honest secrets can be revealed to corrupt clients
through messages sent from honest to corrupt clients, and through
publicly broadcast information from honest clients. Dually,
corrupt clients can impact protocol integrity through the messages
sent from corrupt to honest clients, and through publicly broadcast information
from corrupt clients. We call the former \emph{corrupt views}, and
the latter \emph{honest views}. We let $V$ range over sets
of views, i.e., subsets of $M \cup P$.
\begin{definition}[Corrupt and Honest Views]
  We let $V$ range over \emph{views} which are sets of messages
  and reveals. Given a program $\prog$ with $\iov(\prog) = S \cup M \cup P \cup O$,
  define $\views(\prog) \defeq M \cup P$. Letting $I_1$, $I_2$ be
  a partitioning of the federation clients, define:
  $$
  \begin{array}{lcl}
    V_{I_1 \rhd I_2} & \defeq
        & \{\ \rvl{w} \mid\ \reveal{w}{\be}{\cid} \in \prog \wedge \cid \in I_1 \ \}\ \cup\ 
    \{\ \elab{\mesg{w}}{\cid}\ \mid\  \eassign{\mesg{w}}{\cid}{\be}{\cid'} \in
           \prog \wedge \cid \in I_2 \wedge \cid' \in I_1 \ \} 
  \end{array}
  $$
  Given a partitioning of the federation $H,C$, we call $\houtputs$
  \emph{corrupt views} and $\cinputs$ \emph{honest views}.
\end{definition}

\subsection{Hyperproperties of Confidentiality and Integrity}

In this section we restate hyperproperties developed in previous work
\cite{skalka-near-ppdp24}. The reader is referred to that paper for a more
thorough discussion of their relation to real/ideal security. In this
paper, we leverage that relation to establish security properties in
our type system.

\subsubsection{Confidentiality}

Since MPC protocols release some information about secrets through
their outputs-- that equal the ideal functionality
$\idealf$ applied to all secret inputs-- they do not enjoy strict
noninterference. Public reveals and protocol outputs are fundamentally
forms of declassification. But the following property of probabilistic
noninterference \emph{conditioned on output} is sound for passive
security. Intuitively, in this definition we can think of
$\condd{\progtt(\prog)}{S_H}{S_C \cup O}$ as the distribution of
honest inputs $S_H$ given the adversarial knowledge of corrupt inputs
$S_C$ and public outputs $O$. This establishes an upper bound on
information that can be leaked to the adversary as determined by the
ideal functionality under the assumption of corrupted parties-- or,
put another way, the adversary's ability to guess $S_H$ knowing $O$
and $S_C$.  We can think of $\condd{\progtt(\prog)}{S_H}{S_C \cup O
  \cup \houtputs}$ as the distribution of $S_H$ given the same corrupt
inputs and public outputs, and additionally the adversarial views
$\houtputs$ communicated during runs of $\prog$. If these
distributions are the same, we can deduce that the protocol reveals no
more information about $S_H$ than the corrupt inputs $S_C$ and public
outputs $O$ alone. More details and a proof that this property is
sound for passive simulator security are reported in previous work
\cite{skalka-near-ppdp24}.
\begin{definition}[Noninterference modulo output]
  \label{definition-NIMO}
  We say that a program $\prog$ with $\iov(\prog) = S \cup V \cup O$
  satisfies \emph{noninterference modulo output}
  iff for all $H$ and $C$ and $\store \in \runs(\prog)$ 
  we have:
  $$
  \condd{\progtt(\prog)}{S_H}{S_C \cup O}(\store) = \condd{\progtt(\prog)}{S_H}{S_C \cup O \cup \houtputs}(\store)
 $$
\end{definition}

In practice, MPC protocols typically satisfy a \emph{gradual release}
property \cite{sabelfeld2009declassification}, where the distribution
of messages exchanged during the protocol are independent of the
distribution of input secrets, with only final reveal and output
stages having dependence on input secrets. Referring to Example
\ref{section-smt-example}, the messages exchanged during the protocol
in this case are all in independent uniform distributions, while the
broadcast reveals $\rvl{1-3}$ used to reconstruct the outputs
$\out{1-3}$ are dependent. Gradual release is strictly weaker than
noninterference. It says that intentional reveals and outputs
(declassifications) are the only point of information escape in the
protocol, but does not specify what information may be revealed.
\begin{definition}
  Given $H,C$, a protocol $\prog$ with $\iov(\prog) = S \cup M \cup P \cup O$
  satisfies \emph{gradual release} iff
  $\sep{\progtt(\prog)}{M_C}{S_H}$.
\end{definition}

Noninterference modulo output directly implies passive simulator
security~\cite{skalka-near-ppdp24}, so a sound analysis for
noninterference modulo output is also a sound analysis for passive
simulator security. However, such an analysis is not necessarily
complete: there exist MPC protocols that satisfy passive simulator
security, but which do not satisfy noninterference modulo output. In
particular, our definition of noninterference modulo output requires
equality of the two distributions, which does not hold in protocols
for which indistinguishability relies on computational hardness
assumptions (this includes protocols based on public-key cryptography,
such as many protocols for oblivious transfer). Our approach is not
capable of proving security for these protocols.

\subsubsection{Integrity}

\emph{Integrity} is an important hyperproperty in security models that
admit malicious adversaries. Consistent with formulations in
deterministic settings, we define protocol integrity as the
preservation of high equivalence (of secrets and views). Intuitively,
this property says that any adversarial strategy either ``mimics'' a
passive strategy with some choice of inputs or causes an abort.
\begin{definition}[Integrity]
  \label{def-integrity}
  We say that a protocol $\prog$ with $\iov(\prog) = S \cup V \cup O$ has
  \emph{integrity} iff for all $H$, $C$, and $\adversary$,
  if $\store \in \aruns(\prog)$ there exists $\store' \in \runs(\prog)$ with
  $\store_{S_H \cup X}' = \store_{S_H \cup X}$ and:
    $$
    \condd{\progtt(\prog,\adversary)}{X}{S_H \cup Y}(\store) =
    \condd{\progtt(\prog)}{X}{S_H \cup Y}(\store')
    $$ 
  where $X \defeq (\houtputs \cup O_H) \cap \dom(\store)$ and $Y = \cinputs \cap \dom(\store)$. 
\end{definition}
Integrity plus noninterference modulo output implies malicious security
in the real/ideal model \cite{skalka-near-ppdp24}, and integrity plus
gradual release is a probabilistic form of robust declassification
\cite{sabelfeld2009declassification}.

\section{$\minicat$ SMT-Based Hoare Logic}
\label{section-smt}

In previous work it was shown that the semantics of $\minicat$ in
$\mathbb{F}_2$ can be implemented with Datalog
\cite{skalka-near-ppdp24}, which is a form of constraint
programming. In this paper we explore using SMT instead which
has several advantages. SMT is a highly studied and optimized
system with a variety of theories. In particular, we can leverage
the cvc5 Bit-Vector theory for verification in $\fieldp{2}$
and extend verification to abritrary prime fields using 
Satisfiability Modulo Finite Fields \cite{SMFF}. As we will
show in subsequent sections, this interpretation supports automated
enforcement of correctness guarantees as well as static type analysis
for enforcing confidentiality and integrity properties.  In this
section we formulate a Hoare Logic for $\minicat$ as a foundation for
protocol verification.

\subsection{Constraint Satisfiability and Entailment}

We introduce the following syntax of SMT-style constraints for
$\minicat$. The language of constraint terms includes field values
and operations ($+$, $-$, and $*$), and also bit operations
(construction $\bv{\cdot}$, summation $\oplus$, and subtraction
$\ominus$). We note that bit vector operations are included \emph{only}
in the constraint language-- they allow us to leverage cvc5 bit vector
operations and constraints, while protocols (circuits) themselves
still manipulate individual bits. We illustrate how this can be used for
verification in Section \ref{section-examples}.
$$
\begin{array}{rclr}
  \phi &::=& v \mid x \mid \phi \fplus \phi \mid \phi \fminus \phi \mid \phi \ftimes \phi
     \mid \bv{\phi,\ldots,\phi} \mid \phi \oplus \phi \mid \phi \ominus \phi & \textit{constraint terms}\\ 
     \eqs &::=& \phi \eop \phi \mid \eqs \wedge \eqs \mid \eqs \vee \eqs \mid \neg\eqs \mid \eqs \impl \eqs
        & \textit{constraints} 
\end{array}
$$
%
While constraint terms $\phi$ are similar to expressions $\be$,
in contrast they can ``cross party lines''. This is needed
to express correctness properties-- for example, in the
Goldreich-Micali-Wigderson (GMW) protocol wire values in circuits are
represented by reconstructive shares \cite{evans2018pragmatic}.  If by
convention shares of values $n$ are represented by $\mesg{n}$ on
clients, then assuming two clients $\setit{1,2}$ the reconstructed
value can be expressed as $\mx{n}{1} + \mx{n}{2}$.  So, while summing
values across clients is disallowed in $\minicat$ protocols, we can
express properties of shared values via constraints.

%
%
%

This language of constraints has an obvious direct interpretation in
SMT, in particular the cvc5 Bit-Vector theory in the case of $\fieldp{2}$ and
Satisfiability Modulo Finite Fields \cite{SMFF} for other prime fields. We can
leverage this to implement a critical \emph{entailment} property,
written $\eqs_1 \models \eqs_2$.  Our entailment relation is based on
satisfiability in a standard sense, where we represent models as
memories $\store$ (mappings from variables to values).
\begin{definition}
  We write $\store \models \eqs$ iff $\store$ satisfies, aka is a model
  of, $\eqs$. We write $\eqs_1 \models
  \eqs_2$ iff  $\store \models E_1$ implies $\store \models
  E_2$ for all $\store$, and note that this relation is a preorder.
\end{definition}
Given this definition, the following Theorem is well-known and a fundamental
technique in SMT to implement our (common) notion of entailment. 
\begin{theorem}
  $\eqs_1 \models \eqs_2$ iff $\eqs_1 \wedge \neg\eqs_2$ is
  not satisfiable.
\end{theorem}

\subsection{Programs as Constraint Systems}
\label{section-smt-toeq}

\toeqfig

A central idea of our approach is that we can interpret any protocol
$\prog$ as a set of equality constraints (denoted $\toeq{\prog}$) and use an SMT
solver to verify properties relevant to correctness, confidentiality,
and integrity. Further, we can leverage entailment for tractability--
we can use annotations to obtain a weakened precondition for relevant properties.
That is, given $\prog$, program annotations or other cues can be used
to find a minimal $\eqs$ with $\toeq{\prog} \models \eqs$ for verifying
correctness and security.

The mapping $\toeq{\cdot}$ from programs $\prog$ to constraints is
defined in Figure \ref{fig-toeq}. The interpretation of OT is omitted
from this figure which is general. For $\mathbb{F}_2$ (using 
$\neg$ to denote negation) the interpretation is:
\begin{mathpar}
  \toeq{\elab{\OT{\be_1}{\cid_1}{\be_2}{\be_3}}{\cid_2}} =
  (\toeq{\elab{\be_1}{\cid_1}} \ftimes \toeq{\elab{\be_3}{\cid_2}}) \fplus
  (\neg\toeq{\elab{\be_1}{\cid_1}} \ftimes \toeq{\elab{\be_2}{\cid_2}}) 
\end{mathpar}
In general we will assume that any preprocessing predicate is defined
via a constraint $\eqspre$, i.e., $\preproc(\store) \iff \store
\models \eqspre$ for all $\store$. Intuitively, a protocol satisfies
a correctness (post)condition $\eqs$ if its SMT interpretation
together with $\eqspre$ entails $\eqs$.


\subsubsection{Example: Correctness of 3-Party Addition}
\label{section-smt-example}
Consider the following 3-party additive secret sharing protocol, where
each party generates two additive secret shares
\cite{evans2018pragmatic} for their respective secrets. Each party
publicly reveals the sums of all their shares, and the sum of those is
the protocol output for each party.  
$$
\begin{array}{c@{\qquad}c}
\begin{array}{lll}
  \elab{\mesg{s1}}{2} &:=& \elab{(\secret{1} \fminus \locflip \fminus \flip{x})}{1} \\ 
  \elab{\mesg{s1}}{3} &:=& \elab{\flip{x}}{1} \\ 
  \elab{\mesg{s2}}{1} &:=& \elab{(\secret{2} \fminus \locflip \fminus \flip{x})}{2} \\ 
  \elab{\mesg{s2}}{3} &:=& \elab{\flip{x}}{2} \\ 
  \elab{\mesg{s3}}{1} &:=& \elab{(\secret{3} \fminus \locflip \fminus \flip{x})}{3} \\ 
  \elab{\mesg{s3}}{2} &:=& \elab{\flip{x}}{3}
\end{array} & 
\begin{array}{lll}
  \rvl{1} &:=& \elab{(\locflip \fplus \mesg{s2} \fplus \mesg{s3})}{1} \\ 
  \rvl{2} &:=& \elab{(\mesg{s1} \fplus \locflip \fplus \mesg{s3})}{2} \\
  \rvl{3} &:=& \elab{(\mesg{s1} \fplus \mesg{s2} \fplus \locflip)}{3} \\
  \out{1} &:=& \elab{(\rvl{1} \fplus \rvl{2} + \rvl{3})}{1}\\
  \out{2} &:=& \elab{(\rvl{1} \fplus \rvl{2} + \rvl{3})}{2}\\
  \out{3} &:=& \elab{(\rvl{1} \fplus \rvl{2} + \rvl{3})}{3}
\end{array}
\end{array}
$$
Letting $\prog$ be this protocol, we can verify correctness
as follows-- that is, the output sum of all shares is the sum of the
three input secrets:
$$
\toeq{\prog}\ \models\ \out{1} \eop \sx{1}{1} \fplus \sx{2}{2} \fplus \sx{3}{3}
$$

\noindent We can establish a logical identity between the models of a program and
its (non-adversarial) runs as follows.   
\begin{theorem}
  \label{theorem-toeq}
  $\sigma \models \eqspre \wedge \toeq{\prog}$ iff $\store \in \runs(\prog)$.
\end{theorem}
\noindent This result follows as a consequence of a more general Lemma
that incorporates a notion of preconditions:
\begin{lemma}
  \label{lemma-toeq}
  Given $\eqs, \prog, \store$ with $\store \models \eqs$.
  Then $(\store,\prog) \redxs \store'$ iff $\store' \models
  \eqs \wedge \toeq{\prog}$.
\end{lemma}

\subsection{$\minicat$ Hoare Logic}

To reason effectively about pre and postconditions we introduce
a Hoare logic for $\minicat$,
where $\htrip{\eqs_1}{\prog}{\eqs_2}$ means that $\eqs_1,\eqs_2$ are
the preconditions and postconditions, respectively, of $\prog$.
A semantic notion of validity for Hoare triples is defined as
follows:

\minicathtripfig

\begin{definition}
$\htrip{\eqs_1}{\prog}{\eqs_2}$ is \emph{valid} iff $\eqs_1 \wedge \toeq{\prog} \models \eqs_2$.
\end{definition}

The following Lemma establishes that this is a meaningful (and
standard) definition. It follows by Lemma \ref{lemma-toeq}
and transitivity of $\models$.
\begin{lemma}
  \label{lemma-hoare}
  Given $\eqs_1, \eqs_2, \prog, \store$ with $\store \models \eqs_1$
  and $\eqs_1 \wedge \toeq{\prog} \models \eqs_2$. Then
  $(\store,\prog) \redxs \store'$ implies $\store' \models
  \eqs_2$.
\end{lemma}

Now we can introduce a validity-preserving derivation system for the
$\minicat$ Hoare logic.  The derivation rules are specified in Figure
\ref{fig-minicathtrip}, where in rule \TirName{Frame} we define
$\avars(\prog)$ as the variables in $\prog$ that are assigned to.  The
rules are mostly standard except for \TirName{Assert} which includes
the expected equality in a precondition, and \TirName{Hyp}.  It is
possible that \TirName{Hyp} is derivable, but in any case it is
validity preserving as we show in the following Lemma. Its main use in
this paper will be to enable an automated Hoare logic check we develop
later (Section \ref{section-metahtrip}).
\begin{lemma}
  \label{lemma-hoare-sound}
  If $\htrip{\eqs_1}{\prog}{\eqs_2}$ is derivable, then it is valid.
\end{lemma}

\begin{proof}
  By induction on the derivation and cases analysis on the last derivation
  step. We only need to consider $\TirName{Hyp}$ since other cases are standard.
  In this case the consequent is of the form
  $\htrip{\eqs_1 \wedge (\eqs_2 \impl \eqs)}{\prog}{\eqs_2 \wedge \eqs}$
  where $\htrip{\eqs_1}{\prog}{\eqs_2}$ is derivable by inversion
  of \TirName{Hyp} and hence valid by the induction hypothesis.
  Thus $\eqs_1 \wedge \toeq{\prog} \models \eqs_2$,
  and we can strengthen the condition on the rhs of the entailment as $
  \eqs_1 \wedge (\eqs_2 \impl \eqs) \wedge \toeq{\prog} \models \eqs_2 
  $,
  so by standard reasoning we have:
  $
  \eqs_1 \wedge (\eqs_2 \impl \eqs) \wedge \toeq{\prog} \models \eqs_2 \wedge (\eqs_2 \impl \eqs)
  $,
  and since $\eqs_2 \wedge (\eqs_2 \impl \eqs)$ iff
  $\eqs_2 \wedge \eqs$, we have
  $\eqs_1 \wedge (\eqs_2 \impl \eqs) \wedge \toeq{\prog} \models \eqs_2 \wedge \eqs$.
\end{proof}

\section{Share Types}
\label{section-sty}

Our approach to static confidentiality analysis is to enforce the use
of sound techniques in semi-homomorphic encryption in protocols.  As
discussed in the example in Section \ref{section-smt-example},
protocols integrate interactive elements, such as the sharing of
one-time-pad (OTP) encrypted secrets, and noninteractive elements,
such as addition (of shares). To comport with this we introduce
distinct judgement forms for interactive and noninteractive protocol
elements. Our type system as presented will assume 2-client protocols,
but this is for clarity of presentation and the analysis can be
generalized to more parties.  So in this Section we define the
federation as $\setit{1,2}$ (client 1 and client 2).

Our confidentiality analysis mainly considers protocols in terms of
their SMT interpretations as described in Section \ref{section-smt}.
We introduce a \emph{share type} language $\sty$, classifying
\emph{share terms}. We represent share terms using standard notation
\cite{10.1007/978-3-030-68869-1_3}. A share term $[\phi_1,\phi_2]$
represents a secret-shared value where client 1 has $\phi_1$ and
client 2 has $\phi_2$, and this shared value can be reconstructed
as a public value $\pubx{\phi}$.
$$
\begin{array}{rcl@{\hspace{10mm}}l}
  \sty &::=& \sharety \mid \pubty \mid \outty & \textit{share types} \\
  \stt &::=& [\phi,\phi] \mid \pubx{\phi} & \textit{share terms} 
\end{array}
$$
\paragraph{Share generation and reconstruction}
A \emph{secret sharing scheme} includes functions
$(\genshares,\recon)$ that generate and reconstruct shares,
respectively. The $\genshares$ function takes some
input $\phi$ and randomness $\rx{x}{\cid}$ and produces
two shares (or as many is needed for the number of clients
in the federation): 
$$
\genshares(\phi,\rx{w}{\cid}) = \phi_1,\phi_2
$$
and:
$$
\recon(\genshares(\phi,\rx{w}{\cid})) = \phi
$$
In order for $\genshares$ to be \emph{sound} with regard to
confidentiality, we require
that both of $\phi_1$ and $\phi_2$ are marginally independent
of $\phi$, and also of any other values in the protocol.
This entails the requirement that $\rx{w}{\cid}$ is
consumed by $\genshares$, e.g., as an OTP.
We formally characterize soundness of $\genshares$
in our protocol model in Definition \ref{def-genshares}.
We emphasize that our type system is agnostic with respect
to the definition of $\genshares$ and $\recon$ modulo
soundness-- we use an additive shares encoding in this work,
but the analysis could apply instead to, e.g., Shamir
sharing.

More generally, when we ascribe the type $[\phi_1,\phi_2] : \sharety$,
this means that $[\phi_1,\phi_2]$ are shares of their reconstruction
$\recon(\phi_1,\phi_2)$, noting that shares can be obtained either by
$\genshares$ or by noninteractive homomorphic protocols.  Likewise,
$\pubx{\phi} : \outty$ means that $\phi$ has been publicly
\emph{revealed} via broadcast shares, and our type system requires
that $\phi \eop \recon(\phi_1, \phi_2)$ with $[\phi_1,\phi_2] :
\sharety$. This implies that the process of revealing $\phi$ reveals
nothing besides $\phi$.  The ascription $\pubx{\phi} : \pubty$ is
similar, with the additional meaning that $\phi$ is a publicly-known
value independent of any client inputs. This includes known constant
values used in protocols, and more interestingly \emph{secure
openings}. Secure openings are commonly used to interactively compute
public values that are masked with OTPs secret-shared with clients in
initialization protocols, as for example in Beaver Triple
multiplication which we discuss in Section \ref{section-examples}.

\stylocalfig

We proceed by formally describing the type analysis, then
discussing an illustrative example, and finally by demonstrating share
type soundness in Theorem \ref{theorem-sty}. This result shows that
our analysis enforces noninterference modulo output, in contrast to
the previous system for $\minicat$ that only enforce a gradual release
property \cite{skalka-near-esop25}. By connecting this analysis with
protocol correctness properties enforced by the Hoare logic discussed
in Section \ref{section-smt}, the output can also be connected to an
intended ideal functionality.

\subsection{Non-Interactive Typing Rules}

Our non-interactive typing rules are presented in Figure
\ref{fig-stylocal}. Judgements are of the form $\Gamma, \eqs \vdash
\stt : \sty$, where $\Gamma$ is a set of \emph{share type bindings} of
the form $\stt : \sty$. In the context of typing a protocol $\prog$
soundness requires that the constraint $\eqs$ must be entailed by
$\toeq{\prog}$. The syntax-directed rules are all known semi-homomorphic operations,
including addition of shares (\TirName{HEAdd}), and both addition and
multiplication of shares with a public value (\TirName{HEAddPub} and
\TirName{HEMultPub}, respectively) \cite{10.1007/978-3-030-68869-1_3}.
These rules also leverage SMT entailment in the \TirName{ShareEntails}
and \TirName{PubEntails} rules, meaning these operations are
allowable up to algebraic equivalence in the context of a protocol.

\subsection{Interactive Typing Rules} 

The interactive share typing rules are defined in Figure
\ref{fig-styinteractive}, with judgements of the form $\eqs \vdash
\Gamma_1,R_1, \stredx \Gamma_2,R_2$.  Each such judgement establishes
a sound type binding in $\Gamma_2$ for interactively-defined
variables-- that is, unicast and broadcast messages exchanged between
clients-- assuming the bindings in $\Gamma_1$. In addition to type
bindings, judgments track randomness used by $\genshares$, allowing us
to use type linearity to enforce one-time usage similar to previous
related work \cite{darais2019language}, where any randomness used in
the typing is not in $R_1$ and added to $R_2$. To this end, in these
rules $\uplus$ is disjoint union.

The \TirName{GenShares} rule generates the binding $[x_1,
  x_2] : \sharety$ if the term is a secret sharing as defined
by $\genshares$, and adds $\rx{w}{\cid}$ to the set of variables
consumed by $\genshares$, ensuring one-time usage. The
\TirName{Reveal} rule generates the binding $\pubx{\recon(\px{w_1},\px{w_2})} :
\outty$ with the requirement that $[\px{w_1},\px{w_2}]$ are typable as
shares. The \TirName{SecOpen} rule is similar except it allows the
assertion $\pubx{\recon(\px{w_1},\px{w_2})} : \pubty$ if we can additionally
show that the shared value is masked by the use of random value
$\rx{w}{\oid}$ as an OTP. Here we are positing a third party $\oid$
that secret-shares $\rx{w}{\oid}$ as part of an initialization
protocol. We illustrate secure openings in practice at more length
in Section \ref{section-examples}, though we note that appropriate
properties of initialization protocols can be encoded in $\eqspre$.

We write $\stredx^*$ to denote the transitive
closure of $\stredx$. For any protocol $\prog$, the goal of the
type analysis is to show that all interactive messages can be classified--
that is, to show that $\eqspre \wedge \toeq{\prog} \vdash
\varnothing,\varnothing \stredx^* \Gamma,R$ where all
interactive messages are classified in $\Gamma$. 

\styinteractivefig

\subsection{Instantiating $\genshares$}

It is well-known that adding or subtracting a sample from a uniform
distribution in a finite field yields a value in a uniform
distribution, meaning that samples can be used as OTPs with
perfect secrecy \cite{barthe2019probabilistic,darais2019language}.
So we can define:
\begin{eqnarray*}
  \genshares(\phi,\rx{w}{\cid}) &\defeq& \phi - \rx{w}{\cid}, \rx{w}{\cid}\\
  \recon(\phi_1,\phi_2) &\defeq& \phi_1 \fplus \phi_2
\end{eqnarray*}
where we require that $\genshares$ consumes the random value
$\rx{w}{\cid}$ as an OTP. 

This approach highlights an important interaction between our type
system and SMT entailment as defined by the $\TirName{ShareEntails}$
rule. In particular there are other methods for encrypting values in MPC
protocols that can be shown to be algebraically equivalent
to this fundamental schema. For example, in Yao's Garbled Circuits
(YGC) in $\mathbb{F}_{2}$, the ``garbler'' represents encrypted wire
values as a random sample (for 0) or its negation (for 1), and shares
its secret input encoding with the ``evaluator'' by using
it as a selection bit. That is, assuming that client 2 is the garbler and
client 1 is the evaluator, we can define:
$$
\xassign{\mx{w}{1}}{\mux{\secret{w}}{\flip{w}}{\neg\flip{w}}}{2}
$$
where for all $\be_1,\be_2,\be_3$ (with $\neg$ denoting negation in $\mathbb{F}_2$):
$$
\mux{\be_1}{\be_2}{\be_3} \defeq (\neg\be_1 \ftimes \be_2) \fplus (\be_1 \ftimes \be_3)
$$
and letting this protocol be $\prog$ the following is valid:
\begin{mathpar}
   \toeq{\prog} \models \mx{w}{1} \eop \sx{w}{2} \fminus \rx{w}{2}
\end{mathpar}
For this reason SMT entailment is also an important enabling technology
for the interactive rules. 

\subsection{Share Typing Example} To illustrate details of the
type system we return to the example of additive sharing in the context
of a 2-client protocol. While the ideal functionality of this simple protocol
precisely reveals input secrets, it serves as an understandable example
of the type system mechanics.
We first focus on initial secret sharing steps--
let $\prog$ be the following protocol in $\fieldp{2}$ (noting
that $-$ is equivalent to xor in $\fieldp{2}$):
$$
\begin{array}{lll}
  \elab{\mesg{s1}}{2} := \elab{(\secret{1} \fminus \locflip)}{1}; \quad
  \elab{\mesg{s2}}{1} := \elab{(\secret{2} \fminus \locflip)}{2} 
\end{array}
$$
By $\TirName{GenShares}$ we can then assert:
$$
\toeq{\prog} \vdash \varnothing,\varnothing \stredx
\setit{[\mx{s2}{1}, \elab{\locflip}{2}] : \sharety,
  [\elab{\locflip}{1}, \mx{s1}{2}] : \sharety}, \setit{\elab{\locflip}{1},\elab{\locflip}{2}}
$$
Suppose the protocol goes on to reveal the negation of the
sum (xor-ing) of the inputs:
$$
\begin{array}{l}
  \px{1} := \elab{\mesg{s2} + \locflip + 1}{1};\quad \px{2} := \elab{\mesg{s1} + \locflip}{2};\\
  \out{1} := \elab{\px{1} + \px{2}}{1};\quad \out{2} := \elab{\px{1} + \px{2}}{2}
\end{array}
$$
Letting this extended protocol be $\prog'$, and letting:
$$
\begin{array}{rcl}
\Gamma &\defeq& \setit{[\mx{s2}{1}, \elab{\locflip}{2}] : \sharety,
  [\elab{\locflip}{1}, \mx{s1}{2}] : \sharety}\\
R &\defeq&  \setit{\elab{\locflip}{1},\elab{\locflip}{2}}
\end{array}
$$
observe we can then obtain share typings for noninteractive components for
addition of shares by \TirName{Shared} and \TirName{HeAdd}:
$$
\Gamma, \toeq{\prog'} \vdash
[\mx{s2}{1} + \elab{\locflip}{1}, \mx{s1}{2} + \elab{\locflip}{2}] : \sharety
$$
and hence for negation by \TirName{HeAddPub}:
$$
\Gamma, \toeq{\prog'} \vdash
[\mx{s2}{1} + \elab{\locflip}{1} + 1, \mx{s1}{2} + \elab{\locflip}{2}] : \sharety
$$
This implies:
$$
\toeq{\prog'} \vdash \varnothing,\varnothing \stredx
\Gamma \cup \setit{\pubx{\px{1} + \px{2}} : \outty}, R 
$$
Hence by \TirName{PubEntails} we have:
$$
\begin{array}{l}
\Gamma \cup \setit{\pubx{\px{1} + \px{2}} : \outty}, \toeq{\prog'} \vdash
\pubx{\out{1}} : \outty \\
\Gamma \cup \setit{\pubx{\px{1} + \px{2}} : \outty}, \toeq{\prog'} \vdash
\pubx{\out{2}} : \outty
\end{array}
$$

\subsection{Share Type Soundness}
\label{section-sty-sound}

Share type soundness (Theorem \ref{theorem-sty}) means that share
typing enforces noninterference modulo output. To prove this we need
to assert probabilistic properties in share typings, including for
functions of protocol variables such as reconstruction of shares. So
we make standard extensions of notation for probability distributions
to express probabilities, and conditional probabilities, of functions
of program variables $\phi$. For this extension we posit that
memories can incorporate variables $x_\phi$ that assign valuations
to functions $\phi$.
\begin{definition}
  We let $\Phi$ range over sets of SMT formulas $\phi$, and define
  $
  \mems(\Phi) \defeq \{ \store \mid \dom(\store) = \setit{\ x_{\phi}\ \mid\ \phi \in \Phi\ } \}
  $.
  Given pmf $\pmf$ and $\Phi$ with $\rvs(\pmf) = \bigcup_{\phi \in \Phi} \dom(\phi)$,
  for all $\store \in \mems(\Phi)$ define:
  $$
  \pmf(\store) \defeq \sum_{\store' \in \mems(\Phi)}
  \begin{cases}
    \pmf(\store') \ \text{if}\ \forall \phi \in \Phi . \store'(\phi) = \store(x_\phi)  \\
    0 \ \text{otherwise}
  \end{cases}
  $$
\end{definition}

\begin{definition}
  Given $\pmf$, the \emph{conditional distribution}
  of $X$ given $\Phi$ where $\bigcup_{\phi \in \Phi} \dom(\phi) \cup X = \rvs(\pmf)$ 
  is denoted $\condd{\pmf}{X}{\Phi}$ and defined as follows: 
  $$
  \forall \store_1 \in \mems(X)\ .\ \forall \store_2 \in \mems(\Phi) \ .\ 
  \condd{\pmf}{X}{\Phi}(\store_1 \uplus \store_2) \defeq
  \begin{cases}
    0 \text{\ if\ } \pmf(\store_2) = 0\\
    \pmf(\store_1 \uplus \store_2) / \pmf(\store_2) \text{\ otherwise}
  \end{cases}
  $$
\end{definition}

Now we define what we mean by soundness of $\genshares$. In essence
we mean perfect secrecy \cite{barthe2019probabilistic} in the context
of protocols where multiple encodings can occur. These
conditions are obviously met by, e.g., ``vanilla'' additive or Shamir
shares which generate shares in uniform marginal distributions. Here
and later we leverage (conditional) probabilistic independence as
characterized in Definition
\ref{def-condsep}.
\begin{definition}
  \label{def-genshares}
  We say that $\genshares$ is \emph{sound} iff given any $\prog$ with
  $\vars(\prog) = S \cup V \cup O \cup R$
  and $\phi$ and $\rx{w}{\cid} \not \in R$ and $\genshares(\phi,\rx{w}{\cid}) =
  \phi_1,\phi_2$ we have for all $a \in \setit{1,2}$:
  $$
  \sep{\progd(\prog)}{(S \cup V \cup O \cup \setit{\phi})}{\setit{\phi_a}}
  $$
\end{definition} 

Next we define a well-formedness property of tuples $\prog,\Gamma,R$
intuitively meaning that all the share types in $\Gamma$ are sound with respect
to $\prog$, and consumed randomness in $R$ as OTPs for initial sharing.
This is a representation invariant that we will show is preserved by
typing, and under closure conditions (Definition \ref{def-sty-closed})
the invariant guarantees confidentiality (Theorem \ref{theorem-OK}).
Property (i) ensures that views do not reveal information about secrets
consistent with noninterference modulo output.
Property (ii) ensures a key property of sharing-- that individual shares
don't reveal more about secrets than their reconstruction does. Property
(iii) ensures that public values are independent of outputs or secrets. Properties
(iv) and (v) ensure that OTPs used for generating shares or for
secure openings, respectively, are included in $R$.
\begin{definition}
  \label{def-OK}
  Given $\prog,\Gamma,R$ where $\iov(\prog) = S \cup M \cup P \cup O$. For all $a,b\in \setit{1,2}$ with $a \ne b$ let:
  \begin{eqnarray*}
    \Phi_b & \defeq & \setit{ \phi_b \mid ([\phi_1,\phi_2] : \sharety) \in \Gamma} \cup
                        \setit{ \phi_a \mid (\pubx{\recon(\phi_1,\phi_2)} : \sty) \in \Gamma} \\
    \Phi^p & \defeq & \setit{ \recon(\phi_1,\phi_2) \mid (\pubx{\recon(\phi_1,\phi_2)} : \pubty) \in \Gamma}\\
    \Phi^o & \defeq & \setit{ \recon(\phi_1,\phi_2) \mid (\pubx{\recon(\phi_1,\phi_2)} : \outty) \in \Gamma}
  \end{eqnarray*}
  Then we write $\prog,\Gamma,R : {OK}$ iff all of the following conditions
  hold:
  \begin{enumerate}[\hspace{5mm}i.]
  \item For all $a,b \in \setit{1,2}$ we have
    $
    \condsep{\progtt(\prog)}{(S_{\setit{b}} \cup \Phi^o)}{\Phi_b}{S_{\setit{a}}}
    $.
    \item For all $a,b \in \setit{1,2}$ and $([\phi_1,\phi_2] : \sharety) \in \Gamma$ we have:
    $$\condsep{\progd(\prog)}{(S_{\setit{b}} \cup  \Phi_b \cup \Phi^o \cup \setit{\recon(\phi_1,\phi_2)})}{\phi_a}{S_{\setit{a}}}$$
  \item $\sep{\progd(\prog)}{(S \cup O)}{\Phi^p}$
  \item  For all $a,b \in \setit{1,2}$ and if
    $([x_1,x_2] : \sharety) \in \Gamma$ with $x_a \in \imesgs{M}{\setit{b}}{\setit{a}}$ then
    $\genshares(\phi,\rx{w}{\cid}) = \phi_1,\phi_2$
    and $\toeq{\prog} \models x_1 \eop \phi_1 \wedge x_2 \eop \phi_2$ 
    for some $\phi$ and $\rx{w}{\cid} \in R$.
  \item For all $\phi \in \Phi^p$ if $\toeq{\prog} \models \phi \eop \phi' + \rx{w}{\cid_{\oid}}$ for some $\phi'$
    and $\rx{w}{\cid_{\oid}}$ then $\rx{w}{\cid_{\oid}} \in R$.
  \end{enumerate} 
\end{definition}

Next we show that, given well-formed typing environments, share type judgements preserve
the important properties (i-ii) in Definitions \ref{def-OK}. This result follows from
known semi-homomorphic properties of non-interactive operations that are syntactically
verified in typing judgements.
\begin{lemma}
  \label{lemma-sty-noninteractive-sound}
  Given $\prog,\Gamma,R : {OK}$ where $\secrets(\prog) = S$. For all $a,b\in \setit{1,2}$ with $a \ne b$ let:
  \begin{eqnarray*}
    \Phi_b & \defeq & \setit{ \phi_b \mid ([\phi_1,\phi_2] : \sharety) \in \Gamma} \cup
                        \setit{ \phi_a \mid (\pubx{\recon(\phi_1,\phi_2)} : \sty) \in \Gamma} \\
    \Phi^o & \defeq & \setit{ \recon(\phi_1,\phi_2) \mid (\pubx{\recon(\phi_1,\phi_2)} : \outty) \in \Gamma}
  \end{eqnarray*} 
  and also suppose $\Gamma, \toeq{\prog} \vdash [\phi_1,\phi_2] : \sharety$.
  Then for all $a,b \in \setit{1,2}$ with $a \ne b$ the following conditions hold:
    \begin{enumerate}[\hspace{5mm}i.]
    \item
      $
      \condsep{\progtt(\prog)}{(S_{\setit{b}} \cup \Phi^o)}{(\Phi_b \cup \setit{\phi_b})}{S_{\setit{a}}}
      $
    \item
       $\condsep{\progtt(\prog)}{(S_{\setit{b}} \cup  \Phi_b \cup \Phi^o \cup \setit{\recon(\phi_1,\phi_2)})}{\phi_a}{S_{\setit{a}}}$
    \end{enumerate}
\end{lemma}
\begin{proof}
  The proof is by induction on the derivation of $\Gamma, \toeq{\prog} \vdash [\phi_1,\phi_2] : \sharety$.
  In the base case where the judgement follows by the \TirName{Shared} rule, the result
  follows trivially by Definition \ref{def-OK}. In cases \TirName{HEAdd}, \TirName{HEAddPub}, and
  \TirName{HEMultPub} the result follows by Definition \ref{def-OK} and known
  semi-homomorphic properties of secret sharing \cite{10.1007/978-3-030-68869-1_3}.
  Case \TirName{ShareEntails} follows by correctness of the SMT interpretation of
  protocols as in Theorem \ref{theorem-toeq}.
\end{proof}

Next, we show that interactive typing rules preserve well-formedness of type
environments. This Lemma follows by Lemma \ref{lemma-sty-noninteractive-sound},
and the security of share generation using OTPs as articulated in
Definition \ref{def-genshares}-- that is, shares
generated by $\genshares$ are in uniform marginal distributions. 
\begin{lemma}
  \label{lemma-OK}
  Given $\prog,\Gamma_1,R_1  : {OK}$
  and $\toeq{\prog}  \vdash \Gamma_1, R_1 \stredx \Gamma_2, R_2$.
  Then $\toeq{\prog}, \Gamma_2, R_2 : {OK}$.
\end{lemma}
\begin{proof}
  The proof follows by case analysis on  $\toeq{\prog}  \vdash \Gamma_1, R_1 \stredx \Gamma_2, R_2$.
  In case \TirName{GenShares} the result follows by Definition \ref{def-genshares},
  which satisfies relevant conditions (i-ii,iv) of Definition \ref{def-OK}. In case
  \TirName{SecOpen} the public value must be reconstructed shares, so Lemma
  \ref{lemma-sty-noninteractive-sound} enforces relevant conditions (i-ii) of Definition
  \ref{def-OK}, and the requirement that the public value be masked with an OTP satisfies conditions
  (iii) and (iv) of Definition \ref{def-OK}. In case \TirName{Reveal} since
  the output value must be reconstructed shares, Lemma \ref{lemma-sty-noninteractive-sound}
  satisfies relevant case (ii) of  Definition \ref{def-OK}. 
\end{proof}

In addition to typing judgements, in order to show that a protocol is
secure, we need to classify all corrupt views as shares, all
public values as secure openings, and all outputs as reconstructed
shares by adding them to $\Gamma$. We capture this formally with the following notion of
type environment closure.
\begin{definition}
  \label{def-sty-closed}
  $\Gamma$ is \emph{closed for $\prog$} iff for all $H,C$,
  for all $x \in \houtputs$ there exists $\phi$ such that either
  $([x,\phi] : \sharety) \in \Gamma$ if $C = \setit{1}$
  or  $([\phi,x] : \sharety) \in \Gamma$ if $C = \setit{2}$,
  and for all $\out{\cid} \in O$ we have $(\pubx{\out{\cid}} : \outty) \in \Gamma$.
\end{definition}
Now, the following is a consequence of Definitions \ref{def-OK} and \ref{def-sty-closed}, and then
our main Theorem \ref{theorem-sty} follows by Lemma \ref{lemma-OK} and Theorem \ref{theorem-OK}.
\begin{theorem}
  \label{theorem-OK}
  If $\prog,\Gamma,R : {OK}$ and $\Gamma$ is closed
  for $\prog$, then $\prog$ satisfies noninterference modulo output. 
\end{theorem}
\begin{proof}
  Let $\iov(\prog) = S \cup V \cup O$ and assume some partitioning $H,C$
  with $H = \setit{a}$ and $C = \setit{b}$ for $a,b \in \setit{1,2}$.
  Let $\Phi_b$ and $\Phi^o$ be as defined in Definition \ref{def-OK}. 
  By Definition \ref{def-sty-closed} we have that $x \in \houtputs$
  implies $x \in \Phi_b$ and $(\out{\cid} : \outty) \in \Gamma$ for all
  $\out{\cid} \in O$. 
  Therefore by Definition \ref{def-OK} (and marginalizing away
  extraneous $(\stt : \sty) \in \Gamma$) we have:
  $$
  \condsep{\progd(\prog)}{(S_C \cup \Phi^o \cup O)}{\houtputs}{S_H}
  $$
  By definition of $\TirName{Output}$ it is also the case that for all
  $\out{\cid} \in O$ we have $\toeq{\prog} \models o \eop \phi$ for some
  $\phi \in \Phi^o$. Therefore since we can assume wlog that all revealed
  values are assigned to outputs, we have
  $
  \condsep{\progd(\prog)}{(S_C \cup O)}{\houtputs}{S_H}
  $,
  so by definition we have that for all $\store \in \runs(\prog)$: 
  $$
  \condd{\progd(\prog)}{S_H}{S_C \cup O}(\store) = \condd{\progd(\prog)}{S_H}{S_C \cup \houtputs \cup O}(\store)
  $$
  which is noninterference modulo output by Definition \ref{definition-NIMO}.
\end{proof}

\begin{theorem}
  \label{theorem-sty}
  If $\toeq{\prog} \vdash \varnothing,\varnothing \stredx \Gamma,R$ and $\Gamma$ is closed
  for $\prog$, then $\prog$ satisfies noninterference modulo output. 
\end{theorem}

\begin{proof}
  By Lemma \ref{lemma-OK} we have that $\prog,\Gamma,R : {OK}$, so the result follows
  by Theorem \ref{theorem-OK}.
\end{proof}

\section{Integrity Types}
\label{section-ipj}

Our integrity type system is essentially a taint analysis, with
machinery to allow consideration of any $H,C$ partitioning.  We also
include rules to consider MAC codes to prevent cheating. We consider a
semi-homomorphic encryption scheme that is used in BDOZ and SPDZ
\cite{SPDZ1,SPDZ2,BDOZ,10.1007/978-3-030-68869-1_3} and is applicable
to various species of circuits. However other MAC schemes could be
considered with a modified derivation rule incorporated with our taint
analysis framework. BDOZ/SPDZ assumes a pre-processing phase that
distributes shares with MAC codes to clients.  This integrates well
with pre-processed Beaver Triples to implement malicious secure, and
relatively efficient, multiplication \cite{evans2018pragmatic}.

In BDOZ, a field value $v$ is secret shared among 2 clients 
with accompanying MAC values.  Each client $\cid$ gets a pair of
values $(v_\cid,m_\cid)$, where $v_\cid$ is a share of $v$
reconstructed by addition, i.e., $v = \fcod{v_1 \fplus v_2}$, and
$m_\cid$ is a MAC of $v_\cid$.  More precisely, $m_\cid = k +
(k_\delta * v_\cid)$ where \emph{keys} $k$ and $k_\delta$ are known
only to $\cid' \ne \cid$. The \emph{local key} $k$ is
unique per MAC, while the \emph{global key} $k_\delta$ is common to
all MACs authenticated by $\cid'$. This is a semi-homomorphic
encryption scheme that supports addition of shares and multiplication
of shares by a constant-- for more details the reader is referred to
Orsini \cite{10.1007/978-3-030-68869-1_3}. We note that while we
restrict values $v$, $m$, and $k$ to the same field in this
presentation for simplicity, in general $m$ and $k$ can be in
extensions of $\mathbb{F}_p$.

\ipjfig

As in previous work \cite{skalka-near-ppdp24}, we use
$\macgv{\mesg{w}}$ to refer to secret-shared values reconstructed with
addition, and by convention each share of $\macgv{\mesg{w}}$ is
represented as $\elab{\mesg{w\ttt{s}}}{\cid}$, the MAC of which is
represented as a $\elab{\mesg{w\ttt{m}}}{\cid}$ for all $\cid$, and
each client holds a key $\elab{\mesg{w\ttt{k}}}{\cid}$ for
authentication of the other's share. Each client also holds a global
key $\mesg{\ttt{delta}}$. For any such share identified by string $w$,
the BDOZ MAC scheme is defined by the equality predicate
$\macbdoz{w}$:
$$
\macbdoz{w} \defeq
\mesg{w\ttt{m}} = \mesg{w\ttt{k}} \fplus \ttt{(}\mesg{\ttt{delta}} \ftimes
\mesg{w\ttt{s}}\ttt{)}
$$

\subsection{Typing and Integrity Labeling}

The integrity of values during protocol executions depends on the
partitioning $H,C$, since any value received from a member of $C$ is
initially considered to be low integrity. But as for confidentiality
types, we want our basic type analysis to be generalizable to
arbitrary $H,C$. So integrity type judgements, defined in Figure
\ref{fig-ipj}, are of the form $\ipj{\eqs}{\prog}{\Delta}$, where
$\eqs$ plays a similar role as in confidentiality types, as a possibly
weakened expression of algebraic properties of $\prog$, and
environments $\Delta$ record dependency information that can be
specialized to a given $H,C$. Integrity environments are ordered lists
since order of evaluation is important for tracking integrity--
validating low integrity information through MAC checking allows
subsequent treatment of it as high integrity, but not before. Note
that MAC checking is \emph{not} a form of endorsement as it is usually
defined (as a dual of declassification
\cite{sabelfeld2009declassification}), since a successful check
guarantees no deviation from the protocol.

The typing of protocols depends on a type judgement for expressions
of the form $\itj{\cid}{\be}{V}$, which records
messages and reveals $V$ used in the construction of $\be$.
This source is recorded in $\Delta$ for message sends-- so for
example, for some string $w$ letting:
$$
\prog\ \defeq\ (\xassign{\mx{w}{1}}{\secret{w} - \flip{w}}{2};\ \rvl{w} := \mx{w}{1})
$$
the following judgement is valid:
$$
\ipj{\toeq{\prog}}{\prog}{(\mx{w}{1} : \ity{2}{\varnothing}; \rvl{w} : \ity{1}{\mx{w}{1}})}
$$
More formally, we define validity as follows:
\begin{definition}
  Given pre-processing predicate $\eqspre$ and protocol $\prog$, 
  we say $\ipj{\eqs}{\prog}{\Delta}$ is \emph{valid}
  iff it is derivable and $\eqspre \wedge \toeq{\prog} \models E$.
\end{definition}

In the above example, note that the environment implicitly records
that $\rvl{w}$ depends on data from client 1, which in turn
depends on data from client 2. The rule for MAC
checking allows an overwrite of the dependency. We explore
the application of this in more detail in Section
\ref{section-examples}.

\cheatjfig

\subsection{Assigning Integrity Labels}

In all cases, we want to determine the integrity level of computed
values based on any partitioning $H,C$. In the previous example,
we want to say that $\rvl{w}$ has low integrity if either client 1 or
2 is corrupt.  To this end we define \emph{security labelings}
$\seclev$ which are mappings from variables $x$ to security levels
$\latel$. By default, given a partitioning $H,C$ this labeling maps
messages to the integrity level of the owner.
\begin{definition}  
  Given $H,C$,
  define $\seclev_{H,C}$ such that $\seclev_{H,C}(\mx{w}{\cid} ) = \hilab$
  if $\cid \in H$  and $\lolab$
  otherwise for all $\mx{w}{\cid}$.
\end{definition}
But the values of messages and reveals may be affected by values sent
by other parties, which is recorded in the $\Delta$ obtained from its
typing. To determine whether this impacts integrity for any given
$H,C$, we apply the inductively defined rewrite relation
$\cheatj{\Delta}{H,C}{\seclev}$, axiomatized in Figure
\ref{fig-cheatj}. The composition of typing and this latter rewriting
obtains essentially a traditional taint analysis, modulo MAC checking.
Our main type soundness result for integrity is then given as follows.
Proof details are available in the appendix of an extended arXiv
version of previous work \cite{skalka-near-esop25-arxiv}.
\begin{theorem}
  \label{theorem-ipj}
  Given  $\eqspre$ and $\prog$ with
  $\iov(\prog) = S \cup V \cup O$, if
  $\ipj{\eqs}{\prog}{\Delta}$ is valid
  and for all $H,C$ with $\cheatj{\Delta}{H,C}{\seclev}$ 
  we have $\seclev(x) = \hilab$ for all $x \in \houtputs \cup O_H$, then $\prog$
  has integrity for all $H,C$ and $\adversary$.
\end{theorem}

\section{The $\metaprot$ Language and Automated Hoare Logic}
\label{section-metalang}

The $\metaprot$ language \cite{skalka-near-ppdp24} includes structured
data and function definitions for defining composable protocol
elements at a higher level of abstraction than $\minicat$.  The
$\metaprot$ language is a \emph{metalanguage} aka metaprogramming
language, where an $\fedprot$ protocol is the result of
computation. In addition to these declarative benefits of $\metaprot$,
component definitions support compositional verification of larger
protocols. This includes an automated verification of algebraic
program properties expressed as Hoare logic assertions that we
introduce in this Section. Furthermore, we can compositionally verify
confidentiality and integrity properties expressed as types, as we
will discuss in Sections \ref{section-metaty} and especially
\ref{section-examples} where we present and discuss
extended examples of $\metaprot$ programs.

\metaprotsyntaxfig

\subsection{Syntax and Semantics}

The syntax of $\metaprot$ is defined in Figure
\ref{fig-metaprotsyntax}.  It includes a syntax of values $\mv$ that
include client ids $\cid$, identifier strings $w$, expressions $\be$
in field $\mathbb{F}_p$, record values, and $\minicat$ variables
$x$. $\metaprot$ \emph{expression} forms allow dynamic construction of these
values. Note in particular that $\concat$ is string concatenation.
$\metaprot$ \emph{instructions} allow dynamic
construction of $\minicat$ protocols $\prog$ that incorporate expression
evaluation. The syntax also supports definitions of functions that
compute values $\mv$ and, as a distinct form, functions that compute
protocols $\prog$.  The latter requires type annotations on arguments
which are restricted to being either strings, client ids, or
records-- we use this type information to support Hoare logic
automation as we discuss below in Section \ref{section-metahtrip}.
While type inference for this simple type system would be possible,
it would complicate the novel theory and implementation elements
of our presentation, so we just check annotations instead.

Formally, we consider a complete metaprogram to include
both a codebase and a ``main'' program that uses the codebase.
For all functions (both the expression and instruction forms) we disallow
recursion, which allows us to guarantee the decidability
of type dependence (Section \ref{section-metaty}).
\begin{definition}
A \emph{codebase} $\codebase$ is a list of function 
declarations. We write $ \codebase(f) = y_1,\ldots,y_n,\ e$
if $f(y_1,\ldots,y_n) \{ e \} \in \codebase$, and we
write  $ \codebase(f) = y_1 : \tau_1, \ldots, y_n : \tau_n, \instr$
if $f(y_1 : \tau_1, \ldots, y_n : \tau_n) \{ \instr \} \in \codebase$.
\end{definition}

\metaprotexprsemanticsfig

We define a big-step evaluation relation $\redx$ in Figures
\ref{fig-metaprotexprsemantics} and \ref{fig-metaprotinstrsemantics}
for expressions and instructions, respectively.  In this definition we
write $e[\mv/y]$ and $\instr[\mv/y]$ to denote the substitution of
$\mv$ for free occurrences of $y$ in $e$ or $\instr$ respectively. The
rules are mostly standard.


\paragraph{Sanity conditions on programs.} We assume that
any $\metaprot$ codebase $\codebase$ includes a function
$\ttt{main}()$. And as for $\minicat$, we assume that $\metaprot$
programs are semantically well-formed-- that is, that
$\ttt{main}() \redx \prog$ and $\prog$ is well-formed.

\metaprotinstrsemanticsfig

\subsection{Abstract pre and postcondition annotations}

$\metaprot$ instruction functions can optionally be annotated with pre
and postconditions. To model this in our theory we posit functions
$\precond$ and $\postcond$ that specify conditions for instruction
functions $f \in \dom(\codebase)$.  Note that $\precond(f)$ or
$\precond(f)$ may be undefined if no annotation is given for $f$. In
actual code, we write pre and postconditions directly before and after
functions to define them. Here is a simple example function
$\ttt{bar}$ with pre and postcondition annotations. Note that the
annotations are abstract with respect to the input string $\ttt{x}$,
and that string concatenation can be used in pre and postcondition
constraints:
\begin{verbatimtab}
precondition: ( m[x++"foo"]@1 == 3 )
bar(x:string){ m[x++"bar"]@2 := (m[x++"foo"] + m[x++"baz"])@1 }
postcondition: ( [m[z++"bar"]@2 == m[x++"baz"]@1 + 3 )
\end{verbatimtab}
Since expressions can occur in $\metaprot$ constraints, they are not
SMT constraints, but abstractions of constraints parameterized by
the formal parameters of the function they annotate. We denote these
\emph{abstract constraints} as $\peq$ with the syntax given
in Figure \ref{fig-peq}. When the function is applied, the pre and postconditions
are instantiated as SMT constraints, via evaluating the expressions
within the constraints. The evaluation semantics of abstract constraints
are also given in Figure \ref{fig-peq}.

\peqfig

In the case of \ttt{main}, since it has no arguments, the
pre and postconditions must be concrete SMT expressions.
For example, given the \ttt{bar} function above, we
could define a full program as follows, where \ttt{T}
is the true constraint. 
\begin{verbatimtab}
precondition: ( T )
main(){
  m["thisfoo"]@1 := 3; m["thisbaz"]@1 := 2; bar("this")
}
postcondition: ( [m["thisbar"]@2 == 5 )
\end{verbatimtab}
There are no preconditions to this \ttt{main} because they are
all established by the program, and the postcondition follows
as a consequence of the application \ttt{bar("this")}.

The use of pre and postconditions in \ttt{main} allows verification of
whole-program (protocol) correctness. Further, the use of annotations
on function definitions supports \emph{compositional} verification of
program components. In general postconditions can serve as
pre-verified properties of functions that can be reused in, and reduce
overhead of, whole protocol verification. A representative example of
GMW $\ttt{and}$ gates \cite{evans2018pragmatic} implementations is
given in Figure \ref{fig-gmwprepost}. Here, the $\ttt{xorgmw}$ and
$\ttt{andgmw}$ functions takes as input strings $x$ and $y$ that
distinguish values secret-shared between clients 1 and 2 as
$\mx{x}{1}, \mx{x}{2}$ and $\mx{y}{1},\mx{y}{2}$ respectively-- that
is, these shares encode $\ttt{RECON}(x)$ and $\ttt{RECON}(y)$.  The
postcondition of the non-interactive $\ttt{xorgmw}$ protocol specifies
that it results in a value shared between clients 1 and 2 as
$\ttt{RECON}(z)$ that is the xor-ing of $\ttt{RECON}(x)$ and
$\ttt{RECON}(y)$. The postcondition of the interactive protocol
defined by the $\ttt{andgmw}$ specifies that it results in a value
shared between clients 1 and 2 as $\ttt{RECON}(z)$ that is the
conjunction of $\ttt{RECON}(x)$ and $\ttt{RECON}(y)$. The
postcondition for $\ttt{andgmw}$ especially is non-trivial to prove,
and expressing these postconditions in a compositional manner reduces
proof overhead and improves scalability to large circuits that may
incorporate a multitude of $\ttt{and}$ gates.

\begin{fpfig}[t]{GMW $\ttt{and}$ gates with postcondition.}{fig-gmwprepost}
\begin{verbatimtab}
RECON(x) { (m[x]@1 + m[x]@2) }

xorgmw(z:string, x:string , y:string) {
  m[z]@1 := (m[x] + m[y])@1;
  m[z]@2 := (m[x] + m[y])@2
}
postcondition:  ( RECON(z) == RECON(x) + RECON(y))

andtablegmw(x:string, y:string , z:string) {
  let r11 = (r[z] + (m[x] + 1) * (m[y] + 1)) in
  let r10 = (r[z] + (m[x] + 1) * (m[y] + 0)) in
  let r01 = (r[z] + (m[x] + 0) * (m[y] + 1)) in
  let r00 = (r[z] + (m[x] + 0) * (m[y] + 0)) in
  { row1 = r11; row2 = r10; row3 = r01; row4 = r00 }
}

andgmw(z:string, x:string, y:string) {
  let table = andtablegmw(x,y,z) in
  m[z]@2 := OT4((m[x], m[y])@2, table.row4, table.row3, table.row2, table.row1)@1;
  m[z]@1 := r[z]@1
}
postcondition:  ( RECON(z) == RECON(x) * RECON(y))
\end{verbatimtab}
\end{fpfig}

\subsection{$\metaprot$ Automated Abstract Hoare Logic}
\label{section-metahtrip}

In the $\metaprot$ theory, we introduce Hoare logic assertions
of the form $\htrip{\peq_1}{\cmd}{\peq_2}$. The correctness
condition of this assertion is compositional, in that instantiated
abstract Hoare triples evaluate to valid $\minicat$ Hoare triples.
This is formally elaborated below in Lemma
\ref{lemma-metahtrip} and Corollaries \ref{cor-metahtrip} and
\ref{cor-mainhtrip} below.

In Figure \ref{fig-metahtrip} we give the derivation rules for the
abstract Hoare logic. These are based on $\minicat$ Hoare logic, but
are algorithmic. A main predicate for these rules is the
$\TirName{GenEntails}$ rule, which encodes a generalized version of
entailment. By using fresh variables to instantiate abstract
constraints, we guarantee entailment for any instantiation, as
follows.

\begin{lemma}
  \label{lemma-fresh}
  Given $\eqs$ and $w, \cid = \fresh(\stringty, \cidty)$ and
  some $w',\cid'$ where 
  $\peq[w/y_1,\cid/y_n] \redx \eqs$ and
  $\peq[w'/y_1,\cid'/y_n] \redx \eqs'$.
  Then $\store \models \eqs'$ implies $\store \models \eqs$. 
\end{lemma}

\begin{proof}
  If $w',\cid'$ are fresh or chosen the result is
  trivial.  Otherwise, it must be the case that at least two
  distinct variables $x_1$ and $x_2$ in $\eqs$ are renamed as a single
  variable $x$ in $\eqs'$, given freshness of $w,\cid$. Let
  $\mathit{rename}$ be the function that maps variables in $\eqs$ to
  variables in $\eqs'$ such that $\mathit{rename}(\eqs) = \eqs'$. Suppose
  on the contrary that 
  $\store' \models \eqs'$ and $\store' \not\models \eqs$, and define $\store$ such that:
  $$
  \forall x \in \dom(\vars(E)) . \store(x) \defeq \store'(\mathit{rename}(x)) 
  $$
  Then $\store$ must be a model of $\eqs$, which is a contradiction. 
\end{proof}

\begin{lemma}
  \label{lemma-genentails}
  If $\forall y_1:\tau_1,\ldots,y_n:\tau_n . \peq_1 \models \peq_2$ and
  $\peq_1[\mv_1/y_1]\cdots[\mv_n/y_n] \redx \eqs_1$ and
  $\peq_2[\mv_1/y_1]\cdots[\mv_n/y_n] \redx \eqs_2$, then
  $\eqs_1 \models \eqs_2$.  
\end{lemma}
\begin{proof}
  By assumption and definition of \TirName{GenEntails} we have
  $\peq_1[\mv_1'/y_n \cdots \mv_n'/y_n] \redx \eqs_1'$ and
  $\peq_2[\mv_1'/y_n \cdots \mv_n'/y_n] \redx \eqs_2'$ and $\eqs_1'
  \models \eqs_2'$ for some $\mv_1',\ldots,\mv_n' =
  \fresh(\tau_1,\ldots,\tau_n)$.  Now suppose on the contrary that
  $\eqs_1 \not\models \eqs_2$.  This implies the existence of $\store$
  such that $\store \models \eqs_1$ and $\store \not\models \eqs_2$.
  But Lemma \ref{lemma-fresh} implies that $\store \models \eqs_1'$,
  so by definition we have $\store \models \eqs_2'$ implying $\store
  \models \eqs_2$ by construction, which is a contradiction.
\end{proof}

\metahtripfig

The \TirName{HardPack} rule uses \TirName{GenEntails} to ensure that
the annotated pre and postconditions are correct for a given function,
if they are defined. It is in fact not necessary for \emph{any} pre
and postconditions to be defined, since the analysis will infer valid
pre and postconditions in their absence by rule \TirName{Fn}, but
their intent is to enable fine-tuning of pre and postconditions for
improving verification efficiency as discussed above and later in
Section \ref{section-examples}.  An important implementation detail is
that verification of any given function $f$ only needs to be performed
once, and can then be memoized for other applications using the
\TirName{App} rule as a consequence of Lemma \ref{lemma-genentails}
(and more precisely of Corollary \ref{cor-metahtrip}).

The \TirName{Assign} and \TirName{Assert} rules propagate equations
induced by assignments and assertions to post and preconditions
respectively.  Note that without the use of \TirName{HardPack}, the
entire protocol would accumulate in postconditions of \ttt{main}. The
\TirName{Seq} rule relies on the correctness of our \TirName{Hyp} rule
to join pre and postconditions as we discuss in the relevant case of
Lemma \ref{lemma-metahtrip}, which is the foundational correctness
result for the analysis.

\paragraph{Pre and post sanity conditions.} We require that
annotated pre and postconditions on functions will evaluate to SMT
constraints when instantiated with function arguments, and refer only
to variables that occur in their bodies. Deduction rules other than
\TirName{HardPack} enforce this by definition.  For any $f$ with
$\codebase(f) = y_1 : \tau_1, \ldots, y_n : \tau_n, \cmd$,
$\precond(f) = \peq_1$, and $\postcond(f) = \peq_2$, we require that
for all $[\mv_1/y_1]\cdots[\mv_n/y_n]$ with $\mv_1 : \tau_1, \ldots
\mv_n : \tau_n$ there exists $\eqs_1$, $\eqs_2$
where:
\begin{mathpar}
  \cmd[\mv_1/y_1]\cdots[\mv_n/y_n] \redx \prog
  
  \peq_1[\mv_1/y_1]\cdots[\mv_n/y_n] \redx \eqs_1
  
  \peq_2[\mv_1/y_1]\cdots[\mv_n/y_n] \redx \eqs_2
\end{mathpar}
and $\vars(\eqs_1) \cup \vars(\eqs_2) \subseteq \vars(\prog)$.

\subsection{Abstact Hoare Logic Soundness}

We begin by showing that the analysis preserves the property that in
any $\htrip{\peq_1}{\prog}{\peq_2}$, fully instantiating $\prog$
also fully instantiates $\peq_1$ and $\peq_2$. The result
follows by sanity conditions on pre and postcondition annotations,
the definition of analysis, and induction on $\cmd$.
\begin{lemma}
  \label{lemma-closure}
  If $\htrip{\peq_1}{\prog}{\peq_2}$ and $\instr[\mv_1/y_1]\cdots[\mv_n/y_n] \redx \prog$
  then $\peq_1[\mv_1/y_1]\cdots[\mv_n/y_n] \redx \eqs_1$ and
  $\peq_2[\mv_1/y_1]\cdots[\mv_n/y_n] \redx \eqs_2$ for some $\eqs_1, \eqs_2$. 
\end{lemma}

Now we can prove our main result in a general form. The proof
follows mainly by appeal to the $\minicat$ Hoare logic deduction
system.
\begin{lemma}[$\metaprot$ Hoare Triple Correctness]
  \label{lemma-metahtrip}
  If $\htrip{\peq_1}{\cmd}{\peq_2}$ with $\peq_1 \redx
  \eqs_1$, $\peq_2 \redx \eqs_2$, and $\cmd \redx
  \prog$, then $\htrip{\eqs_1}{\prog}{\eqs_2}$.
\end{lemma}

\begin{proof}
  By structural induction on the derivation of $\cmd \redx \prog$ and
  case analysis on $\cmd$.

  \textit{Case} $\cmd = (\xassign{e_1}{e_2}{e_3})$. In this case by definition
  and inversion of $\redx$ and $\TirName{Assign}$ we have:
  \begin{mathpar}
    e_1 \redx x

    e_2 \redx \be

    e_3 \redx \cid

    \prog = (\xassign{x}{\be}{\cid})

    \peq_1 \redx \true

    \peq_2  \redx x \eop \toeq{\elab{\be}{\cid}}
  \end{mathpar}
  Now by \TirName{Assign} we have
  $\htrip{\toeq{\elab{\be}{\cid}} \eop \toeq{\elab{\be}{\cid}}}{\xassign{x}{\be}{\cid}}{x \eop \toeq{\elab{\be}{\cid}}}$, so
  we have
  $\htrip{\true}{\xassign{x}{\be}{\cid}}{x \eop
    \toeq{\elab{\be}{\cid}}}$ by \TirName{Consequence}.

  \textit{Case} $\cmd = \elab{\assert{e_1 = e_2}}{e_3}$. In this case by definition we have:
  \begin{mathpar}
    e_1 \redx \be_1

    e_2 \redx \be_2

    e_3 \redx \cid

    \prog = \elab{\assert{\be_1 = \be_2}}{\cid}

    \peq_1 \redx \toeq{\elab{\be_1}{\cid}} \eop \toeq{\elab{\be_2}{\cid}}
    
    \peq_2 \redx \true
  \end{mathpar}
  But we have $\htrip{\toeq{\elab{\be_1}{\cid}} \eop \toeq{\elab{\be_2}{\cid}}}{\elab{\assert{\be_1 = \be_2}}{\cid}}{\toeq{\elab{\be_1}{\cid}} \eop \toeq{\elab{\be_2}{\cid}}}$ by $\TirName{Assert}$, so
  $\htrip{\toeq{\elab{\be_1}{\cid}} \eop \toeq{\elab{\be_2}{\cid}}}{\elab{\assert{\be_1 = \be_2}}{\cid}}{\true}$ by $\TirName{Consequence}$.

  \textit{Case} $\cmd = \cmd_1;\cmd_2$. In this case by assumption and
  definition and inversion of $\redx$ and $\TirName{Seq}$ we have:
  \begin{mathpar}
    \cmd_1 \redx \prog_1

    \cmd_2 \redx \prog_2

    \prog = \prog_1;\prog_2

    \peq_1 \redx \eqs_1^1 \wedge (\eqs_2^1 \impl \eqs_1^2)
      
    \peq_2 \redx \eqs_2^1 \wedge \eqs_2^2
  \end{mathpar}
  and by inversion of the $\TirName{Seq}$ rule we have $\htrip{\peq_1^1}{\cmd_1}{\peq_2^1}$ and $\htrip{\peq_1^2}{\cmd_2}{\peq_2^2}$,
  %
  and by the induction hypothesis we have $\htrip{\eqs_1^1}{\cmd_1}{\eqs_2^1}$ and $\htrip{\eqs_1^2}{\cmd_2}{\eqs_2^2}$.
  %
  But by \TirName{Hyp} we have:
  $\htrip{\eqs_1^1 \wedge (\eqs_2^1 \impl \eqs_1^2)}{\cmd_1}{\eqs_2^1 \wedge \eqs_1^2}$
  and by sanity conditions we have
  $\vars(\eqs_2^1) \cap \avars(\prog_2) = \varnothing$ so by \TirName{Frame} we
  have
  $\htrip{\eqs_2^1 \wedge \eqs_1^2}{\cmd_2}{\eqs_2^1 \wedge \eqs_2^2}$.
  Thus by $\TirName{Seq}$ we have:
  $$\htrip{\eqs_1^1 \wedge (\eqs_2^1 \impl \eqs_1^2)}{\prog_1;\prog_2}{\eqs_2^1 \wedge \eqs_2^2}$$

  \textit{Case} $\cmd = f(e_1,\ldots,e_n)$. In this case by assumption and
  definition and inversion of $\TirName{App}$ and $\redx$ we have:
  \begin{mathpar}
    \codebase(f) = y_1,\ldots,y_n, \instr
    
    e_1 \redx \mv_1 \cdots e_n \redx \mv_n

    \instr[\mv_1/y_1]\cdots[\mv_n/y_n] \redx \prog

    \peq_1 = \peq^0_1[\mv_1/y_1]\cdots[\mv_n/y_n]

    \peq_2 = \peq^0_2[\mv_1/y_1]\cdots[\mv_n/y_n]
    
    \htrip{\peq^0_1}{f(y_1:\tau_1,\ldots,y_n:\tau_n)}{\peq^0_2}
  \end{mathpar}
  Note that sanity conditions on pre- and post-condition annotations
  means there exists $\eqs_1$ and $\eqs_2$ where $\peq_1 \redx
  \eqs_1$ and $\peq_2 \redx \eqs_2$. Now, the judgement
  $\htrip{\peq^0_1}{f(y_1:\tau_1,\ldots,y_n:\tau_n)}{\peq^0_2}$
  follows by one or two subcases. Where it follows by \TirName{Fn} the
  result is immediate by definition and the induction
  hypothesis. Otherwise it follows by \TirName{Hardpack} and we
  proceed as follows.

  By inversion of \TirName{Hardpack} we have
  $\htrip{\peq_1'}{f(y_1:\tau_1,\ldots,y_n:\tau_n)}{\peq_2'}$ by
  \TirName{Fn} where $\forall y_1:\tau_1,\ldots,y_n:\tau_n . \peq^0_1
  \models \peq_1'$. Then $\peq_1'[\mv_1/y_1]\cdots[\mv_n/y_n] \redx
  \eqs_1'$ and $\peq_2'[\mv_1/y_1]\cdots[\mv_n/y_n] \redx
  \eqs_2'$ by Lemma \ref{lemma-closure}. 
  By Lemma \ref{lemma-genentails} we have $\eqs_1
  \models \eqs_1'$, so by \TirName{Consequence} we have
  $\htrip{\eqs_1}{\prog}{\eqs_2'}$, and by \TirName{Hyp}
  we have:
  $$\htrip{\eqs_1 \wedge (\eqs_2' \impl \eqs_2)}{\prog}{\eqs_2' \wedge \eqs_2}$$

  Furthermore, also by inversion of \TirName{Hardpack} we have
  $\forall y_1:\tau_1,\ldots,y_n:\tau_n . \peq^0_1 \wedge 
  \peq_2' \models \peq^0_2$. 
  By Lemma \ref{lemma-genentails} we thus have $\eqs_1
  \wedge \eqs_2' \models \eqs_2$, which implies
  $\eqs_1 \models \eqs_2' \impl \eqs_2$.
  Thus by \TirName{Consequence} we have $\htrip{\eqs_1}{\prog}{\eqs_2}$.

  Finally, \textit{case} $\prog = \elet{y}{e}{\cmd}$ follows trivially
  by the induction hypothesis, which covers all cases to prove
  the result.
\end{proof}

As a consequence of the preceding Lemma, we can state important
Corollaries that characterize, first, the compositionality of function
constraints, and second, the correctness of whole program
analysis. These both follow immediately by reasoning in the
\TirName{App} case of Lemma \ref{lemma-metahtrip}.

\begin{corollary}[Compositionality of Abstract Function Constraints]
  \label{cor-metahtrip}
  Given:
  \begin{mathpar}
    \htrip{\peq_1}{f(y_1:\tau_1,\ldots,y_n:\tau_n)}{\peq_2}

    \codebase(f) = y_1 : \tau_1, \ldots, y_n : \tau_n, \cmd

    \cmd[\mv_1/y_1]\cdots[\mv_n/y_n] \redx \prog

    \peq_1[\mv_1/y_1]\cdots[\mv_n/y_n] \redx \eqs_1

    \peq_2[\mv_1/y_1]\cdots[\mv_n/y_n] \redx \eqs_2
  \end{mathpar}
  Then $\htrip{\eqs_1}{\prog}{\eqs_2}$ is valid.
\end{corollary}

\begin{corollary}
  \label{cor-mainhtrip}
  If $\htrip{\peq_1}{\ttt{main()}}{\peq_2}$ and $\ttt{main}() \redx \prog$ and
  $\peq_1 \redx \eqs_1$ and $\peq_2 \redx \eqs_2$ then
  $\htrip{\eqs_1}{\prog}{\eqs_2}$ is valid.
\end{corollary}

\section{$\metaprot$ Type Checking}
\label{sec:metapr-type-check}

\stafig

Now we develop a compositional type system for $\metaprot$ based on 
$\minicat$ typing and the automated Hoare logic developed in
Section \ref{section-metalang}. We focus on share types for confidentiality
which is more complex and general than our integrity type system.
An algorithmic approach to the latter was developed previously for
$\metaprot$ \cite{skalka-near-esop25}, and also the core element
of the system-- BDOZ/SPDZ-style MAC checking-- can be treated as
algebraic correctness conditions enforceable by Hoare triples
as we show in Section \ref{section-examples}.

\subsection{Share Type Checking with Annotations}
\label{section-metaty}

Share types rely on SMT entailments via the \TirName{ShareEntails} and
\TirName{PubEntails} rules, which are not syntax directed, and likewise
the interactive typing judgement is nondeterministic. Our goal
with the $\metaprot$ language is to automate share type checking. To
enable this we introduce \emph{share type annotations}. These
annotations allow programmer specification of properties that can be
verified as part of the Hoare triple analysis defined in Section
\ref{section-metahtrip} and that enable automated checking of types.

A fundamental component of the system is automated share type checking
of noninteractive computations.  We define an algorithmic, syntax
directed subset of noninteractive share typing as follows:
\begin{definition}
  Define $\Gamma \Vdash \stt : \sty$ iff $\Gamma,\true \vdash \stt : \sty$ without the
  use of \TirName{ShareEntails} and \TirName{PubEntails}.
\end{definition}

\shtripstafig

The syntax of type annotations $\sta$ is defined in Figure
\ref{fig-sta}, which we extend the language of commands $\cmd$ to
include. Type annotations are of the form $e_1 \kwas e_2$, which
defines an SMT equivalence-- the key idea here being that explicit
equivalences can stand in for implicit \TirName{ShareEntails} and
\TirName{PubEntails} judgements-- and $\notg{\stt} : \sty$ which
defines a share type assertion verified using the automated $\Vdash$
typing judgement. In the $\metaprot$ semantics type annotations are
interpreted as $\pskip$s, while their integration in Hoare triple
derivations requires annotated equivalences to be included in
preconditions.  In an implementation, Hoare triple verification can be
conducted independently prior to type checking.

\subsubsection{Annotation Type Checking}

Type checking judgements for typing annotations are defined as follows.
The \emph{definitions} element $\xdefs$ of judgements is an environment
of SMT equivalences induced either by the program or by annotation. 
\begin{definition}[$\metaprot$ Type Checking Judgements]
  \emph{Type checking judgements} are of the syntactic form
  $\shtrip{\Gamma, R, \xdefs}{\cmd}{\Gamma, R, \xdefs}$
  where $\xdefs : \phi \rightarrow \phi$ are called
  \emph{definitions}.
\end{definition}

All typing judgements are dependent on closed expressions $e$ which we
equate up to evaluation. Since the semantics of expressions is total,
we simply evaluate expressions in preconditions of derivation
rules. The derivation rules for annotations $\sta$ are given in Figure
\ref{fig-shtripsta}. First, we note that the \TirName{Coerce} rule
extends definitions to include the equivalence specified by the $e_1
\kwas e_2$ annotation form-- we assume the validity of the equivalence
is enforced by Hoare triple verification and so is not checked (again)
here.  These definitions extend or even overwrite variable definitions
established by the protocol itself (see the \TirName{Defn} rule in
Figure \ref{fig-shtripcmd} discussed below).

Programmers can also specify bindings for local computations, for
example if $\mx{x}{1},\mx{x}{2}$ are shares obtained by, e.g., local
homomorphic addition, the programmer can specify
$[\mx{x}{1},\mx{x}{2}] : \sharety$. This supports compositionality of
type checking. Thus, in the rules for share type annotations of the
form $\notg{\stt} : \sty$, we have different rules for checking local and interactive
$\sharety$ type ascriptions (\TirName{LocalShares} and \TirName{GenShares},
respectively), and for checking local and interactive
$\pubty$ type ascriptions (\TirName{LocalPublic} and \TirName{MaskedPublic},
respectively). In these rules we use $\xdefs$ to recover defined
SMT equivalences, and $\Vdash$ to verify share typing.

\subsubsection{Functions and Application}

\shtripcmdfig

Rules for other commands are defined in Figure \ref{fig-shtripcmd}.
As anticipated above, the \TirName{Defn} rule automatically adds the
protocol definition of variables to $\xdefs$ as a default equivalence,
which may allow automatic share typing without any use of
\TirName{Coerce}. The most complex rules are for function application
and definition (\TirName{App} and \TirName{Fn}). These rules allow for
generalization and instantiation of share type environments and
thus compositionality in typing. Generalized aka abstract
share type environments are dependent on $\stringty$ and
$\cidty$ types, and we equate them up to evaluation in
our rules.
\begin{definition}[Abstract Share Type Environments]
  We write $\notg{\Gamma}$ to denote sets of bindings $\notg{\stt} : \sty$, and
  $\notg{R}$ to denote sets of expressions (that evaluate to random variables).
  We define evaluation of these, respectively, as follows:
  \begin{mathpar}
  \inferrule
      {\notg{\stt}_1 \redx \stt_1\ \cdots\ \notg{\stt}_n \redx \stt_n}
      {\setit{\notg{\stt}_1 : \sty_1, \ldots, \notg{\stt}_n : \sty_n} \redx
        \setit{\stt_1 : \sty_1, \ldots, \stt_n : \sty_n}}
      
  \inferrule
      {e_1 \redx \rx{w_1}{\cid_1}\ \cdots\ e_1 \redx \rx{w_n}{\cid_n}}
      {\setit{e_1,\ldots,e_n} \redx
        \setit{\rx{w_1}{\cid_1}, \ldots, \rx{w_n}{\cid_n}}}
  \end{mathpar}
\end{definition}

\noindent We also introduce dependent $\Pi$ types for functions, that allow
dependence on $\stringty$ and $\cidty$ types.
\begin{definition}[$\metaprot$ Function Types]
  Dependent function types are of the syntactic form
  $\Pi y_1,\ldots,y_n. \notg{\Gamma} \rightarrow (\notg{\Gamma}, \notg{R})$.
  where $y_1,\ldots,y_n$ range over $\stringty$ or $\cidty$ types.
\end{definition}
As defined in \TirName{Fn}, a dependent function type $\Pi
y_1,\ldots,y_n. \notg{\Gamma_1} \rightarrow (\notg{\Gamma_2},
\notg{R})$ specifies that, given share typing assumptions in
$\notg{\Gamma_1}$, application of the function results in the bindings
in $\notg{\Gamma_2}$ using the random values in $\notg{R}$ for share
generation or secure openings.  We require that the type is closed,
i.e., only $y_1,\ldots,y_n$ can occur free in $\notg{\Gamma_1}$,
$\notg{\Gamma_2}$ $\notg{R}$.  Type checking for functions is
accomplished by using fresh ``witness'' values of $\stringty$ and
$\cidty$ types to instantiate dependent variables, similarly to the
\TirName{GenEntails} rule. As defined in the \TirName{App} rule,
typing ensures that application of the function instantiates the
dependencies, checks that the instantiation of $\notg{\Gamma_1}$ is
contained in the share typing environment, and adds the instantiation
of $\notg{\Gamma_2}$ as an effect along with instantiation of
$\notg{R}$ with linearity enforced. A practical programming detail is
where the share type environment precondition $\notg{\Gamma_1}$ is
defined by the programmer-- by convention we write these as
additional components of \ttt{precondition} annotations, as
illustrated in the extended examples discussed in Section
\ref{section-examples}.

\subsubsection{Example: GMW \ttt{and} Gate}

Returning to the GMW \ttt{and} gate example introduced in Figure \ref{fig-gmwprepost},
we note that the $\ttt{xorgmw}$ gate needs no type annotations for type checking.
We can extend the definition of \ttt{andgmw} with type annotations to ensure a
valid share typing, as in Figure \ref{fig-gmwsharetypes}. Here a coercion
is needed to specify that the gate output wire values are secret-shared using
$\ttt{r[z]@1}$ as an OTP. The type for $\ttt{andgmw}$ can then be
obtained as:
$$
\ttt{andgmw} : \Pi x,y,z . \varnothing \rightarrow (\setit{[m[z]@1,m[z]@2] : \sharety}, \setit{r[z]@1})  
$$
Note that the function does not require any share type preconditions, since
the (sub)protocol it defines interactively generates new shares using an OTP
($r[z]@1$).

\begin{fpfig}[t]{GMW $\ttt{and}$ gate with share type annotations.}{fig-gmwsharetypes}
\begin{verbatimtab}
andgmw(z:string, x:string, y:string) {
  let table = andtablegmw(x,y,z) in
  m[z]@2 := OT4((m[x], m[y])@2, table.row4, table.row3, table.row2, table.row1)@1;
  m[z]@1 := r[z]@1;
  m[z]@2 as (RECON(x) * RECON(y)) + r[z]@1;
  [m[z]@1,m[z]@2] : shares
}
postcondition: (RECON(z) == RECON(x) * RECON(y))
\end{verbatimtab}
\end{fpfig}

\subsubsection{Soundness of $\metaprot$ Type Checking}

Demonstrating soundness of $\metaprot$ type checking is similar to $\minicat$--
we show the ${OK}$ invariant is preserved by the algorithm. By Theorem \ref{theorem-OK}
soundness follows. Lemma \ref{lemma-metaty} follows by appeal to Lemma \ref{lemma-sty-noninteractive-sound}
since the algorithmic judgement $\Vdash$ is sound with regard to the noninteractive
typing rules by definition. Theorem \ref{theorem-metaty} follows as a special
case of Lemma \ref{lemma-metaty}.
\begin{lemma}
  \label{lemma-metaty}
  Given $\prog_1,\Gamma_1,R_1 : {OK}$ and $\cmd \redx \prog_2$ and $\htrip{\eqs_1}{\cmd}{\eqs_2}$ with
  $\eqs_1$ satisfiable and $\toeq{\prog_1} \models \eqs_1 \wedge \toeq{\xdefs_1}$ and
  $\shtrip{\Gamma_1,R_1,\xdefs_1}{\cmd}{\Gamma_2,R_2,\xdefs_2}$.
  Then $\prog_1;\prog_2,\Gamma_2,R_2 : {OK}$ and $\toeq{\prog_1;\prog_2} \models \toeq{\xdefs_2}$.
\end{lemma}

\begin{theorem}
  \label{theorem-metaty}
  Given $\htrip{\eqspre}{\ttt{main()}}{\eqs}$ and $\ttt{main} : \Gamma_1 \rightarrow (\Gamma_2,R)$
  and $\main() \redx \prog$ and $\progpre,\Gamma_1,R' : {OK}$ 
  and $\toeq{\progpre} \models \eqspre$. Then $(\progpre;\prog),\Gamma_2,R \uplus R': {OK}$.
\end{theorem}

Note that Theorem \ref{theorem-metaty} assumes the existence of
$\progpre$ with $\progpre,\Gamma_1,R' : {OK}$ and $R' \cap R =
\varnothing$ and $\toeq{\progpre} \models \eqspre$. The assumption of
some preprocessing protocol such as Beaver triple or MAC share
distributions is typical, and our approach does not require the
protocol to be defined by the programmer, but rather just (weakened)
postconditions of it that are sufficient to verify the protocol being
programmed.

Theorem \ref{theorem-metaty} also does not provide a complete recipe
for obtaining noninterference modulo output-- as per Theorem
\ref{theorem-sty}, it is also necessary to check closure of
$\Gamma_2$. There are various ways this can be accomplished for the
$\metaprot$ analysis-- for example, we evaluate $\ttt{main}$ to obtain
all program variables and then check closure of $\Gamma_2$ as
prescribed by Definition \ref{def-sty-closed}. We can also add checks
to the type system to guarantee that all messages exchanged between
clients are classified when type checking functions, which is
a simple addition.

\section{Implementation and Extended Examples}
\label{section-examples}

In Section \ref{section-smt} we gave the example of 3-party additive sharing
for any field $\fieldp{p}$, in Section \ref{section-sty} we discussed the
example of YGC input encoding, and in Sections \ref{section-metalang} and
\ref{section-metaty} we presented and discussed the example of gates in the GMW
protocol for binary circuits. In this section we present an extended example of
2-party arithmetic circuits with multiplication gates defined using Beaver
triples \cite{evans2018pragmatic}. We first define a compositional gate library
in $\metaprot$ with correctness conditions defined using the Hoare logic
developed in Section \ref{section-metalang}, along with a small example circuit
(Figure \ref{fig-beaver}). Then we extend the library with type annotations
illustrating the use of share types developed in Sections \ref{section-sty} and
\ref{section-metaty} to illustrate how confidentiality properties can be
enforced (Figure \ref{fig-beaver-types}); this annotated variant is
illustrative and is not part of the current implementation. Lastly, we extend
the base library with additional pre and postcondition annotations to enforce
integrity properties supported by BDOZ/SPDZ style MACs
\cite{SPDZ1,SPDZ2,BDOZ,10.1007/978-3-030-68869-1_3} (Figure
\ref{fig-beaver-bdoz}).

\subsection{Evaluation}

We have implemented the $\metaprot$ language and automated Hoare logic
using an Antlr parser and Scala backend with cvc5 Java bindings for
interaction with proof engines.  To evaluate performance, we have
written and verified gate libraries using the GMW protocol for binary
circuits as excerpted in Figure \ref{fig-gmwprepost} in Section
\ref{section-metalang}, and for passive and malicious secure
arithmetic circuits using Beaver Triple multiplication as presented in
Figures \ref{fig-beaver} and \ref{fig-beaver-bdoz} in this Section.
We have considered field sizes typically used to encode binary, 16
bit, 32 bit, and 64 bit integers. We ran all experiments on an Intel(R)
Core(TM) Ultra 7 255H CPU. Results are presented in Figure \ref{fig-preludetest}.

\preludetestfig

We found that verification of library components was efficient even up to large
field sizes. Since typing is compositional (see, e.g., Figure
\ref{fig-beaver-types} for an illustrative extension), verification of gate
properties relevant to \emph{security} should scale to larger circuits.

For an empirical comparison of how compositionality can reduce solve time in
GMW arithmetic circuits, we ran the experiments shown in Figure
\ref{fig-preludetest} with and without customized postconditions. Omitting
these postconditions reduces verification to the prior whole-program analysis
approach of Overture~\cite{skalka-near-ppdp24}, since the solver can no longer reuse component summaries
and must verify constraints for the fully expanded circuit directly.
This change enables a direct comparison against our prior approach.
The
scatter plot shown in Figure \ref{fig-solve-time-scatter} demonstrates that the
compositional use of customized postconditions improves verification time by
orders of magnitude, especially for 64-bit addition and subtraction, which have
the most gates of all the circuits we tested. The SMFF solver, however, fared
poorly in all these circuits. We believe the optimizations that make the
bit-vector solver efficient in these cases are unavailable to the SMFF solver,
even when it runs in $\mathbb{F}_2$. In all but the simplest cases, it did
not finish.

\solvetimescatterfig

Circuit \emph{correctness} also benefits from compositional
verification-- for example, the correctness of a GMW ``and'' gate (as
in Figure \ref{fig-gmwprepost}) only needs to be proven once, though
instantiated and integrated into the whole-program analysis needed to
verify a large circuit specification. We have verified 2, 8, and
64-bit vector zero-test GMW circuits containing 3, 15, and 127 gates
respectively, based on ``Bristol Fashion'' circuits \cite{bristol}. By
verified, we mean automated verification of the zero-test
specification expressed as a \ttt{postcondition}. We considered both
the cvc5 Bit-Vector and Satisfiability Modulo Finite Fields (SMFF)
theories as proof engines for these binary circuits.

Verification of the zero-test circuits was efficient in the Bit-Vector
theory, which verified the 64-bit zero test in GMW in less than 7
seconds. However, it took several minutes to verify even the small
2-bit zero-test circuit in SMFF, and larger circuits took too long to
consider for a practical evaluation. This could be attributed to the
relative maturity of the Bit-Vector theory as compared to SMFF.  In
any case, the SMT-based approach achieves categorically better
performance than previous brute-force whole program methods developed for binary
circuits \cite{skalka-near-ppdp24}.

\subsection{Arithmetic Circuits with Beaver Triples}
\label{section-beaver}

\begin{fpfig}[t]{2-party gate library with Beaver Triple multiplication, circuit definition, and compositional correctness conditions.}{fig-beaver}
\footnotesize
\begin{verbatimtab}
IDEALF() { (s["x"]@1 + s["y"]@2) + (s["i"]@1 * s["j"]@2) }
  
RECON(x) { (m[x]@1 + m[x]@2) }

BEAVERTRIP(x) {
  m[x++"c"]@3 == r[x++"a"]@3 * r[x++"b"]@3 AND
  m[x++"a"]@1 + m[x++"a"]@2 == r[x++"a"]@3 AND
  m[x++"b"]@1 + m[x++"b"]@2 == r[x++"b"]@3 AND
  m[x++"c"]@1 + m[x++"c"]@2 == m[x++"c"]@3 
}

genshares(x:string, i1:cid, i2:cid) { m[x]@i2 := (s[x] + -r[x])@i1; m[x]@i1 := r[x]@i1 }

_open(x:string, i1:cid, i2:cid, n:string){ p[x++"n"] := m[x]@i2; m[x++"p"]@i1 := p[x++"n"] + m[x]@i1 }

open(x:string) { _open(x,1,2,"1"); _open(x,2,1,"2") }

secopen(d:string;x:string,a:string) { sum(d, x, a); open(d) }

_sum(z:string, x:string, y:string, i:cid) { m[z]@i := (m[x] + m[y])@i }

sum(z:string, x:string, y:string){ _sum(z,x,y,1); _sum(z,x,y,2) }
postcondition: ( RECON(z) == RECON(x) + RECON(y) )

_mult(z : string, y : string, i : cid) {
    m[z]@i := ((m[z++"d"] * m[y]) + (-m[z++"e"] * m[z++"a"]) + m[z++"c"])@i
}

precondition: ( BEAVERTRIP(z) )
mult(z : string, x : string, y : string) {
  secopen(z++"d", x, z++"a"); secopen(z++"e", y, z++"b"); _mult(z,y,1); _mult(z,y,2)
}
postcondition: ( RECON(z) == RECON(x) * RECON(y) )

precondition: ( BEAVERTRIP("k") ) 
main() {
    genshares("x", 1, 2); genshares("y", 2, 1); genshares("i", 1, 2); genshares("j", 2, 1);
    sum("z","x","y"); mult("k","i","j"); sum("out","z","k");
    open("out");
    out@1 = m["out"]@1; out@2 =  m["out"]@2
    
}
postcondition: ( out@1 == IDEALF() AND out@2 == IDEALF() )
\end{verbatimtab}
\end{fpfig}

\begin{fpfig}[t]{2-party gate library with Beaver Triple multiplication extended with illustrative share type annotations.}{fig-beaver-types}
\footnotesize
\begin{verbatimtab}
genshares(x:string, i1:cid, i2:cid) {
   m[x++"s"]@i1 := (s[x] + -r[x])@i2; m[x++"s"]@i2 := r[x]@i2;
   [m[x++"s"]@i1,m[x++"s"]@i2] : shares
}

([m[x]@i1, m[x]@i2] : shares)
open(x:string) {
  _open(x,1,2,"1"); _open(x,2,1,"2");
  [p[x++"1"],p[x++"2"]] : shares;
  m[d++"p"]@2 as m[d++"p"]@1  
}

precondition: ([m[x]@i1, m[x]@i2] : shares,
	       [m[a]@i1, m[a]@i2] : shares)
secopen(d:string;x:string,a:string) {
    sum(d, x, a); open(d);
    m[d++"p"]@1 as RECON(x) + r[x++a]@3;
    <m[d++"p"]@1> : public;
    <m[d++"p"]@2> : public
}

precondition: ([m[x]@1,m[x]@2] : shares;
               [m[y]@1,m[y]@2] : shares)
sum(z:string, x:string, y:string) {
    _sum(z,x,y,1); _sum(z,x,y,2);
    [m[z]@1,m[z]@2] : shares
}
postcondition: ( RECON(z) == RECON(x) + RECON(y) )

precondition: ( [m[x]@1,m[x]@2] : shares;
                [m[y]@1,m[y]@2] : shares;
	        [m[x++"c"]@1,m[x++"c"]@2] : shares ),
	       ( BEAVERTRIP(z) )
mult(z : string, x : string, y : string) {
    secopen(z++"d", x, z++"a"); secopen(z++"e", y, z++"b"); _mult(z,y,1); _mult(z,y,2);
    [m[z]@i1,m[z]@i2] : shares
}
postcondition: ( RECON(z) == RECON(x) * RECON(y)) )
\end{verbatimtab}
\end{fpfig}

\begin{fpfig}[t]{2-party gate library with Beaver Triple multiplication extended with BDOZ MACs.}{fig-beaver-bdoz}
\footnotesize
\begin{verbatimtab}
MACOK(x, i1, i2) { m[x++"m"]@i2 == m[x++"k"]@i1 + (m["delta"]@i1 * m[x++"s"]@i2) }

MACed(x) { MACOK(x, 1, 2) AND MACOK(x, 2, 1) }

_open(x:string, i1:cid, i2:cid, n:string){
    p[x++"sn"] := m[x++"s"]@i2;
    p[x++"mn"] := m[x++"m"]@i2;
    assert(p[x++"mn"] = m[x++"k"] + (m["delta"] * p[x++"sn"]))@i1;
    m[x++"p"]@i1 := p[x++"sn"] + m[x]@i1
}

precondition: ( MACed(x) )
open(x:string) { _open(x,1,2,"1"); _open(x,2,1,"2") }

_sum(z:string, x:string, y:string, i:cid){
    m[z++"s"]@i := (m[x++"s"] + m[y++"s"])@i;
    m[z++"m"]@i := (m[x++"m"] + m[y++"m"])@i;
    m[z++"k"]@i := (m[x++"k"] + m[y++"k"])@i
}

precondition: ( MACed(x) AND MACed(y) )
sum(z:string, x:string, y:string){ _sum(z,x,y,1); _sum(z,x,y,2) }
postcondition: ( MACed(z) AND RECON(z) == RECON(x) + RECON(y) )
		 
_mult(z : string, y : string, i : cid) {
   m[z++"s"]@i := (m[z++"d"] * m[y++"s"] + -m[z++"e"] * m[z++"as"] + m[z++"cs"])@i;
   m[z++"m"]@i := (m[z++"d"] * m[y++"m"] + -m[z++"e"] * m[z++"am"] + m[z++"cm"])@i;
   m[z++"k"]@i := (m[z++"d"] * m[y++"k"] + -m[z++"e"] * m[z++"ak"] + m[z++"ck"])@i
}

precondition: ( MACed(x) AND MACed(y) AND
                MACed(z++"a") AND MACed(z++"b") AND MACed(z++"c") AND BEAVERTRIP(z) )
mult(z : string, x : string, y : string) {
    secopen(z++"d", x, z++"a"); secopen(z++"e", y, z++"b"); _mult(z,y,1); _mult(z,y,2);
}
postcondition: ( MACed(z) AND RECON(z) == RECON(x) * RECON(y) )
\end{verbatimtab}
\end{fpfig}

In Figure \ref{fig-beaver} we define a 2-party arithmetic circuit with
4 input gates, 3 internal gates, and 1 output gate. The ideal functionality
is defined as $\ttt{IDEALF}$. As in GMW or BDOZ, circuit wire values
are represented as additive secret shares between two clients we call
$\ttt{1}$ and $\ttt{2}$. As in the GMW example discussed in Section
\ref{section-metaty}, we use string identifiers $g$ to distinguish
gates, and gate output wire shares are defined as identifiers
$\mx{g}{1}$ and $\mx{g}{2}$. The function $\ttt{RECON}$
represents share reconstruction \emph{in constraints} (not the
protocols themselves). Input gates are
implemented using the $\ttt{genshares}$ function as described in
Section \ref{section-sty}. The $\ttt{open}$ function
implements a federation-wide reveal of a shared value through
share broadcast and reconstruction.

A \ttt{sum} gate can be defined non-interactively using 
homomorphic addition. The correctness \ttt{postcondition} is simple
and perhaps unnecessary to accelerate circuit verification, but is
appealing declaratively.  The multiplication gate \ttt{mult} is more
complex, where we use the common approach of \emph{Beaver Triples}
\cite{evans2018pragmatic}. Recall that Beaver triples are values
$a,b,c$ with $a$ and $b$ chosen randomly and $c = a * b$, unique per
multiplication gate, that are secret-shared with clients during a
pre-processing phase by an assumed oracular client $\ttt{3}$.
The uniqueness requirement is necessary because
$a$ and $b$ are random values that serve as OTP masks for secure
openings, implemented as \ttt{secopen}. In our encoding we follow the
convention that each gate output identifier distinguishes the Beaver
triple, so for example the share of the $a$ value for a gate
$\ttt{"g"}$ is identified by $\ttt{"ga"}$, etc. The correctness
property of the multiplication gate, expressed as a
\ttt{postcondition}, relies on the $\ttt{BEAVERTRIP}$
precondition. This compositional correctness property can then be used
in circuit construction, without the need for ``re-proof'' in our
Hoare logic inference.

The circuit itself is constructed in \ttt{main} using these
gates. This is a small example but arbitrarily large circuits can be
constructed with this library. The precondition of \ttt{main} requires
distinct Beaver triples for each multiplication gate in the circuit
(just one in this example). The postcondition allows verification of
circuit correctness-- that the protocol output is equal to the ideal
functionality $\ttt{IDEALF}$.

\subsection{Share Type Annotations (Illustrative).}

Figure \ref{fig-beaver-types} presents an illustrative extension of our
previous gate definitions with share type annotations to enforce
confidentiality. This annotated variant is not part of the current
implementation. The core idea is that wire values must be \ttt{share}s. The
\ttt{genshares} function includes an annotation specifying that all input gates
establish the $\ttt{shares}$ invariant for wire values. In the \ttt{sum}
function, the \ttt{share} type preconditions for input wires are specified. The
output wire shares can be declared as \ttt{share}s without any coercion needed,
since given preconditions the typing judgment ($\Vdash$) can check the use of
homomorphic addition.

We add coercions and typing assertions to the $\ttt{secopen}$ function
to specify properties of secure openings. After noninteractively
summing the secret-shared wire value $\ttt{RECON(x)}$ and Beaver triple value
$\ttt{r[x++a]@3}$, the federation can publicly open the result since it is
masked by the random value. This is expressed using the coercion
$\ttt{as RECON(x) +  r[x++a]@3}$. The coercion then allows the
securely opened value to be declared $\ttt{public}$.

In the \ttt{mult} function, each client $\ttt{i}$ symmetrically
computes:
$$
\ttt{((m[z++"d"] * m[y]) + (-m[z++"e"] * m[z++"a"]) + m[z++"c"])@i}
$$
Since $\ttt{m[z++"d"]@i\,:\,public}\ $ and $\ \ttt{m[z++"e"]@i\,:\,public}$
are both established by $\ttt{secopen}$, and
$\ttt{m[y]@i}$, $\ttt{m[z++"a"]@i}$, and $\ttt{m[z++"c"]@i}$ are
all declared to be shares in preconditions, we have that each
of the shared computations $\ttt{(m[z++"d"] * m[y])@i}$ and $\ttt{(-m[z++"e"] * m[z++"a"])@i}$
are shares, so the gate outputs, which are then sums of shares, are shares, and this will
be successfully checked by syntax-directed $\Vdash$ type checking that
results from the declaration $\ttt{[m[z]@i1,m[z]@i2]\,:\,shares}$. 

\subsubsection{Malicious-Secure Arithmetic Circuits}

Our Definition \ref{def-integrity} establishes a general semantics of
integrity, while the $\minicat$ integrity type system presented in
Section \ref{section-ipj} focuses on the use of BDOZ/SPDZ-style
MACs. A main concern here is the enforcement of MAC checking which is
an algebraic property. Here we show how pre and postconditions can be
used to enforce these properties in our arithmetic circuit library, in
Figure \ref{fig-beaver-bdoz}. Note that this code is an extension of the
library presented in Figure \ref{fig-beaver} and does not incorporate
the type annotations in Figure \ref{fig-beaver-types}, but this is
just for clarity of the presentation and these features could be combined. 

In this code we use a convention where shared values $\ttt{x}$ are
represented as a 3-tuple of values, MAC, and keys, where the value and
MAC of client \ttt{1} are $\ttt{m[x++"s"]@1}$ and $\ttt{m[x++"m"]@1}$
respectively, and the corresponding key $\ttt{m[x++"k"]@2}$ is held by
client \ttt{2}. The correctness condition of this 3-tuple is
$\ttt{MACOK}(x,2,1)$. The same is true symmetrically, so the
$\ttt{MACed}$ predicate establishes both sides of the correctness
condition.  In general, this predicate is used in gate preconditions
to assume that the MAC conditions hold for input shares, and in
postconditions to ensure that output shares are $\ttt{MACed}$. The
homomorphic properties that preserve $\ttt{MACed}$ness, especially in
\ttt{mult}, are perhaps even more complex than the algebraic
correctness of gates, so it is beneficial to verify them compositional
especially for large circuits.

Our code is written so that share MACs are checked for when they are
\ttt{open}ed by both clients, specified by the appropriate
$\ttt{assert}$ included in the protocol. This ensures malicious
security of secure openings and outputs. As is standard in BDOZ/SPDZ,
we assume that input values, MACs, and keys are distributed to clients
by a pre-processing protocol with relevant postconditions
declared in the preconditions of $\ttt{main}$ (as in Figure
\ref{fig-beaver} with respect to Beaver triples).

\section{Conclusion}
\label{section-conclusion}

In this paper we have developed the $\minicat$ language for defining
low-level MPC protocols, and the $\metaprot$ metaprogramming language
for specifying $\minicat$ protocols compositionally and at a higher
level of abstraction. We have developed a Hoare logic based on
SMT-style verification for $\minicat$, as well as confidentiality and
integrity type systems that enforce passive and malicious security in
the real/ideal model. These establish a foundation for an automated
Hoare logic and type system in $\metaprot$ that is compositional,
scalable, and generalizes to arbitrary prime fields and secret sharing
schemes.

Our automated $\metaprot$ Hoare logic and type theory leverages
Satisfiability Modulo Finite Fields for efficiency, scalability, and
generalizability to arbitrary prime fields. We developed examples of
real binary and arithmetic circuit protocols, illustrated the
applicability of our analyses to enforce confidentiality and integrity
properties, and evaluated performance of our implementation on these
examples.

\subsection{Future Work}

There are several directions for future work on the topics explored in
this paper through the $\metaprot/\minicat$ framework. One interesting
theoretical question is whether there are conditions under which
verification of a protocol in $\fieldp{2}$ generalizes to arbitrary
prime fields. It is often the case that MPC protocols are presented in
a way that generalizes over the field size, but it is not clear how to
articulate necessary conditions for generalizing from verification in
$\fieldp{2}$. This would have an obvious advantage since verification in
$\fieldp{2}$ is more efficient than verification in larger fields.

With regard to efficiency of implementation, integration of our system
with the Swanky framework \cite{swanky} would be an interesting
direction to explore. This framework could present opportunities for
accelerating SMT verification due to the optimizations especially
in their SAW backend.

Another major direction to explore would be an extension of the
$\metaprot/\minicat$ framework to non-information-theoretic, or
\emph{computational} security. This is security gained through
computaitonal hardness of, e.g., public or symmetric key cryptography.
Several MPC protocols rely on this sort of security, such as versions
of OT, committment, garbled circuits, etc. Making this extension would
require theoretical and practical extensions of $\metaprot/\minicat$
and make it applicable to a broader class of MPC protocols.

\bibliographystyle{ACM-Reference-Format}
\bibliography{logic-bibliography,secure-computation-bibliography}

@inproceedings{beaver1991efficient,
  title={Efficient multiparty protocols using circuit randomization},
  author={Beaver, Donald},
  booktitle={Annual international cryptology conference},
  pages={420--432},
  year={1991},
  organization={Springer}
}

@inproceedings{goldreich1987play,
  title={How to play ANY mental game},
  author={Goldreich, O and Micali, S and Wigderson, A},
  booktitle={Proceedings of the nineteenth annual ACM symposium on Theory of computing},
  pages={218--229},
  year={1987}
}

@article{li2023lilac,
  title={Lilac: a Modal Separation Logic for Conditional Probability},
  author={Li, John M and Ahmed, Amal and Holtzen, Steven},
  journal={Proceedings of the ACM on Programming Languages},
  volume={7},
  number={PLDI},
  pages={148--171},
  year={2023},
  publisher={ACM New York, NY, USA}
}

@article{evans2018pragmatic,
  title={A pragmatic introduction to secure multi-party computation},
  author={Evans, David and Kolesnikov, Vladimir and Rosulek, Mike and others},
  journal={Foundations and Trends{\textregistered} in Privacy and Security},
  volume={2},
  number={2-3},
  pages={70--246},
  year={2018},
  publisher={Now Publishers, Inc.}
}

@inproceedings{almeida2018enforcing,
  title={Enforcing ideal-world leakage bounds in real-world secret sharing MPC frameworks},
  author={Almeida, Jos{\'e} Bacelar and Barbosa, Manuel and Barthe, Gilles and Pacheco, Hugo and Pereira, Vitor and Portela, Bernardo},
  booktitle={2018 IEEE 31st Computer Security Foundations Symposium (CSF)},
  pages={132--146},
  year={2018},
  organization={IEEE}
}

@INPROCEEDINGS{6266151,
  author={Mitchell, John C. and Sharma, Rahul and Stefan, Deian and Zimmerman, Joe},
  booktitle={2012 IEEE 25th Computer Security Foundations Symposium}, 
  title={Information-Flow Control for Programming on Encrypted Data}, 
  year={2012},
  volume={},
  number={},
  pages={45-60},
  doi={10.1109/CSF.2012.30}}

@inproceedings{rastogi2014wysteria,
  title={Wysteria: A programming language for generic, mixed-mode multiparty computations},
  author={Rastogi, Aseem and Hammer, Matthew A and Hicks, Michael},
  booktitle={2014 IEEE Symposium on Security and Privacy},
  pages={655--670},
  year={2014},
  organization={IEEE}
}

@article{barthe2019probabilistic,
author = {Barthe, Gilles and Hsu, Justin and Liao, Kevin},
title = {A probabilistic separation logic},
year = {2019},
issue_date = {January 2020},
publisher = {Association for Computing Machinery},
address = {New York, NY, USA},
volume = {4},
number = {POPL},
url = {https://doi.org/10.1145/3371123},
doi = {10.1145/3371123},
journal = {Proc. ACM Program. Lang.},
month = {dec},
articleno = {55},
numpages = {30}
}

@article{darais2019language,
  title={A language for probabilistically oblivious computation},
  author={Darais, David and Sweet, Ian and Liu, Chang and Hicks, Michael},
  journal={Proceedings of the ACM on Programming Languages},
  volume={4},
  number={POPL},
  pages={1--31},
  year={2019},
  publisher={ACM New York, NY, USA}
}

@inproceedings {269581,
author = {Dahlia Malkhi and Noam Nisan and Benny Pinkas and Yaron Sella},
title = {{Fairplay{\textemdash}A} Secure {Two-Party} Computation System},
booktitle = {13th USENIX Security Symposium (USENIX Security 04)},
year = {2004},
address = {San Diego, CA},
url = {https://www.usenix.org/conference/13th-usenix-security-symposium/fairplay{\textemdash}-secure-two-party-computation-system},
publisher = {USENIX Association},
month = aug,
}

@Inbook{Lindell2017,
author="Lindell, Yehuda",
editor="Lindell, Yehuda",
title="How to Simulate It -- A Tutorial on the Simulation Proof Technique",
bookTitle="Tutorials on the Foundations of Cryptography: Dedicated to Oded Goldreich",
year="2017",
publisher="Springer International Publishing",
address="Cham",
pages="277--346",
isbn={978-3-319-57048-8},
doi={10.1007/978-3-319-57048-8\_6},
url={https://doi.org/10.1007/978-3-319-57048-8\_6}
}

@inproceedings{10.1145/2637113.2637119,
author = {Bogdanov, Dan and Laud, Peeter and Randmets, Jaak},
title = {Domain-Polymorphic Programming of Privacy-Preserving Applications},
year = {2014},
isbn = {9781450328623},
publisher = {Association for Computing Machinery},
address = {New York, NY, USA},
url = {https://doi.org/10.1145/2637113.2637119},
doi = {10.1145/2637113.2637119},
booktitle = {Proceedings of the Ninth Workshop on Programming Languages and Analysis for Security},
pages = {53-65},
numpages = {13},
location = {Uppsala, Sweden},
series = {PLAS'14}
}

@inproceedings{wysstar,
  author    = {Aseem Rastogi and
               Nikhil Swamy and
               Michael Hicks},
  editor    = {Flemming Nielson and
               David Sands},
  title     = {Wys*: {A} {DSL} for Verified Secure Multi-party Computations},
  booktitle = {8th International Conference on Principles of Security and Trust (POST)},
  shortbooktitle = {POST},
  series    = {Lecture Notes in Computer Science},
  volume    = {11426},
  pages     = {99--122},
  publisher = {Springer},
  year      = {2019},
  url       = {https://doi.org/10.1007/978-3-030-17138-4\_5},
  doi       = {10.1007/978-3-030-17138-4\_5},
}

@article{10.1145/3498713,
author = {Ye, Qianchuan and Delaware, Benjamin},
title = {Oblivious Algebraic Data Types},
year = {2022},
issue_date = {January 2022},
publisher = {Association for Computing Machinery},
address = {New York, NY, USA},
volume = {6},
number = {POPL},
url = {https://doi.org/10.1145/3498713},
doi = {10.1145/3498713},
journal = {Proc. ACM Program. Lang.},
month = {jan},
articleno = {51},
numpages = {29}
}

@INPROCEEDINGS {8429300,
author = {H. Haagh and A. Karbyshev and S. Oechsner and B. Spitters and P. Strub},
booktitle = {2018 IEEE 31st Computer Security Foundations Symposium (CSF)},
title = {Computer-Aided Proofs for Multiparty Computation with Active Security},
year = {2018},
volume = {},
issn = {2374-8303},
pages = {119-131},
doi = {10.1109/CSF.2018.00016},
url = {https://doi.ieeecomputersociety.org/10.1109/CSF.2018.00016},
publisher = {IEEE Computer Society},
address = {Los Alamitos, CA, USA},
month = {jul}
}

@inproceedings{li2021privacy,
  title={Privacy-preserving feature selection with secure multiparty computation},
  author={Li, Xiling and Dowsley, Rafael and De Cock, Martine},
  booktitle={International Conference on Machine Learning},
  pages={6326--6336},
  year={2021},
  organization={PMLR}
}

@article{knott2021crypten,
  title={Crypten: Secure multi-party computation meets machine learning},
  author={Knott, Brian and Venkataraman, Shobha and Hannun, Awni and Sengupta, Shubho and Ibrahim, Mark and van der Maaten, Laurens},
  journal={Advances in Neural Information Processing Systems},
  volume={34},
  pages={4961--4973},
  year={2021}
}

@incollection{koch2020privacy,
  title={Privacy-preserving analytics for data markets using MPC},
  author={Koch, Karl and Krenn, Stephan and Pellegrino, Donato and Ramacher, Sebastian},
  booktitle={IFIP International Summer School on Privacy and Identity Management},
  pages={226--246},
  year={2020},
  publisher={Springer}
}

@article{liu2020privacy,
  title={Privacy preserving distributed data mining based on secure multi-party computation},
  author={Liu, Jun and Tian, Yuan and Zhou, Yu and Xiao, Yang and Ansari, Nirwan},
  journal={Computer Communications},
  volume={153},
  pages={208--216},
  year={2020},
  publisher={Elsevier}
}

@inproceedings{lu2019honeybadgermpc,
  title={Honeybadgermpc and asynchromix: Practical asynchronous mpc and its application to anonymous communication},
  author={Lu, Donghang and Yurek, Thomas and Kulshreshtha, Samarth and Govind, Rahul and Kate, Aniket and Miller, Andrew},
  booktitle={Proceedings of the 2019 ACM SIGSAC Conference on Computer and Communications Security},
  pages={887--903},
  year={2019}
}

@article{ishai2009zero,
  title={Zero-knowledge proofs from secure multiparty computation},
  author={Ishai, Yuval and Kushilevitz, Eyal and Ostrovsky, Rafail and Sahai, Amit},
  journal={SIAM Journal on Computing},
  volume={39},
  number={3},
  pages={1121--1152},
  year={2009},
  publisher={SIAM}
}

@article{gao2022symmeproof,
  title={SymmeProof: Compact zero-knowledge argument for blockchain confidential transactions},
  author={Gao, Shang and Peng, Zhe and Tan, Feng and Zheng, Yuanqing and Xiao, Bin},
  journal={IEEE Transactions on Dependable and Secure Computing},
  year={2022},
  publisher={IEEE}
}

@article{tomaz2020preserving,
  title={Preserving privacy in mobile health systems using non-interactive zero-knowledge proof and blockchain},
  author={Tomaz, Antonio Emerson Barros and Do Nascimento, Jose Claudio and Hafid, Abdelhakim Senhaji and De Souza, Jose Neuman},
  journal={IEEE access},
  volume={8},
  pages={204441--204458},
  year={2020},
  publisher={IEEE}
}

@article{sabelfeld2009declassification,
author = {Sabelfeld, Andrei and Sands, David},
title = {Declassification: Dimensions and Principles},
year = {2009},
issue_date = {October 2009},
publisher = {IOS Press},
address = {NLD},
volume = {17},
number = {5},
issn = {0926-227X},
journal = {J. Comput. Secur.},
month = {oct},
pages = {517-548},
numpages = {32}
}

@inproceedings{10.1145/3453483.3454074,
author = {Acay, Co\c{s}ku and Recto, Rolph and Gancher, Joshua and Myers, Andrew C. and Shi, Elaine},
title = {Viaduct: An Extensible, Optimizing Compiler for Secure Distributed Programs},
year = {2021},
isbn = {9781450383912},
publisher = {Association for Computing Machinery},
address = {New York, NY, USA},
url = {https://doi.org/10.1145/3453483.3454074},
doi = {10.1145/3453483.3454074},
booktitle = {Proceedings of the 42nd ACM SIGPLAN International Conference on Programming Language Design and Implementation},
pages = {740-755},
numpages = {16},
location = {Virtual, Canada},
series = {PLDI 2021}
}

@INPROCEEDINGS {viaduct-UC,
author = {C. Acay and J. Gancher and R. Recto and A. Myers},
booktitle = {2024 IEEE 37th Computer Security Foundations Symposium (CSF)},
title = {Secure Synthesis of Distributed Cryptographic Applications},
year = {2024},
volume = {},
issn = {2374-8303},
pages = {315-330},
doi = {10.1109/CSF61375.2024.00021},
url = {https://doi.ieeecomputersociety.org/10.1109/CSF61375.2024.00021},
publisher = {IEEE Computer Society},
address = {Los Alamitos, CA, USA},
month = {jul}
}

@InProceedings{10.1007/978-3-030-68869-1_3,
author="Orsini, Emmanuela",
editor="Bajard, Jean Claude
and Topuzo{\u{g}}lu, Alev",
title="Efficient, Actively Secure {MPC} with a Dishonest Majority: A Survey",
booktitle="Arithmetic of Finite Fields",
year="2021",
publisher="Springer International Publishing",
address="Cham",
pages="42--71",
isbn="978-3-030-68869-1"
}

@InProceedings{SPDZ2,
author="Damg{\aa}rd, Ivan
and Pastro, Valerio
and Smart, Nigel
and Zakarias, Sarah",
editor="Safavi-Naini, Reihaneh
and Canetti, Ran",
title="Multiparty Computation from Somewhat Homomorphic Encryption",
booktitle="Advances in Cryptology -- CRYPTO 2012",
year="2012",
publisher="Springer Berlin Heidelberg",
address="Berlin, Heidelberg",
pages="643--662",
isbn="978-3-642-32009-5"
}

@InProceedings{SPDZ1,
author="Damg{\aa}rd, Ivan
and Orlandi, Claudio",
editor="Rabin, Tal",
title="Multiparty Computation for Dishonest Majority: From Passive to Active Security at Low Cost",
booktitle="Advances in Cryptology -- CRYPTO 2010",
year="2010",
publisher="Springer Berlin Heidelberg",
address="Berlin, Heidelberg",
pages="558--576",
isbn="978-3-642-14623-7"
}

@InProceedings{BDOZ,
author="Bendlin, Rikke
and Damg{\aa}rd, Ivan
and Orlandi, Claudio
and Zakarias, Sarah",
editor="Paterson, Kenneth G.",
title="Semi-homomorphic Encryption and Multiparty Computation",
booktitle="Advances in Cryptology -- EUROCRYPT 2011",
year="2011",
publisher="Springer Berlin Heidelberg",
address="Berlin, Heidelberg",
pages="169--188",
isbn="978-3-642-20465-4"
}

@InProceedings{SMFF,
author="Ozdemir, Alex
and Kremer, Gereon
and Tinelli, Cesare
and Barrett, Clark",
editor="Enea, Constantin
and Lal, Akash",
title="Satisfiability Modulo Finite Fields",
booktitle="Computer Aided Verification",
year="2023",
publisher="Springer Nature Switzerland",
address="Cham",
pages="163--186",
isbn="978-3-031-37703-7"
}

@inproceedings{skalka-near-ppdp24,
  title = {Language-Based Security for Low-Level {MPC}},
  author = {Christian Skalka and Joseph P. Near},
  year = {2024},
  booktitle = {Proceedings of the ACM Conference on Principles and Practice of Declarative Programming}
}

@InProceedings{10.1007/978-3-319-48869-1_5,
author="Dockins, Robert
and Foltzer, Adam
and Hendrix, Joe
and Huffman, Brian
and McNamee, Dylan
and Tomb, Aaron",
editor="Blazy, Sandrine
and Chechik, Marsha",
title="Constructing Semantic Models of Programs with the Software Analysis Workbench",
booktitle="Verified Software. Theories, Tools, and Experiments",
year="2016",
publisher="Springer International Publishing",
address="Cham",
pages="56--72",
isbn="978-3-319-48869-1"
}

@inproceedings{skalka-near-esop25,
  author = {Christian Skalka and Joseph P. Near},
  title = {{SMT}-Boosted Security Types for Low-Level {MPC}},
  booktitle = {Proceedings of the European Symposium on Programming},
  year = 2025
}

@misc{skalka-near-esop25-arxiv,
      title={SMT-Boosted Security Types for Low-Level MPC}, 
      author={Christian Skalka and Joseph P. Near},
      year={2025},
      eprint={2501.17824},
      archivePrefix={arXiv},
      primaryClass={cs.CR},
      url={https://arxiv.org/abs/2501.17824}, 
}

@Misc{bristol,
  author =	 {David Archer and Victor Arribas Abril and Steve Lu and Pieter Maene and Nele Mertens and Danilo Sijacic and Nigel Smart},
  title =	 {Bristol Fashion MPC Circuits.},
  howpublished = {\url{https://nigelsmart.github.io/MPC-Circuits/}},
  year =	 {2019}
}

@misc{swanky,
    author = {{Galois, Inc.}},
    title = {{swanky}: A suite of rust libraries for secure computation},
    howpublished = {\url{https://github.com/GaloisInc/swanky}},
    year = 2019,
}

@misc{cryptoeprint:2020/1410,
      author = {Carsten Baum and Alex J.  Malozemoff and Marc B.  Rosen and Peter Scholl},
      title = {Mac'n'Cheese: Zero-Knowledge Proofs for Boolean and Arithmetic Circuits with Nested Disjunctions},
      howpublished = {Cryptology {ePrint} Archive, Paper 2020/1410},
      year = {2020},
      url = {https://eprint.iacr.org/2020/1410}
}

\end{document}